\documentclass[a4paper,english]{lipics-v2016}

\usepackage{thm-restate}
\usepackage[dvipsnames]{xcolor}

\usepackage{algpseudocode}
                           \usepackage{algorithm}
\usepackage{booktabs} 
\usepackage{enumerate}
\usepackage{soul} 
\usepackage{xspace,amstext,amssymb,url}
\usepackage{colortbl}

\usepackage{graphicx}
\usepackage{tikz}
\usetikzlibrary{automata}
\usetikzlibrary{calc}

\renewcommand{\Bbb}{\mathbb} 
\newcommand{\nc}{\newcommand}
\nc{\cI}{{\mathcal I}}
\nc{\cM}{{\mathcal M}}
\nc{\cO}{{\mathcal O}}
\nc{\cR}{{\mathcal R}} 
\nc{\cS}{{\mathcal S}}
\nc{\nat}{\Bbb{N}}

\nc{\lab}{\text{lab}}
\nc{\true}{``true''\xspace}
\nc{\false}{``false''\xspace}

\nc{\nfa}{\text{NFA}\xspace}
\nc{\dfa}{\text{DFA}\xspace}
\nc{\product}{\text{product}\xspace}
\nc{\productautomaton}{\product}
\nc{\re}{\text{RE}\xspace}

\nc{\np}{NP\xspace} 
\nc{\nphard}{NP-hard\xspace} 
\nc{\nl}{NL\xspace} 
\nc{\ptime}{PTIME\xspace} 
\nc{\fpt}{FPT\xspace} 
\nc{\sharpp}{\#P\xspace}

\nc{\init}{0}

\nc{\enumpaths}{\textsf{EnumPaths}\xspace}
\nc{\enumshortpaths}{\textsf{EnumShortPaths}\xspace}
\nc{\enumsimpaths}{\textsf{EnumSimPaths}\xspace}
\nc{\enumspaths}{\enumsimpaths}
\nc{\enumeratesimpaths}{\textsf{EnumSimPaths}}

\nc{\kclique}{\textsf{$k$-Clique}\xspace}
\nc{\pth}{\textsf{Path}\xspace}

\nc{\nodespath}{\textsf{SimPath}\xspace}
\nc{\nodesimpath}{\nodespath} 
\nc{\edgespath}{\textsf{Trail}\xspace}
\nc{\trail}{\textsf{Trail}\xspace}
\nc{\edgesimpath}{\edgespath} 
\nc{\enumnodespaths}{\textsf{EnumSimPaths}\xspace}
\nc{\enumedgespaths}{\textsf{EnumTrails}\xspace}
\nc{\enumtrails}{\textsf{EnumTrails}\xspace}

\nc{\longdirpath}{$\nodesimpath_{\geq k}$\xspace}

\nc{\disjointpaths}{\textsf{DisjointPaths}\xspace}
\nc{\twodisjointpaths}{\textsf{TwoDisjointPaths}\xspace}
\nc{\twonodedisjointpaths}{\textsf{TwoNodeDisjointPaths}\xspace}
\nc{\twoedgedisjointpaths}{\textsf{TwoEdgeDisjointTrails}\xspace}

\nc{\colordisjointpaths}{\textsf{Two\-Color\-Disjoint\-Paths}\xspace}
\nc{\kdisjointpaths}{\textsf{Two\-Disjoint\-Paths}$_k$\xspace}
\nc{\kcolordisjointpaths}{\textsf{Two\-Color\-Disjoint\-Paths}$_k$\xspace}
\nc{\knodedisjointpaths}{\textsf{Two\-Node\-Disjoint\-Paths}$_k$\xspace}
\nc{\knodecolordisjointpaths}{\textsf{Two\-Color\-Node\-Disjoint\-Paths}$_k$\xspace}
\nc{\undirdisjointpaths}{\textsf{Undirected\-Two\-Disjoint\-Paths}\xspace}

\nc{\kedgedisjointpaths}{\textsf{Two\-Edge\-Disjoint\-Trails}$_k$\xspace}

\nc{\stpath}{\text{$s$-$t$-path}\xspace}
\nc{\stpaths}{\text{$s$-$t$-paths}\xspace}
\newcommand{\ctract}{\ensuremath{\textsf{C}_\text{tract}}\xspace}

\nc{\Split}{\text{split}\xspace}
\nc{\Line}{\text{line}\xspace}
\nc{\Head}{\text{head}\xspace}
\nc{\Tail}{\text{tail}\xspace}

\usepackage{paralist}

\theoremstyle{plain}

\newtheorem{observation}[theorem]{Observation}
\newtheorem{claim}[theorem]{Claim}

\setstcolor{red}
\newcommand{\oldnew}[3]{{\textcolor{blue}{\setlength{\fboxsep}{1pt}\fbox{\small #1}}} \st{#2} \textcolor{blue}{#3}}
    \nc{\won}[2]{\oldnew{W}{#1}{#2}}
\nc{\kon}[2]{\oldnew{K}{#1}{#2}}

\newcommand{\tina}[1]{#1}

\begin{document}

\title{Enumeration Problems for Regular Path Queries}

\author[1]{Wim Martens}
\author[1]{Tina Trautner}

\affil[1]{University of Bayreuth}

\maketitle
\begin{abstract}
  Evaluation of regular path queries (RPQs) is a central problem in graph
  databases. We investigate the corresponding enumeration problem, 
  that is, given a graph and an RPQ, enumerate all paths in the graph
  that match the RPQ. We consider several versions of this problem,
  corresponding to different semantics of RPQs that have recently been
  considered: \emph{arbitrary paths}, \emph{shortest paths},
  \emph{simple paths}, and \emph{trails}. 

  Whereas arbitrary and shortest paths can be enumerated in polynomial
  delay, the situation is much more intricate for simple paths and
  trails. For instance, already the question if a given graph contains
  a simple path or trail of a certain length has cases
  with highly non-trivial solutions and cases that are long-standing
  open problems. In this setting, we study RPQ evaluation from a
  parameterized complexity perspective. We define a class of
  \emph{simple transitive expressions} that is prominent in practice
  and for which we can prove two dichotomy-like results: one for
  simple paths and one for trails paths. We observe
  that, even though simple path semantics and trail semantics are
  intractable for RPQs in general, they are feasible for the vast
  majority of the kinds of RPQs that users use in practice.  At the
  heart of this study is a result of independent interest on the
  parameterized complexity of finding disjoint 
  paths in graphs:
  the two disjoint paths problem is W[1]-hard if parameterized by the
  length of one of the two paths.
\end{abstract}

\makeatletter{}\section{Introduction}

Regular path queries (RPQs) are a crucial feature of graph database
query languages, since they allow us to pose queries about arbitrarily
long paths in the graph. Essentially, RPQs are regular expressions
that are matched against labeled directed paths in the graph
database. Currently, the World Wide Web Consortium \cite{w3c} and the
openCypher project \cite{opencypher} are considering how RPQ
evaluation can be formally defined for the development of SPARQL 1.1
\cite{w3c-sparql11} and Neo4J Cypher \cite{cypher,CIP17},
respectively. Several popular candidates that have been considered are
\emph{arbitrary paths}, \emph{shortest paths}, \emph{simple paths},
and \emph{trails}
(cfr.~\cite[Section 4.5]{pablo}, \cite{CIP17}).\footnote{Simple paths
  and trails 
  are called \emph{no-repeated-node paths} and \emph{no-repeated-edge
    paths} in \cite[Section
   4.5]{pablo}, respectively.}

We briefly explain these semantics. Given a graph, an RPQ $r$
considers directed paths for which the labels on the edges form a
word in the language of $r$. We call such paths \emph{candidate matches}. The different semantics restrict the kind of paths that
\emph{match} the RPQ, i.e., can be returned as answers. \emph{Arbitrary paths} imposes no restriction and
returns every candidate match. \emph{Shortest paths} on the other hand, only returns the
shortest candidate matches. \emph{Simple paths}, resp., \emph{trails}, only
return candidate matches that do not have duplicate nodes, resp., edges.

Under \emph{arbitrary paths}, the number of matches may be infinite if
the graph is cyclic. This may pose a challenge for designing the query
language, even if one does not choose to return all matching
paths. Indeed, a popular alternative semantics of RPQs is to return
\emph{node pairs} $(x,y)$ such that there exists a matching path
from $x$ to $y$. If one wants to consider a bag semantics version for
node pairs,
where each $(x,y)$ is returned as often as the number of matches from $x$ to $y$, one needs to deal with the case where this
number is infinite.

Under \emph{shortest paths}, \emph{simple paths}, and \emph{trails}, the number of
matching paths is always finite, which simplifies the aforementioned
design challenge. However, these versions face other challenges. \emph{Simple paths} may present complexity
issues. Two fundamental problems  are that
\begin{itemize}
\item counting the number of simple paths between two nodes is
  \sharpp-complete \cite{valiant} and
\item deciding if there exists a simple path of even length between
  two given nodes is \np-complete \cite{lapaugh-papadimitriou}.
\end{itemize}
Indeed, the first problem implies that evaluating the RPQ $a^*$ under
bag semantics is \sharpp-complete and the second one implies that one
needs to solve an \np-complete problem to evaluate the RPQ
$(aa)^*$.\footnote{It is also known that answering the RPQ $a^*ba^*$
  under simple path semantics is at least as difficult as the
  \textsf{Two Disjoint Paths} problem \cite{mendelzon}.}
\emph{Trails} faces similar challenges as simple paths, due to the
similar no-repetition constraint. \emph{Shortest paths} does not have
these complexity issues, but it is unclear if its semantics is very
natural. For instance, under shortest paths semantics, if we ask how
many paths there are from $x$ to $y$, then this number may decrease if
a new, shorter, path is added.\footnote{Notice that each semantics
  only returns or counts the number of matching paths.} This may seem
counter-intuitive to users.

Since it seems that there is no one-size-fits-all solution, the
openCypher project team recently proposed to support several kinds of
semantics for Neo4J Cypher
\cite{CIP17}. This situation motivated us to shed more light on RPQ evaluation, focusing on the
following aspects:
\begin{compactitem}
\item We focus on returning \emph{paths} as answers and on
  \emph{enumeration} versions of evaluation. That is, we study
  problems where the task is to enumerate all matching paths, without
  duplicates. We are interested in which situations it is possible to
  answer queries in \emph{polynomial delay}, i.e., such that the time
  between consecutive answers is at most polynomial.
\item We take into account a recent study 
  that investigated the structure of about 250K RPQs in a wide range
  of SPARQL query logs \cite{BonifatiMT-corr17}. It turns out that these RPQs have a
  relatively simple structure, which is remarkable because their syntax
  is not restricted by the SPARQL recommendation.
\end{compactitem}
Our contributions are the following. 
\begin{compactenum}
\item After observing that enumeration of arbitrary or shortest
  paths that match an RPQ can be done in polynomial delay
  (Section~\ref{sec:enumerate}), we turn to enumeration for simple
  paths and trails. For
  \emph{downward-closed} languages (i.e., languages that are closed
  under taking subsequences), this is an easy consequence of
  Yen's algorithm~\cite{yen} (Section~\ref{sec:downwardclosed}).
\item We show that Bagan et al.'s dichotomy for deciding the existence
  of a simple path that matches an RPQ \cite{bagan} carries over to
  enumeration problems (Section~\ref{sec:BeyondDWC}). Furthermore, we
  show that Bagan.'s dichotomy carries over from simple paths to
  trails.  Since Bagan et al.'s dichotomy is about the \emph{data
    complexity} of RPQ evaluation, this gives us some understanding
  about the data complexity of enumeration.
\item However, our goal is to get a better understanding of the
  \emph{combined complexity} of enumerating simple paths or trails. This is a
  challenging task because it contains subproblems that are highly
  non-trivial. One such subproblem is testing if there exists a
  directed simple path of length $\log n$ between two given nodes in graph
  $G$ with $n$ nodes. This problem was shown to be in PTIME by Alon et al., using
  their color coding technique \cite{AlonYZ-jacm95}.  It is open for
  over two decades if it can be decided in
  PTIME if there is a simple path of length $\log^2 n$ \cite{AlonYZ-jacm95}.
  Notice that these two problems are special cases of RPQ evaluation under
  simple path semantics (i.e., evaluate the RPQs $a^{\log n}$ and
  $a^{\log^2 n}$ in a graph where every edge has label $a$).

  We therefore investigate RPQ evaluation from the angle of
  parameterized complexity. We introduce the class of \emph{simple
    transitive expressions (STEs)} that capture over 99\% of the RPQs
  that were found in SPARQL query logs in a recent study
  \cite{BonifatiMT-corr17}. We identify a property of STEs that we
  call \emph{cuttability} and prove that the combined parameterized
  complexity for evaluating STEs $\cR$ is in FPT if $\cR$ is cuttable
  and W[1]-hard otherwise. Examples of cuttable classes of expressions
  are $a^k a^*$ and $(a+b)^k a^*$ (for $k \in \nat$).  Examples of
  non-cuttable classes are $a^k b^*$, $a^k b a^*$, and
  $a^k(a+b)^*$. For trails, we also show a dichotomy, but here the FPT
  fragment is larger. That is, if the class $\cR$ is not cuttable,
  evaluation is still FPT if $\cR$ is \emph{almost conflict-free}.  We
  show that these dichotomies carry over to enumeration problems
  (Section~\ref{sec:param-srp}).

\item At the core of these results are two results of independent
  interest (Section~\ref{sec:ParamCompl}). The first shows that the
  \textsf{Two Disjoint Paths} problem is W[1]-hard when parameterized
  by the length of one of the two paths
  (Theorem~\ref{theorem:two-disjoint-is-hard}). The second
  is by the authors of \cite{fomin}, who showed that it can be decided
  in FPT if there is a simple path of length \emph{at least $k$}
  between two nodes in a graph
  (Theorem~\ref{theorem:longdirpathInFPT}).

\end{compactenum}

Putting everything together, we see that, although simple path and trail
semantics lead to high complexity in general, their complexity for
RPQs that have been found in SPARQL query logs is reasonable. We
discuss this in the conclusions.

\paragraph*{Related Work}
RPQs on graph databases have been
  researched in the literature since the end of the 80's
  \cite{ConsensM-pods90,CruzMW-sigmod87,Yannakakis-pods90}
and many problems have been investigated, such as optimization \cite{AbiteboulV-jcss99},
 query rewriting and query answering using views
 \cite{CalvaneseGLV-jcss02,CalvaneseGLV-pods00}, 
and containment
\cite{CalvaneseGLV-kr00,DeutschT-dbpl01,FlorescuLS-pods98}. RPQ
evaluation is therefore a fundamental problem in the field. We refer
to \cite{Barcelo-pods13}
for an excellent overview on RPQs and queries for graph databases in general.

Mendelzon and Wood \cite{mendelzon} were the first to consider simple
paths for answering regular path queries. They proved that testing if
there exist simple paths matching $a^* b a^*$ or $(aa)^*$ is
\np-complete and studied classes of graphs for which evaluation
becomes tractable. Arenas et al.~\cite{ArenasCP-www12} and Losemann
and Martens~\cite{LosemannM-tods13} studied counting problems related
to RPQs in SPARQL 1.1 (which are called \emph{property paths} in the
specification). They showed that, under the definition 
of SPARQL at
that time, query evaluation was highly complex. They made proposals on
how to amend this, which were largely taken into account by the W3C.
Extensions of RPQ-like queries with data value comparisons and
branching navigation were studied by Libkin et al.\ \cite{LibkinMV-jacm16}.

Bagan et al.~\cite{bagan} studied the \emph{data complexity} of RPQ
evaluation under simple path semantics (i.e., the regular path query
is considered to be constant). They proved that there is a trichotomy
for the evaluation problem: the data complexity of RPQ evaluation 
 is
\np-complete for languages outside a class they call \ctract,
it is NL-complete for infinite
languages in \ctract, and in AC$^0$ for finite
languages. (Since the results are on data complexity, the
representation of the languages does not matter.)

We also consider problems where the task is to enumerate paths in
graphs.
  In this
context we will use Yen's algorithm \cite{yen} which is a procedure
for enumerating simple paths in graphs. Yen's algorithm was
generalized by Lawler~\cite{lawler} and Murty~\cite{murty} to a tool
for designing general algorithms for enumeration
problems. Lawler-Murty's procedure has been used for solving
enumeration problems in databases in various contexts
\cite{golenberg2011,KalinskyEK-edbt17,KimelfeldS-icdt13}.

\makeatletter{}\section{Preliminaries}\label{sec:preliminaries}\label{sec:definitions}

By $\Sigma$ we always denote an \emph{alphabet}, that is, a finite
set. A \emph{$(\Sigma$-$)$symbol} is an element of $\Sigma$. A
\emph{word} (over $\Sigma$) is a finite sequence $w = a_1\cdots a_n$
of $\Sigma$-symbols. The \emph{length} of $w$, denoted by $|w|$, is
its number of symbols $n$. We denote the empty word by $\varepsilon$.
We denote the concatenation of words $w_1$ and $w_2$ as $w_1 \cdot
w_2$ or simply as $w_1w_2$. 
We assume familiarity with regular expressions and finite
automata. The regular expressions (\re) we use in this paper are
defined as follows: $\emptyset$, $\varepsilon$ and every $\Sigma$-symbol is a
regular expression; and when $r$ and $s$ are regular expressions, then
$(rs)$, $(r + s)$, $(r?)$, $(r^*)$, and $(r^+)$ are also regular
expressions. From now on, we use the usual precedence rules to omit
braces.  The \emph{size} $|r|$ of a regular expression is the number
of occurrences of $\Sigma$-symbols in $r$. For example, $|aba^*| =
3$. We define the \emph{language $L(r)$} of $r$ as usual. Since it is
easy to test if $L(r) = \emptyset$ for a given expression $r$, we
assume in this paper that $L(r) \neq \emptyset$ for all expressions,
unless mentioned otherwise.
For $n \in
\nat$, we use $r^n$ 
to abbreviate the $n$-fold concatenation $r\cdots r$ of $r$. We
abbreviate $(r?)^n$ by $r^{\leq n}$.

A \emph{non-deterministic finite automaton (\nfa)} $N$
over $\Sigma$ is a tuple $(Q,\Sigma, \Delta, Q_I ,Q_F)$, where $Q$ is a
finite set of states, $\Sigma$ is a finite alphabet, $\Delta: Q \times
\Sigma \times Q$ is the transition relation, $Q_I \subseteq Q$ is the
set of initial
states, and $Q_F$ is the set of final states. By $\delta^*(w)$ we
denote the set of states reachable by $N$ after reading $w$, that is,
$\delta^*(\varepsilon) = Q_I$ and, for every word $w$ and symbol $a$,
we define $\delta^*(wa) = \{\delta(q,a) \mid q
\in \delta^*(w)\}$.
The \emph{size} of an
\nfa is $|Q|$, i.e., its number of states. 
We define the \emph{language} $L(N)$ of $N$ as usual.

\subsection{Graph Databases, Paths, and Trails}
We use edge-labeled directed graphs as abstractions for graph
databases. A graph $G$ (with labels in $\Sigma$) will be denoted as $G
= (V,E)$, where $V$ is the finite set of \emph{nodes} of $G$ and $E \subseteq
V \times \Sigma \times V$ is the set of \emph{edges}. We say that edge
$e = (u,a,v)$ goes \emph{from $u$ to node $v$} and \emph{has the label
  $a$}. Sometimes we write an edge as $(u,v) \in V\times V$ if the
label does not matter. In this paper, we assume that graphs are
directed, unless mentioned otherwise. The \emph{size} of a graph $G$,
denoted by $|G|$ is $|V|+|E|$.

We assume familiarity with basic terminology on graphs. A \emph{path}
from node $u$ to node $v$ in $G$ is a sequence $p =
(v_0,a_1,v_1)(v_1,a_2,v_2) \cdots (v_{n-1},a_n,v_n)$ of edges in $G$
such that $u = v_0$ and $v = v_n$. For $0 \leq i \leq n$, we denote by
$p[i,i]$ (or $p[i]$) the node $v_i$ and, for $0 \leq i < j \leq n$, we denote by
$p[i,j]$ the subpath $(v_i,a_{i+1},v_{i+1}) \ldots
(v_{j-1},a_{j},v_{j})$. A path $p$ is a \emph{simple path} if it has
no repeated nodes, that is, all nodes $v_0,
\ldots, v_n$ are pairwise different. It is a \emph{trail} if it has no
repeated edges, that is, all triples $(v_i,a_{i+1},v_{i+1})$ are
pairwise different.
    The \emph{length} of
$p$, denoted $|p|$, is the number $n$ of edges in $p$.  By definition
of paths, we consider two paths to be different if they are different
sequences of edges. In particular, two paths going through the same
nodes in the same order, but using different edge labels are
different.

The set of \emph{nodes of path $p$} is $V(p) = \{v_0,\ldots,v_n\}$.
The \emph{word of $p$} is $a_1\cdots a_n$ and is denoted by
$\lab^G(p)$. We omit $G$ if it is clear from the context. Path $p$
\emph{matches} a regular expression $r$ (resp., \nfa $N$) if $\lab(p)
\in L(r)$ (resp., $\lab(p)\in L(N)$).  The \emph{concatenation} of
paths $p_1 = (v_0,a_1,v_1) \cdots (v_{n-1}, a_n, v_n)$ and $p_2 =
(v_n, a_{n+1}, v_{n+1}) \cdots (v_{n+m-1}, a_{n+m}, v_{n+m})$ is
simply the concatenation $p_1 p_2$ of the two sequences. 

We will often consider a graph $G = (V,E)$ together with a
\emph{source node} $s$ and a \emph{target node} $t$, for example, when
considering paths from $s$ to $t$. We denote such a graph with source
$s$ and target $t$ as $(G,s,t)$ and define their size $|(G,s,t)|$ as $|G|$.

The \emph{\product} of graph $(G,s,t)$ and \nfa $N =
(Q,\Sigma,\Delta,Q_I,Q_F)$ is a graph $(G,s,t)\times N = (V',E')$ with $V' = (V \times Q)$ and $E' = \{((u_1,q_1), a, (u_2,q_2)) \mid (u_1,a,u_2) \in E$ and $(q_1,a,q_2)
\in \Delta\}$.

\subsection{Decision and Enumeration Problems}

An \emph{enumeration problem} \textsf{P} is a set of pairs $(i,O)$ where $i$
is an \emph{input} and $O$ is a finite or countably infinite set of
\emph{outputs for $i$}, denoted by $\textsf{P}(i)$. Terminologically, we say
that the task is to \emph{enumerate $O$, given $i$.}

We consider the following problems, where $G$ is always a graph, $s$
and $t$ are nodes in $G$, and $r$ is a regular expression (also called
\emph{regular path query (RPQ)}).
\begin{itemize}
\item \pth: Given $(G,s,t)$ and $r$, is there a path from $s$ to
  $t$ that matches $r$?
\item \nodespath:  Given $(G,s,t)$ and $r$,  is
   there a simple path from $s$ to $t$ that matches $r$?
 \item \edgespath: Given $(G,s,t)$ and $r$,  is
 there a trail from $s$ to $t$ that matches $r$?
 \item \enumpaths: Given $(G,s,t)$ and $r$,  enumerate the paths in $G$ from $s$ to $t$
   that match $r$.
 \item \enumshortpaths: Given $(G,s,t)$ and $r$,  enumerate the shortest paths in $G$ from $s$ to $t$
   that match $r$.
 \item \enumnodespaths: Given $(G,s,t)$ and $r$,  enumerate the simple paths in $G$ from $s$ to $t$
   that match $r$.
   \item\enumedgespaths: Given $(G,s,t)$ and $r$, enumerate the trails in $G$ from $s$ to $t$
   that match $r$.
\end{itemize}
An \emph{enumeration algorithm} for \textsf{P} is an algorithm that, given
input $i$, writes a sequence of answers to the output such that every
answer in $\textsf{P}(i)$ is written precisely once. If $A$ is an enumeration
algorithm for enumeration problem \textsf{P}, we say that $A$ runs in
\emph{polynomial delay}, if the time before writing the
first answer and the time between writing every two consecutive
answers is polynomial in $|i|$.

For a class $\cR$ of regular expressions, we denote by $\pth(\cR)$ the
problem \pth where we always assume that $r \in \cR$. If $\cR$ consists of a single expression $r$,
we simplify the notation to $\pth(r)$. We use the same convention for
all other decision- and enumeration problems.
We assume familiarity with the notions \emph{combined}- and
\emph{data} complexity. In our decision problems, $(G,s,t)$ is the
data and $r$ is the query.

\subsection{Reducing Between Trails and Simple
  Paths}\label{sec:trails-simplepaths}
Lapaugh and Rivest~\cite{LapaughR-jcss80} showed that there is a
strong correspondence between trail and simple path problems that we
will use extensively and therefore revisit here. Unfortunately,
Lapaugh and Rivest's Lemmas 1 and 2 do not precisely capture what we
need, so we have to be a bit more precise.

The following construction is from \cite[Proof of Lemma 1]{LapaughR-jcss80}.
Let $(G,s_1,t_1,\ldots,s_k,t_k)$ be a graph $G$ together with nodes
$s_1,t_1,\ldots,s_k,t_k$. We denote by $\Split(G,s_1,t_1,\ldots,s_k,t_k)$ the
tuple $(G',s'_1,t'_1,\ldots,s'_k,t'_k)$ obtained as follows. The graph
$G'$ is obtained from $G$ by replacing each node $v$ by two nodes
$\Head(v)$ and $\Tail(v)$. A directed edge is added from $\Head(v)$ to
$\Tail(v)$. All incoming edges of $v$ become incoming edges of
$\Head(v)$ and all outgoing edges of $v$ become outgoing edges of
$\Tail(v)$. For every $s_i$ and $t_i$, we define $s'_i = \Head(s_i)$
and $t'_i = \Tail(t_i)$. 

\begin{restatable}{lemma}{splitgraphlemma}\label{lem:splitgraph}
  Let $(G',s'_1,t'_1,\ldots,s'_k,t'_k) =
  \Split(G,s_1,t_1,\ldots,s_k,t_k)$. Then there
  exists pairwise node disjoint simple paths of length $k_i$ from
  $s_i$ to $t_i$ in $G$ iff there exist pairwise edge disjoint trails
  of length $2k_i + 1$ from $s'_i$ to $t'_i$ in $G'$.
\end{restatable}

We denote by $\Line(G,s_1,t_1,\ldots,s_k,t_k)$ a variation on the
\emph{line graph of $G$}. More precisely,
$\Line(G,s_1,t_1,\ldots,s_k,t_k)$ is the tuple
$(G',s_1,t_1,\ldots,s_k,t_k)$ obtained as follows. Let $G =
(V,E)$. The nodes of
$G'$ are $E \cup \{s_1,t_1,\ldots,s_k,t_k\}$. The edges of $G'$ are
the disjoint union of
\begin{itemize}
\item $\{((u,v),(v,w)) \mid (u,v)$ and $(v,w) \in E\}$,
\item $\{(s_i,(s_i,v)) \mid i = 1,\ldots,k$ and $(s_i,v) \in E\}$, and
\item $\{((v,t_i)) \mid i = 1,\ldots,k$ and $(v,t_i) \in E\}$.
\end{itemize}
It is well known that the line graph of $G$ is useful for reducing
trail problems to simple path problems \cite[Proof of Lemma 2]{LapaughR-jcss80}.
\begin{restatable}{lemma}{linegraphlemma}\label{lem:linegraph}
  Let $(G',s'_1,t'_1,\ldots,s'_k,t'_k) =
  \Line(G,s_1,t_1,\ldots,s_k,t_k)$. Then there exist pairwise edge
  disjoint trails of length $k_i$ from $s_i$ to $t_i$ in $G$ iff there
  exist pairwise node disjoint simple paths of length $k_i+1$ from
  $s'_i$ to $t'_i$ in $G'$.
\end{restatable}

\makeatletter{}\section{Enumerating All Regular Paths and Shortest Regular Paths}\label{sec:enumerate}
The following result is due to Ackerman and Shallit.
\begin{theorem}[Theorem 3 in \cite{Ackerman-TCS09}] \label{theo:ackermanAndShallit}
  Given an NFA $N$, enumerating the words in $L(N)$ 
    can be done in polynomial delay.
\end{theorem}
This result generalizes a result of M\"akinen
\cite{Makinen-actaCyber97}, who proved that the words in $L(N)$ can be
enumerated in polynomial delay if $N$ is deterministic. Ackermann and
Shallit genereralized his algorithm and proved that, for a given length
$n$ (which they call \emph{cross-section}), the lexicographically
smallest word in $L(N)$ can be found in time $O(|Q|^2n^2)$
(\cite{Ackerman-TCS09}, Theorem 1). They then prove that the set of
all words of length $n$ can be computed in time $O(|Q|^2n^2 + |\Sigma|
|Q|^2 x)$, where $x$ is the sum of lengths of outputted words
(\cite{Ackerman-TCS09}, Theorem 2). A closer inspection of their
algorithm actually shows that it has delay $O(|\Sigma| |Q|^2 |w|)$
where $|w|$ is the size of the next output.
In fact, Ackermann and Shallit prove that the words in $L(N)$ can be
enumerated in \emph{radix order}.\footnote{That is, $w_1 < w_2$ in
  radix order if $|w_1| < |w_2|$ or $|w_1| = |w_2|$ and $w_1$ is
  lexicographically smaller than $w_2$.}

It is easy to extend the algorithm of Ackerman and Shallit to solve
\enumpaths in polynomial delay as follows. We construct an NFA $N_r$
for $r$ and take the product with $(V,E,s,t)$. The product automaton
therefore has states $(q,u)$ where $q$ is a state from $N_r$ and $u$ a
node from $G$. In the resulting automaton, replace every transition
$[(q_1,u_1), a, (q_2,u_2)]$ 
with 
$[(q_1,u_1), (u_1,a,u_2), (q_2,u_2)]$.  Enumerating the words from the
resulting automaton in radix order corresponds to enumerating the
paths from $s$ to $t$ that match $r$ in radix order in polynomial
delay. We therefore have the following corollary.
\begin{corollary}
  \enumpaths and \enumshortpaths can be solved in polynomial delay.
\end{corollary}

For completeness, we note that counting the number of paths from $s$
to $t$ that match a given regular expression $r$ is \sharpp-complete
in general, even if $G$ is acyclic, see \cite[Theorem 4.8(1)]{LosemannM-tods13}
and \cite[Theorem 6.1]{ArenasCP-www12}.\footnote{Arenas et al.~\cite{ArenasCP-www12}
  actually prove that the problem is \textsc{spanL}-complete. Although
  it is not known if \textsc{spanL} $=$ \sharpp, they are equal under
  Cook reductions.} The
same holds for counting the number of shortest paths, since all paths
in the proof of \cite[Theorem 4.8(1)]{LosemannM-tods13} have equal length.

\section{Enumerating Simple Regular Paths}\label{sec:enum-yen}

We now turn to the question of enumerating simple paths with
polynomial delay. A starting point is Yen's algorithm \cite{yen} for
finding simple paths from a source $s$ to target $t$. Yen's algorithm
usually takes another parameter $K$ and returns the $K$ shortest simple paths, but
we present a version here for enumerating all simple paths.

\begin{algorithm}[t]
\algrenewcommand\algorithmicindent{10pt}\caption{Yen's algorithm}
\begin{algorithmic}[1]
\Require Graph $G=(V,E)$, nodes $s,t$ 
\Ensure \begin{tabular}[t]{l}
  The simple paths from $s$ to $t$ in $G$   \end{tabular}
\State $A$ $\gets$ $\emptyset$ \Comment{$A$ is the set of paths already written to output}
\State $B$ $\gets$ $\emptyset$ \Comment{$B$ is a set of paths from $s$ to $t$}
\State $p$ $\gets$ a shortest path from $s$ to $t$ in $G$ \label{alg:yen:3}
\While{$p \neq$ null} \Comment{As long as we find a path $p$}
\State \textbf{output} $p$
\State Add $p$ to $A$ \label{alg:yen:6}
\For{$i = 1$ to $|p|$ }   
\State $G' \gets (V',E')$, where $V' = V\setminus V(p[0,i-1])$ and $E' = E \cap (V' \times V')$
    \label{alg:yen:8}
  \For{every path $p_1$ in $A$ with $p_1[0,i-1] = p[0,i-1]$}\label{alg:yen:9}
    \State Delete the edge $p_1[i-1,i]$ in $G'$   \EndFor
  \Comment{$G'$ now no longer has paths already in $A$ }
  \State Find a shortest path $p_2$ from $p[i,i]$ to $t$ in $G'$ 
  		\label{alg:yen:12}  \State Add $p[0,i]\cdot p_2$ to $B$  \label{alg:yen:13}
\EndFor
\State $p \gets$ a shortest path in $B$ \Comment{$p \gets$ null if $B
  = \emptyset$} \label{alg:yen:15}
\State Remove $p$ from $B$
\EndWhile
\end{algorithmic}
\label{alg:yen}
\end{algorithm}

\begin{theorem}[Implicit in \cite{yen}]
  Given a graph $G$ and nodes $s$, $t$, Algorithm~\ref{alg:yen}
  enumerates all simple paths from $s$ to $t$ in polynomial delay.
\end{theorem}
\begin{proof}[Proof sketch]
  The original algorithm of Yen \cite{yen} finds, for a given $G$,
  $s$, $t$, and $K \in \nat$ the $K$ shortest simple paths from $s$ to
  $t$ in $G$. Its only difference to Algorithm~\ref{alg:yen} is that
  it stops when $K$ paths are returned.

  Yen does not prove that the algorithm has polynomial delay, but
  instead shows that the delay is $O(KN + N^3)$, where $N$ is the
  number of nodes in $G$.\footnote{In \cite{yen}, Section 5, he notes
    that computing path number $k$ in the output costs, in his
    terminology, $O(KN)$ time in Step I(a) and $O(N^3)$ in Step
    I(b).} Unfortunately,
  $K$ can be exponential in $|G|$ 
  in general. However, the reason why
  the algorithm has $K$ in the complexity is line~\ref{alg:yen:9},
  which iterates over all paths in $A$. If we do not store $A$ as a
  linked list as in \cite{yen} but as a prefix tree of paths instead,
  the algorithm only needs $O(N^2)$ steps to complete the entire
  for-loop on line~\ref{alg:yen:9} (without any optimizations). We
  therefore obtain delay $O(N^3)$ from Yen's analysis.
\end{proof}

\subsection{Downward Closed Languages}\label{sec:downwardclosed}
Yen's algorithm immediately shows that \enumnodespaths can be solved in
polynomial delay for languages that are closed under taking
subsequences. Formally, we say that a language $L$ is \emph{downward
  closed} if, for every word $w = a_1 \cdots a_n \in L$ and every
sequence $0 < i_1 < \cdots < i_k < n+1$, we have that $a_{i_1}
\cdots a_{i_k} \in L$. A regular expression is downward closed if it
defines a downward closed language.

\begin{restatable}{proposition}{propprefixclosed} \label{prop:prefixclosed}
 $\enumnodespaths(\cR)$ is in polynomial delay for the class $\cR$ of
 downward closed regular expressions, even when the paths need to be
 output in radix order.
\end{restatable}
\begin{proof}[Proof sketch.]
  Assume that $(G,s,t)$ and $r$ is an input for \enumnodespaths such that
  $L(r)$ is downward closed. Let $N = (Q,\Sigma,\delta, Q_I, Q_F)$ be
  an NFA for $r$.  We change Algorithm~\ref{alg:yen} as follows:
  \begin{itemize}
  \item In line~\ref{alg:yen:3}, instead of finding a shortest path
    $p$ in $G$, we first find a shortest path $p$ in $(G,s,t) \times
    N$. We then replace every node of the form $(u,q) \in V\times Q$
    in $p$ by $u$.
  \item In line~\ref{alg:yen:12} we need to find a shortest path in a
    product between $(G',p[i,i],t)$ and $N$. More precisely, let $J =
    \delta^*(\lab(p[0,i]))$ and denote by $N_J$ the NFA with initial
    state set $J$, that is, $(Q,\Sigma,\delta,J,Q_F)$. Then, in line
    \ref{alg:yen:12} we first find a shortest path $p_2$ from any node
    in $\{(p[i,i],q_i) \mid q_i \in \delta^*(\text{lab}(p[0,i]))\}$ to any
    node in $\{(t,q_F) \mid q_F \in Q_F\}$ in $(G',p[i,i],t) \times
    N_J$. We then replace every node of the form $(u,q) \in V\times Q$
    in $p_2$ by $u$.
  \end{itemize}
  We prove in Appendix~\ref{app:enum-yen} that the adapted algorithm is correct.
    By using Ackermann and Shallit's algorithm \cite{Ackerman-TCS09} from
  Theorem~\ref{theo:ackermanAndShallit}, we can even find a smallest
  path in $(G',p[i,i],t) \times N_J$ in radix order. Therefore, we can
  even enumerate the paths in $\enumnodespaths(\cR)$ in polynomial delay
  in radix order.
\end{proof}

Now we prove that upper bounds transfer from simple path problems to
trail problems. This is not immediate from Lemma~\ref{lem:linegraph},
since it only deals with unlabeled graphs.
Furthermore, there cannot be a polynomial time reduction from $\nodespath$ to $edgespath(a^kb^*)$ since $\edgespath(a^kb^*)$ is in FPT while $\nodespath(a^kb^*)$ is W[1]-hard. We will prove this later in Theorem~\ref{theo:dichotomy} and Theorem~\ref{theo:edgedichotomy}
\begin{restatable}{lemma}{trailtopathfptreduction}\label{lemma:edgeToNodes}
  Let $r$ be a regular expression and $(G,s,t)$ a graph. Then there
  exist graphs $(H_1,s_1,t_1),\ldots,(H_n,s_n,t_n)$ with $n \leq |G|$ such that there
  exists a trail from $s$ to $t$ in $G$ that matches $r$ if and only
  if there exists an $i$ such that there exists a simple path from
  $s_i$ to $t_i$ in $H_i$ that matches $r$.  Furthermore, each $H_i$
  is computable in polynomial time.
		\end{restatable}

Using this Lemma, we can immediately show that the upper bound from Lemma~\ref{prop:prefixclosed} also holds for edge-disjoint problems.
\begin{corollary}
	$\enumedgespaths(\cR)$ is in polynomial delay for the class $\cR$ of
	downward closed regular expressions, even when the paths need to be
	output in radix order.
\end{corollary}
\begin{proof} Given $r \in \cR$ and a graph $G$. We use Lemma~\ref{lemma:edgeToNodes} to construct the graphs $(H_1,s_1,t_1),\allowbreak \ldots,\allowbreak (H_n,s_n,t_n)$. The algorithm in Lemma~\ref{prop:prefixclosed} allows us to enumerate all simple paths from $s_i$ to $t_i$ in $H_i$ in radix order. Therefore, we use $n$ parallel instances of this algorithm to enumerate, for all $i$, all simple paths from $s_i$ to $t_i$ in $H_i$ in radix order. Since each simple path in each $H_i$ corresponds to a trail in $G$, we can also output the corresponding paths in polynomial delay with radix order.
\end{proof}

\subsection{Beyond Downward Closed Languages, Data Complexity} \label{sec:BeyondDWC}

Once we go beyond downward-closed languages, simple paths or trails can not
always be enumerated in polynomial delay (if P $\neq$ NP). For instance, the
problems $\nodespath(a^*ba^*)$ and $\nodespath((aa)^*)$ are well known to be
\np-complete \cite{mendelzon} and it is easy to see that the
corresponding problems for trails are \np-complete too.

Bagan et al.\ \cite{bagan} studied the data complexity of \nodespath and
discovered a dichotomy w.r.t.\ a class $\ctract$ of regular
languages.\footnote{They actually proved that there is a trichotomy:
  the third characterization is that \nodespath is in AC$^0$ if $L(r)$ is
  finite.}
More precisely, although $\nodespath(r)$ can be \np-complete in general,
it is in \ptime if $L(r) \in \ctract$ and
\np-complete otherwise \cite[Theorem 2]{bagan}.  Here, $\ctract$ is defined as follows.
\begin{definition}[Similar to \cite{bagan}, Theorem 4] \label{def:loopabbrev}
  For $i \in \nat$, we say that a regular language \emph{can be
    $i$-loop abbreviated} if, for all $w_\ell, w,w_r \in \Sigma^*,
  w_1, w_2 \in \Sigma^+$, we have that, if $w_\ell w_1^i w w_2^i w_r
  \in L$, then $w_\ell w_1^i w_2^i w_r \in L$. We define
  $\ctract$ as the set of regular languages $L$ such that there
  exists an $i \in \nat$ for which $L$ can be $i$-loop abbreviated.
    \end{definition}

We show that Bagan et al.'s classification also leads to a dichotomy
w.r.t.\  polynomial delay enumeration in terms of data complexity.
\begin{restatable}{theorem}{BaganPolyDelay} \label{theorem:BaganPolyDelay}
  In terms of data complexity, 
  \begin{enumerate}[(a)]
  \item $\enumnodespaths(r)$ can be solved in polynomial delay if $L(r) \in \ctract$ and
  \item $\nodespath(r)$ is \np-complete otherwise.
  \end{enumerate}
\end{restatable}
\begin{proof}[Proof sketch]
  Part (b) is immediate from \cite[Theorem 1]{bagan}. For (a), our
  plan is to use Bagan et al.'s algorithm for simple paths (which we
  call BBG algorithm) as a subroutine in Yen's algorithm. We call BBG
  in lines \ref{alg:yen:3} and \ref{alg:yen:12}, so that the algorithm
  receives
  \begin{enumerate}[(i)]
  \item a simple path from $s$ to $t$ that matches $r$ in line
    \ref{alg:yen:3} and
  \item a simple path $p_2$ from $p[i,i]$ to $t$ such that
    $p[0,i]\cdot p_2$ matches $r$ in line \ref{alg:yen:12},
  \end{enumerate}
  respectively.   Change (i) to Yen's algorithm is trivial. Change (ii) can be done by
  calling BBG with $G'$ for the language of the automaton $N_J$ in the
  proof of Proposition~\ref{prop:prefixclosed}. We show that the
  adapted algorithm is correct in Appendix~\ref{app:enum-yen}.
\end{proof}
As we argue in Appendix~\ref{app:enum-yen}, the algorithm for Theorem~\ref{theorem:BaganPolyDelay}(a) can even be
adapted to output paths in increasing length (even radix order).

In fact, Bagan et al.'s dichotomy can also be extended to
$\edgespath(r)$. 
We note that the NP hardness of $\nodespath(r)$ does not carry over to
$\edgespath(r)$ with the reductions introduced in
Lemmas~\ref{lem:splitgraph} or \ref{lemma:edgeToNodes}. Lemma \ref{lem:splitgraph} only applies to
unlabeled graphs and, when adjusting it to labeled graphs, one would
only obtain hardness for a very restricted class of expressions
instead of all expressions in
\ctract. Lemma~\ref{lemma:edgeToNodes} on the other hand only allows to transfer
the \emph{upper bound}. We therefore need to revisit some of Bagan et al's methods.

\begin{restatable}{theorem}{bagantrails} \label{theo:baganForEdgeDisjoint}
	Let $r$ be a regular expression.
	\begin{enumerate}[(a)]
	\item If $L(r)$ belongs to \ctract, $\edgespath(r)$ is in \ptime.
	\item Otherwise, $\edgespath(r)$ is NP-complete.
	\end{enumerate}
\end{restatable}
\begin{proof}
Part (a) follows directly from Lemma~\ref{lemma:edgeToNodes} 
and the upper bound of Bagan et
al.~\cite[Theorem 2]{bagan}.
It remains to show (b). The upper bound again follows from
Lemma~\ref{lemma:edgeToNodes}. The hardness is similar to \cite[Lemma
2]{bagan}. We prove it in the Appendix.
\end{proof}

\subsection{Beyond Downward Closed Languages, Combined Complexity}

Unfortunately, Bagan et al.'s classification does not go through when
we consider combined complexity. Indeed, if $G$ is a graph with $n$
nodes and only $a$-labeled edges, then asking if there is a simple path that matches the
expression $a^n$ (which is finite and therefore in $\ctract$) is
the \np-complete \textsc{Hamilton Path} problem.

On the other hand, Alon et al.~\cite{AlonYZ-jacm95} proved that \nodespath for
graphs with $n$ nodes is in \ptime for the language $a^{\log n}$,
which is also in \ctract. It is open since 1995 whether \nodespath is in
\ptime for $a^{\log^2 n}$ \cite{AlonYZ-jacm95}. Recently, Bj\"orklund et
al.\ \cite{bjoerklund} showed that, under the Exponential Time
Hypothesis, there is no \ptime algorithm that can decide if there
exists a simple
path of length
$\Omega(f(n)\log^2 n)$ between two nodes in a graph of size $n$ for
any nondecreasing polynomial time computable function $f$ that tends
to infinity. The same holds if we consider trails instead of simple paths.
 
So, first of all, we see that all these languages are in \ctract and
behave very differently in terms of combined complexity. Second, the parameter $k$ of
$a^k$ plays a great role, which motivates us to study the
problem from
the angle of parameterized complexity next.

\makeatletter{}\section{Simple Paths With Length Constraints}\label{sec:ParamCompl}

In this section we investigate the parameterized complexity of
problems that involve simple paths with length constraints. The
problems we consider here are the core of the RPQ evaluation problems
in Section~\ref{sec:tractlanguages}. 
We first give a quick overview of some notions in parameterized complexity.
We follow the exposition of Cygan et al.~\cite{paramAlgo} and refer to
their work for further details.  A \emph{parameterized problem} is a
language $L_k \subseteq \Sigma^* \times \nat$ where, as before,
$\Sigma$ is a fixed, finite alphabet.  For an instance $(x, k) \in
\Sigma^* \times \nat$, we call $k$ the \emph{parameter}. The
\emph{size} $|(x,k)|$ of an instance $(x,k)$ is defined as
$|x|+k$.           A parameterized problem $L_k$ is called
\emph{fixed-parameter tractable} if there exists an algorithm
$\mathcal{A}$, a computable function $f: \nat \rightarrow \nat$, and a
constant $c$ such that, given $(x,k) \in \Sigma^* \times \nat$, the
algorithm $\mathcal{A}$ correctly decides whether $(x,k) \in L_k$ in
time bounded by $f(k) \cdot |(x,k)|^c$, where $c$ is a constant. The
complexity class containing all fixed-parameter tractable problems is
called \fpt.

Let $L_k$ and $L'_{k}$ be two parameterized problems. A
\emph{parameterized reduction} from $L_k$ to $L'_{k}$ is an algorithm
$\mathcal{R}$ that, given an instance $(x,k)$ of $L_k$, outputs an
instance $(x',k')$ of $L'_{k}$ such that
\begin{itemize}
\item $(x,k)$ is a yes-instance of $L_k$ if and only if $(x',k')$ is a
  yes-instance of $L'_{k}$,
\item $k' \leq g(k)$ for some computable function $g$, and
\item the running time of $\mathcal{R}$ is $f(k)\cdot |x|^{O(1)}$ for
  some computable function $f$.
\end{itemize}
Downey and Fellows introduced the W-hierarchy
\cite{DowneySiam-95}. The $k$-Clique problem is W[1]-complete, that is,
complete for the first level of the W-hierarchy
\cite{DowneyTCS-95}. Therefore, $k$-Clique not being fixed-parameter
tractable is equivalent to FPT $\neq$ W[1], which is a standard
assumption in parameterized complexity.

\subsection{One Simple Path}\label{sec:onepath}

We consider the following parameterized problems.
\begin{itemize}
\item $\nodesimpath_k$: Given an instance $((G,s,t),k)$ with $k \in \nat$, is
  there a simple path from $s$ to $t$ of length exactly $k$ in $G$?
\item $\nodesimpath_{\leq k}$ and $\nodesimpath_{\geq k}$: 
      these are
  defined analogously to $\nodesimpath_k$ but ask if there is a simple
  path of length $\leq k$ and $\geq k$, respectively.
\end{itemize}
The problems $\edgesimpath_k$,  $\edgesimpath_{\leq k}$, and
$\edgesimpath_{\geq k}$ are defined analogously but consider trails
instead of simple paths.

These three problems are in \fpt, but the techniques to prove it are
quite different. For $\nodesimpath_k$, membership in \fpt follows from
the famous color
coding technique~\cite{AlonYZ-jacm95}.
\begin{theorem}[Alon et al.~\cite{AlonYZ-jacm95}] \label{FPT:ak}
  $\nodesimpath_{k}$ is in \fpt.
\end{theorem}
$\nodesimpath_{\leq k}$ is trivially in \fpt because the shortest path
problem is in \ptime.
\begin{theorem} \label{FPT:leq-ak}
$\nodesimpath_{\leq k}$ is in \ptime (and therefore in \fpt).\end{theorem}

Finally, $\nodesimpath_{\geq k}$ can be shown to be in \fpt by adapting methods from Fomin et
al.~\cite{fomin}. They proved that finding simple cycles of length at
least $k$ is in \fpt for cycles and discovered that their technique also works for paths
\cite{holgercomm}. The following theorem is therefore due to the
authors of \cite{fomin}. We present a proof in Appendix~\ref{app:paramcompl}
because we need it to prove 
   Theorems \ref{theo:WConstOrNOTIsFPT} and \ref{theo:dichotomy}.
  (We note that Fomin et al.~\cite{fomin} did already consider
\longdirpath on \emph{undirected} graphs, but the techniques
needed on directed graphs are quite different.) 

\begin{restatable}{theorem}{LongDirPathFPT}(Similar to Theorem 5.3 in \cite{fomin}) \label{theorem:longdirpathInFPT}
\longdirpath is in \fpt.
\end{restatable}

By Lemma~\ref{lem:linegraph}, the complexities of
Theorems~\ref{FPT:ak},\ref{FPT:leq-ak}, and
\ref{theorem:longdirpathInFPT} carry over from simple paths to trails.
\begin{theorem}
	$\edgesimpath_k$,  $\edgesimpath_{\leq k}$, and $\edgesimpath_{\geq k}$ are in FPT
\end{theorem}

\subsection{Two Node-Disjoint Paths}\label{sec:twodisjointpaths}

We consider variants of the \twodisjointpaths problem \cite{FortuneHW-TCS80}. A \emph{two-colored graph} is a directed graph in which every
edge is given one of two colors, say $a$ or $b$. An \emph{$a$-colored
  path} is a path consisting of only $a$-colored edges. We will denote an
$a$-colored edge from $u$ to $v$ with $u \stackrel{a}{\to} v$ (similar
for $b$-colored edges). In the remainder we abbreviate
$a$-colored edge and $a$-colored path by $a$-edge and $a$-path, respectively.
We consider the following parameterized problems.
\begin{itemize}
\item \knodedisjointpaths: Given a graph $G$, nodes $s_1,t_1,s_2,t_2$,
  and parameter $k \in \nat$, are there simple paths
  $p_1$ from $s_1$ to $t_1$ and $p_2$ from $s_2$ to $t_2$ such that
  $p_1$ and $p_2$ are node-disjoint and $p_1$ has length $k$?
\item \knodecolordisjointpaths: Given a two-colored graph $G$ and nodes $s_a,t_a,s_b,t_b$,
  is there a simple $a$-path $p_a$ from $s_a$ to $t_a$ and
  a simple $b$-path $p_b$ from $s_b$ to $t_b$ such that $p_a$
  and $p_b$ are node-disjoint and $p_a$ has length $k$?
\end{itemize}
It is well-known that \twodisjointpaths, the non-parameterized version
of \knodedisjointpaths, is \np-complete~\cite{FortuneHW-TCS80}.
Cai and Ye~\cite{CaiWG-16}
proved that \knodedisjointpaths is in FPT for \emph{undirected
  graphs}, both for the cases where one wants simple paths or
trails. They left the cases for directed graphs as open
problems \cite[Problem 2]{CaiWG-16}. We solve one of the cases by
showing in Theorem~\ref{theorem:two-disjoint-is-hard} that \knodedisjointpaths is W[1]-hard. We start
by proving that \knodecolordisjointpaths is W[1]-hard, because the proof
for \knodedisjointpaths relies on it.

\begin{restatable}{theorem}{TheoTwoColor}\label{theo:kcolor}
  \knodecolordisjointpaths is W[1]-hard. 
\end{restatable}
\begin{proof}
  The proof is inspired by an adaptation of Grohe and Gr\"uber
  \cite[Lemma 16]{GroheICALP-07} of a proof by 
  Slivkins~\cite[Theorem 2.1]{slivkins}. Slivkins proved that
  $k$ \textsf{Disjoint Paths} is W[1]-hard in acyclic graphs, that is, he
  showed that it W[1]-hard to decide, given a DAG   $G$ and nodes $s_1, t_1, \ldots, s_k, t_k$ (with parameter $k$), if
  there are pairwise \emph{edge-disjoint} simple paths from $s_i$ to
  $t_i$ for each $i = 1,\ldots,k$. 
              
  We reduce from \kclique, which is well known to be
  W[1]-complete~\cite[Corollary 3.2]{DowneyTCS-95}.
                                  Let $G=(V,E)$ be an \emph{undirected} graph and assume \mbox{w.l.o.g.} that $V =
  \{1,\ldots,n\}$. We will construct a two-colored graph $G'$ with
  $kn\cdot 2(k+1) +k(k-1)/2+ 2(k+1)$   nodes such that $G$ has a \kclique if and only if $G'$ has
  node-disjoint simple paths $p_1$ from $s_1$ to $t_1$ and $p_2$ from
  $s_2$ to $t_2$ such that $p_1$ is $a$-colored and has length $k'
  \in \Theta(k^2)$ while $p_2$ is $b$-colored.  The graph $G'$
  contains $kn$ gadgets $G_{i,j}$ with $i=1,\ldots, k$ and $j =
  1,\ldots,n$, each consisting of $2(k+1)$ nodes. Gadgets will be
  ordered in $k$ rows, where row $i$ has gadgets $G_{i,1}, \ldots,
  G_{i,n}$.  Furthermore, $G'$ contains $k+1$ additional nodes
  $r_1,\ldots,r_{k+1}$ that link the rows together, and $k+1+k(k-1)/2$
  control nodes $c_1, \ldots c_{k+1}$ and $c_{i_1i_2}$ with $1 \leq i_1 <
  i_2 \leq k$ that will limit the number    of disjoint paths from row
  $i$ to row $i+1$ or from row $i_1$ to $i_2$,
  respectively.\footnote{We note that Theorem~\ref{theo:kcolor} can 
    be proved without using control nodes, but we need them to for
    Theorem~\ref{theorem:two-disjoint-is-hard} where we then only require a small change to the
    construction.}
  We define $s_1 = c_1$, $t_1 = c_{k+1}$, $s_2 = r_1$, and $t_2 = r_{k+1}$.

  We will now explain how the nodes are connected.  Each gadget
  contains a disjoint copy of $2(k+1)$ nodes which we call $u_1,u_2,\ldots,u_{k+1}$ and
  $v_1,v_2,\ldots,v_{k+1}$. To simplify notation, we give these nodes
  the same name in figures, even though they are different. One such gadget is depicted in Figure~\ref{fig:OneGadget}. To avoid
  ambiguity, we may also refer to node $u_\ell$ in gadget $G_{i,j}$ by
  $G_{i,j}[u_\ell]$. Each gadget contains edges $u_\ell
  \stackrel{a}{\to} v_\ell$ (for every $\ell = 1,\ldots,k+1$) and
  $u_\ell \stackrel{b}{\to} u_{\ell+1}$ and $v_\ell \stackrel{b}{\to}
  v_{\ell+1}$ (for every $\ell = 1,\ldots,k$).  
        
\tikzset{
  NODE/.style = {fill, inner sep=1pt, circle },
  T/.style = {draw, rounded corners,
                     to path={-| (\tikztotarget)},
                     }
}

    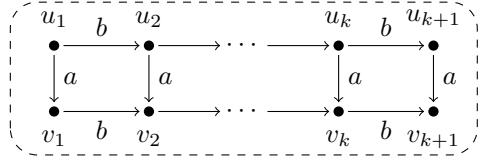
\begin{figure}[t]
   \centering
 \begin{tikzpicture}[->, scale = 0.96, auto ]    \def\distance{1.3}      \foreach \i in {1,2}{
	  \node[label=above:$u_\i$] (A\i) at (\distance*\i,0) {};
	  \node[label=below:$v_\i$] (B\i) at (\distance*\i,\distance*-.7) {}; 
    	  \path	(A\i) edge  node  {$a$} (B\i);
    	  \fill (A\i) circle (2pt);
    	  \fill (B\i) circle (2pt);
    }

    \foreach \i in {4}{
	  \node[label=above:$u_k$] (A\i) at (\distance*\i,0) {};
	  \node[label=below:$v_k$] (B\i) at (\distance*\i,\distance*-.7) {}; 
    	  \path	(A\i) edge  node  {$a$} (B\i);
    	  \fill (A\i) circle (2pt);
    	  \fill (B\i) circle (2pt);
    }    
    
    \foreach \i in {5}{
	  \node[label=above:$u_{k+1}$] (A\i) at (\distance*\i,0) {};
	  \node[label=below:$v_{k+1}$] (B\i) at (\distance*\i,\distance*-.7) {}; 
    	  \path	(A\i) edge  node  {$a$} (B\i);
    	  \fill (A\i) circle (2pt);
    	  \fill (B\i) circle (2pt);
    }
	\path 
    (A1) edge	  node  {$b$} (A2)
    (B1) edge	 [->]    node  [swap] {$b$} (B2) 
    (A4) edge	 [->]    node  {$b$} (A5)
    (B4) edge	 [->]    node  [swap] {$b$} (B5);
    
    \node[label=center:$\cdots$] (A3) at (\distance*3,0) {}; 
    \node[label=center:$\cdots$] (B3) at (\distance*3,\distance*-.7) {}; 
    \draw[->] (A2) -- ($(A2)!0.75!(A3)$);
    \draw[->] (B2) -- ($(B2)!0.75!(B3)$);
    \draw[->] ($(A3)!0.25!(A4)$) -- (A4);
    \draw[->] ($(B3)!0.25!(B4)$) -- (B4);
    \def\offset{0.6}     \draw[dashed, rounded corners=10pt] (\distance*1-\offset,\distance*-.7-\offset) rectangle (\distance*5+\offset,\offset) {};     \end{tikzpicture}
    \caption{Internal structure of each of the gadgets $G_{i,j}$.}
    \label{fig:OneGadget}
  \end{figure}

  We now explain how the gadgets $G_{i,j}$ are connected within the
  same row, see Figure~\ref{fig:gadgets:row}. In each row $i\in
  \{1,\ldots, k\}$, node $r_i$ has two outgoing edges $r_i
  \stackrel{b}{\to} G_{i,1}[u_1]$ and $r_i \stackrel{b}{\to}
  G_{i,2}[v_1]$. We also have two incoming edges for $r_{i+1}$, namely
  $G_{i,n-1}[u_{k+1}] \stackrel{b}{\to} r_{i+1}$ and $G_{i,n}[v_{k+1}]
  \stackrel{b}{\to} r_{i+1}$. Furthermore, we have the edges
  $G_{i,j}[u_{k+1}] \stackrel{b}{\to} G_{i,j+1}[u_1]$ and
  $G_{i,j}[v_{k+1}] \stackrel{b}{\to} G_{i,j+1}[v_1]$ for every $j =
  1,\ldots,n-1$. We also add edges $G_{i,j}[u_{k+1}] \stackrel{b}{\to}
  G_{i,j+2}[v_1]$ for every $j = 1,\ldots,n-2$.

\begin{figure*}[t]
\centering
\begin{tikzpicture}[scale = 1]
\node[NODE] (Si) at (-1.5,1) {};
\node[anchor=south, yshift=+0.2ex] at (Si) {$r_i$};
\node[NODE] (Si2) at (2.5+9,0) {};
\node[anchor=north, yshift=-0.2ex] at (Si2) {$r_{i+1}$};

\foreach \i in {1,...,3}{
	\node[NODE] (A\i) at (2.5*\i-2.5,1) {};
    \node[anchor=south, yshift=+0.2ex]       at (A\i) {$\mathstrut u_1$};
	
	\node[NODE] (B\i) at (2.5*\i-2.5,0) {};
    \node[anchor=north, yshift=-0.2ex]       at (B\i) {$\mathstrut v_1$};
	
	\node[NODE] (Ak\i) at (2.5*\i-1.5,1) {};
    \node[anchor=south, yshift=+0.2ex]        at (Ak\i) {$\mathstrut u_{k+1}$};
	\node[NODE] (Bk\i) at (2.5*\i-1.5,0) {};
    \node[anchor=north, yshift=-0.2ex]        at (Bk\i) {$\mathstrut v_{k+1}$}; 
  }

\node[NODE] (A4) at (9,1) {};
\node[anchor=south, yshift=+0.2ex]  at (A4) {$u_1$};

\node[NODE] (B4) at (9,0) {};
\node[anchor=north, yshift=-0.2ex] at (B4) {$v_1$};

\node[NODE] (Ak4) at (10,1) {};
\node[anchor=south, yshift=+0.2ex]  at (Ak4) {$u_{k+1}$};

\node[NODE] (Bk4) at (10,0) {};
\node[anchor=north, yshift=-0.2ex] at (Bk4) {$v_{k+1}$};
\def\offsetx{0.5} \def\offsety{0.6} \draw[dashed, rounded corners=10pt] (9-\offsetx,0-\offsety) rectangle (10+\offsetx,1+\offsety) {};

\foreach \i in {1,...,4}{
  \draw[dotted] (A\i) -- (Ak\i);
  \draw[dotted] (B\i) -- (Bk\i); 
}
\foreach \i in {1,...,3}{
\draw[dashed, rounded corners=10pt] (2.5*\i-2.5-\offsetx,0-\offsety) rectangle (2.5*\i-1.5+\offsetx,1+\offsety) {};   } 

\foreach \i in {1,...,3}{
  \draw (2.5*\i-2,-0.7) coordinate[label={below:$G_{i,\i}$}];
  }
  \draw (9.5,-0.7) coordinate[label={below:$G_{i,n}$}];
  
  \draw[->] (Si) -- (A1);  
  \draw[->] (Bk4) -- (Si2); 
  
  \draw[->] (Ak1) -- (A2);
  \draw[->] (Bk1) -- (B2);
  \draw[->] (Ak2) -- (A3);
  \draw[->] (Bk2) -- (B3);
  
    \draw[->] (Ak1) .. controls ($(A2)!0.5!(B2)$) and ($(Ak2)!0.5!(Bk2)$) .. (B3);
  \draw[->] (Si) .. controls ($(A1)!0.5!(B1)$) and ($(Ak1)!0.5!(Bk1)$) .. (B2);
  \draw (Ak2) .. controls ($(A3)!0.5!(B3)$) and ($(Ak3)!0.5!(Bk3)$) .. (7.5-1.5+0.2*4-0.2,0.2);   \draw (Ak3) -- (7.5-1.5+0.2*4-0.2,0.8);   \draw[->](2.5*3-1.5+0.8*4-0.8,0.2) -- (B4);
  \draw[->](2.5*3-1.5+0.8*4-0.8,0.8) .. controls ($(A4)!0.5!(B4)$) and ($(Ak4)!0.5!(Bk4)$) .. (Si2);  
  
    \draw[dotted, ->] (Ak3) -- (A4); 
  \draw[] (Ak3) -- ($(Ak3)!0.2!(A4)$); 
  \draw[] ($(Ak3)!0.8!(A4)$) -- (A4);
  \draw[dotted, ->] (Bk3) -- (B4); 
  \draw[] (Bk3) -- ($(Bk3)!0.2!(B4)$); 
  \draw[] ($(Bk3)!0.8!(B4)$) -- (B4); 
 
\end{tikzpicture}
\caption{The $b$-edges in row $i$. The internal structure of the
  $G_{i,j}$ is as in Figure~\ref{fig:OneGadget}.} 
\label{fig:gadgets:row}
\end{figure*}
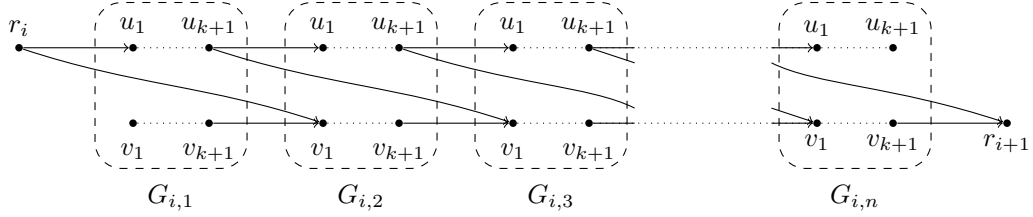

   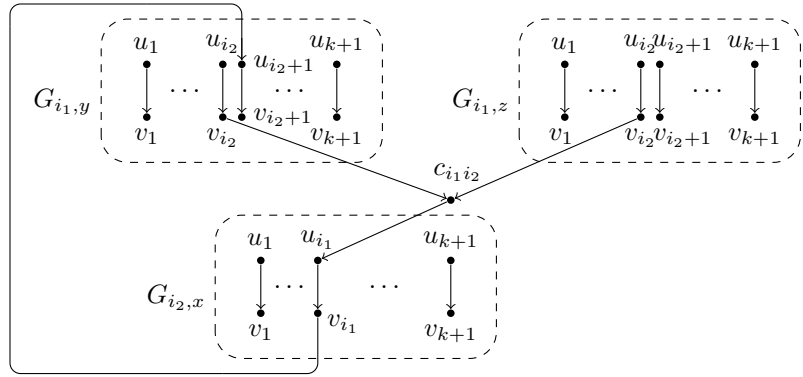
\begin{figure}[t]
   \centering
      \begin{tikzpicture}[scale = 1]
   \def\uphere{0.9}
   \def\shortenA{0.3}
        \foreach \i in {1}{
	\node[NODE] (A\i) at (1.5*\i-2.5,1-3.5*\i) {};
	\node[NODE] (B\i) at (1.5*\i-2.5,0-3.5*\i+\shortenA) {};
	\node[NODE] (Ak\i) at (1.5*\i,1-3.5*\i) {};
	\node[NODE] (Bk\i) at (1.5*\i,0-3.5*\i+\shortenA) {};  
	\node[anchor=south, yshift=+0.2ex]  at (A\i) {$u_1$};
	\node[anchor=north, yshift=-0.2ex]  at (B\i) {$v_1$};
	\node[anchor=south, yshift=+0.2ex]  at (Ak\i) {$u_{k+1}$};
	\node[anchor=north, yshift=-0.2ex]  at (Bk\i) {$v_{k+1}$};  
      }
      \foreach \i in {2}{ 	\node[NODE] (A\i) at (1.5*\i-2.5,1-3.5*\i+\uphere) {};
	\node[NODE] (B\i) at (1.5*\i-2.5,0-3.5*\i+\uphere+\shortenA) {};
	\node[NODE] (Ak\i) at (1.5*\i,1-3.5*\i+\uphere) {};
	\node[NODE] (Bk\i) at (1.5*\i,0-3.5*\i+\uphere+\shortenA) {};  
	\node[anchor=south, yshift=+0.2ex]  at (A\i) {$u_1$};
	\node[anchor=north, yshift=-0.2ex]  at (B\i) {$v_1$};
	\node[anchor=south, yshift=+0.2ex]  at (Ak\i) {$u_{k+1}$};
	\node[anchor=north, yshift=-0.2ex]  at (Bk\i) {$v_{k+1}$};  
      }
    \foreach \i in {3}{
	\node[NODE] (A\i) at (1.5*\i,-2.5) {};
	\node[NODE] (B\i) at (1.5*\i,-3.5+\shortenA) {};
	\node[NODE] (Ak\i) at (1.5*\i+2.5,-2.5) {};
	\node[NODE] (Bk\i) at (1.5*\i+2.5,-3.5+\shortenA) {};  
	\node[anchor=south, yshift=+0.2ex]  at (A\i) {$u_1$};
	\node[anchor=north, yshift=-0.2ex, xshift=-0.3ex]  at (B\i) {$v_1$};
	\node[anchor=south, yshift=+0.2ex]  at (Ak\i) {$u_{k+1}$};
	\node[anchor=north, yshift=-0.2ex]  at (Bk\i) {$v_{k+1}$};  
      }
            \node[NODE] (Ai1) at ($(A1)!0.4!(Ak1)$){};
      \node[NODE] (Bi1) at ($(B1)!0.4!(Bk1)$) {};
      \node[anchor=south, yshift=+0.2ex] at (Ai1) {$u_{i_2}$};
      \node[anchor=north, yshift=-0.2ex] at (Bi1) {$v_{i_2}$};       
      \node[NODE] (Ai2) at ($(A2)!0.3!(Ak2)$){};
      \node[NODE] (Bi2) at ($(B2)!0.3!(Bk2)$) {};  
      \node[anchor=south, yshift=+0.2ex] at (Ai2) {$u_{i_1}$};
      \node[anchor=west,yshift=-0.9ex] at (Bi2) {$v_{i_1}$};  
      \node[NODE] (Ai3) at ($(A3)!0.4!(Ak3)$){};
      \node[NODE] (Bi3) at ($(B3)!0.4!(Bk3)$) {};
      \node[anchor=south, yshift=+0.2ex] at (Ai3) {$u_{i_2}$};
      \node[anchor=north, yshift=-0.2ex] at (Bi3) {$v_{i_2}$};     
            \node[NODE] (Ai5) at ($(A1)!0.5!(Ak1)$){};
      \node[NODE] (Bi5) at ($(B1)!0.5!(Bk1)$) {};
      \node[anchor=west, xshift=+0.2ex] at (Ai5) {$u_{i_2+1}$};
      \node[anchor=west, xshift=+0.2ex] at (Bi5) {$v_{i_2+1}$};
      
       \node[NODE] (Ai4) at ($(A3)!0.5!(Ak3)$){};
      \node[NODE] (Bi4) at ($(B3)!0.5!(Bk3)$) {};
      \node[anchor=south, yshift=+0.2ex,xshift=+2ex] at (Ai4) {$u_{i_2+1}$};
      \node[anchor=north, yshift=-0.2ex,xshift=+2ex] at (Bi4) {$v_{i_2+1}$};  
            \node[label=center:$\ldots$, fill = white] (dotol) at ($(A1)!0.5!(Bi1)$) {};  
      \node[label=center:$\ldots$, fill = white] (dotor) at ($(Ai5)!0.5!(Bk1)$) {}; 
      \node[label=center:$\ldots$, fill = white] (dotul) at ($(B2)!0.5!(Ai2)$) {}; 
      \node[label=center:$\ldots$, fill = white] (dotur) at ($(Bi2)!0.5!(Ak2)$) {}; 
      \node[label=center:$\ldots$, fill = white] (dotol) at ($(A3)!0.5!(Bi3)$) {};  
      \node[label=center:$\ldots$, fill = white] (dotor) at ($(Ai4)!0.5!(Bk3)$) {}; 
            \foreach \i in {1,2,3}{
        \draw[->] (A\i) -- (B\i);
        \draw[->] (Ai\i) -- (Bi\i);
        \draw[->] (Ak\i) -- (Bk\i);  
      }
            \def\offset{0.6}         \draw[dashed, rounded corners=10pt] (1.5-2.5-\offset,0-3.5-\offset+\shortenA) rectangle (1.5+\offset,1-3.5+\offset) {};         \draw[dashed, rounded corners=10pt] (1.5*2-2.5-\offset,0-3.5*2-\offset+\uphere+\shortenA) rectangle (1.5*2+\offset,1-3.5*2+\offset+\uphere) {};       \draw[dashed, rounded corners=10pt] (1.5*3-\offset,-3.5-\offset+\shortenA) rectangle (1.5*3+2.5+\offset,-2.5+\offset) {};       
            \draw (-1.6,-3) coordinate[label={left:$G_{i_1,y}$}];
	  \draw (3.9,-3) coordinate[label={left:$G_{i_1,z}$}];
      \draw (-0.1,-6.5+\uphere) coordinate[label={left:$G_{i_2,x}$}];
  
            \node[NODE] (C) at (3,-5.5+\uphere+\shortenA) {};
      \node[anchor=west, xshift=-0.2ex, label={above:$c_{i_1i_2}$}] at (C) {};
      \draw[->] (Bi1) -- (C) ;
      \draw[->] (Bi3) -- (C);
      \draw[->] (C) -- (Ai2);
      \draw[->] (Ai5) -- (Bi5);
      \draw[->] (Ai4) -- (Bi4);  
     \draw (0,-7.8+\uphere+\shortenA) edge[T] (Bi2)            (0,-7.8+\uphere+\shortenA) edge[T] (-2.8,-2)
           (-1,-1.7) edge[T] (-2.8,-2)
           (-1,-1.7) edge[T] (Ai5)
     ;
     \draw[->] ($(Ai5)!-0.5!(Bi5)$)--(Ai5); 
     ;   
    \end{tikzpicture}
  \caption{
    The $a$-edges in the gadgets and between gadgets $G_{i_1,y}$,
    $G_{i_1,z}$ and $G_{i_2,x}$, with $i_1<i_2-1$,
    under the
    assumption that $(x,y) \in E$ and $(x,z) \notin E$.}
     \label{fig:gadgets:differentrows}
   \end{figure}

\begin{figure*}[ht]
\centering
\begin{tikzpicture}[scale = 1]

\foreach \i in {1,...,3}{
		\node[NODE] (A\i) at (4*\i-2.5,1) {};
    \node[anchor=south, yshift=+0.2ex, xshift=-0.5ex]       at (A\i) {$\mathstrut u_1$};
    
	\node[NODE] (B\i) at (4*\i-2.5,0) {};
    \node[anchor=south, yshift=+0.2ex, xshift=-0.5ex]       at (B\i) {$\mathstrut v_1$};
	
	\node[NODE] (Ak\i) at (4*\i,1) {};
    \node[anchor=south, yshift=+0.2ex, xshift=+1.3ex]        at (Ak\i) {$\mathstrut u_{k+1}$};
	\node[NODE] (Bk\i) at (4*\i,0) {};
    \node[anchor=south, yshift=+0.2ex, xshift=+1.3ex]        at (Bk\i) {$\mathstrut v_{k+1}$}; 
  }
\def\dist{2.3} 
\foreach \i in {1,...,3}{
	\node[NODE] (bA\i) at (4*\i-2.5,1-\dist) {};
    \node[anchor=north, yshift=-0.2ex, xshift=-0.5ex]       at (bA\i) {$\mathstrut u_1$};
	
	\node[NODE] (bB\i) at (4*\i-2.5,0-\dist) {};
    \node[anchor=north, yshift=-0.2ex, xshift=-0.5ex]       at (bB\i) {$\mathstrut v_1$};
	
	\node[NODE] (bAk\i) at (4*\i,1-\dist) {};
    \node[anchor=north, yshift=-0.2ex, xshift=+1.3ex]        at (bAk\i) {$\mathstrut u_{k+1}$};
	
	\node[NODE] (bBk\i) at (4*\i,0-\dist) {};
    \node[anchor=north, yshift=-0.2ex, xshift=+1.3ex]        at (bBk\i) {$\mathstrut v_{k+1}$}; 
  }

\foreach \i in {1,...,3}{
  \draw[dotted] (A\i) -- (Ak\i);
  \draw[dotted] (B\i) -- (Bk\i); 
}
\foreach \i in {1,...,3}{
  \draw[dotted] (bA\i) -- (bAk\i);
  \draw[dotted] (bB\i) -- (bBk\i); 
}

\def\offsetx{0.6} \def\offsety{0.6} \foreach \i in {1,...,3}{
\draw[dashed, rounded corners=10pt] (4*\i-2.5-\offsetx,0.3-\offsety) rectangle (4*\i+\offsetx,1+\offsety) {};   } 
\foreach \i in {1,...,3}{
\draw[dashed, rounded corners=10pt] (4*\i-2.5-\offsetx,0-\offsety-\dist) rectangle (4*\i+\offsetx,1-0.3+\offsety-\dist) {};   } 
\foreach \i in {1,...,3}{
  \draw (4*\i-1,2.3) coordinate[label={below:$G_{i,\i}$}];
  }
\foreach \i in {1,...,3}{
  \draw (4*\i-1,-0.7-\dist) coordinate[label={below:$G_{i+1,\i}$}];
  } 
\foreach \i in {1,...,3}{
 \node[NODE] (bu1\i) at ($(bA\i)!0.7!(bAk\i)$) {};
 \node[anchor=north, yshift=-0.2ex,xshift=+1.3ex]  at (bu1\i) {$\mathstrut u_{i+2}$};
 \node[NODE] (bv1\i) at ($(bB\i)!0.7!(bBk\i)$) {};
 \node[anchor=north, yshift=-0.2ex,xshift=+1.3ex]  at (bv1\i) {$\mathstrut v_{i+2}$};
}
\foreach \i in {1,...,3}{
 \node[NODE] (bu2\i) at ($(bA\i)!0.5!(bAk\i)$) {};
 \node[anchor=north, yshift=-0.2ex]  at (bu2\i) {$\mathstrut u_{i}$};
 \node[NODE] (bv2\i) at ($(bB\i)!0.5!(bBk\i)$) {};
 \node[anchor=north, yshift=-0.2ex]  at (bv2\i) {$\mathstrut v_{i}$};
}
\foreach \i in {1,...,3}{
 \node[NODE] (u2\i) at ($(A\i)!0.6!(Ak\i)$) {};
 \node[anchor=south, yshift=+0.2ex]  at (u2\i) {$\mathstrut u_{i+1}$};
 \node[NODE] (v2\i) at ($(B\i)!0.6!(Bk\i)$) {};
 \node[anchor=south, yshift=+0.2ex]  at (v2\i) {$\mathstrut v_{i+1}$};
}

 \node[NODE] (control2)  at (bu11 |- 1,1-0.7*\dist) {}; 
\node[anchor=east, xshift=-0.2ex]  at (control2) {$c_{ii+1}$};
 \node[NODE] (control1) at (bu23 |- 1,1-0.7*\dist) {};
\node[anchor=west, xshift=+0.3ex, yshift=-0.2ex]  at (control1) {$c_{i+1}$};

\foreach \i in {1,...,3}{
 \draw[->] (v2\i) -- (control2); 
 \draw[->] (control2) -- (bu2\i); 
}
\foreach \i in {1,...,3}{
 \draw[->] (Bk\i) -- (control1); 
 \draw[->] (control1) -- (bu1\i); 
}
\end{tikzpicture}
\caption{The $a$-edges from row $i$ to row $i+1$. (We
  assume $n=3$ in the picture).}
\label{fig:gadgets:ci}
\end{figure*}
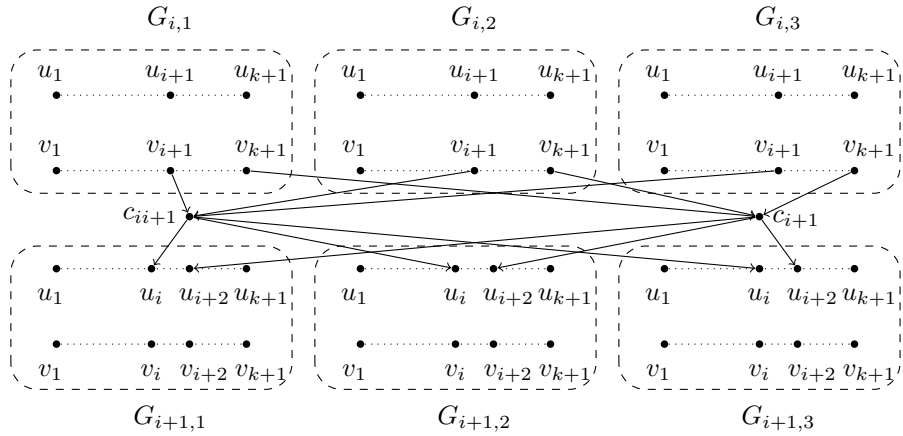

   Finally, we explain how the gadgets $G_{i,j}$ are connected in
   different rows via the control nodes $c_i$ and $c_{i_1i_2}$.  We
   first consider the edges from row $i$ to $i+1$.
      In each row $i=1,\ldots,k-1$, and every $j = 1,\ldots,n$, we add
   the edges $G_{i,j}[v_{k+1}] \stackrel{a}{\to} c_{i+1}$ and $c_{i+1}
   \stackrel{a}{\to} G_{i+1,j}[u_{i+2}]$.  Furthermore, we add the
   edges $c_{1} \stackrel{a}{\to} G_{1,j}[u_{2}]$ and
   $G_{k,j}[v_{k+1}] \stackrel{a}{\to} c_{k+1}$.  We connect two rows
   $i_1,i_2$, with $1 \leq i_1 < i_2 \leq k$, by adding the edges
   $G_{i_1,j}[v_{i_2}] \stackrel{a}{\to} c_{i_1i_2}$, and $c_{i_1i_2}
   \stackrel{a}{\to} G_{i_2,j}[u_{i_1}]$ for all $j=1,\ldots,n$.  The
   edges in $G$ are modeled in $G'$ by adding the edge
   $G_{i_2,x}[v_{i_1}] \stackrel{a}{\to} G_{i_1,y}[u_{i_2+1}]$ if and
   only if $1\leq i_1 < i_2 \leq k$, $x \neq y$, and $(x,y) \in
   E$. This is illustrated in Figures~\ref{fig:gadgets:differentrows} and \ref{fig:gadgets:ci}.
   Finally, we define $k' = k(k-1)/2\cdot 5+3k$. We prove in
   Appendix~\ref{app:twodisjointpaths} that the reduction is correct.
\end{proof}

In the proof of Theorem~\ref{theo:kcolor}, we have $t_1 = c_k$ only has incoming edges and
$s_2 = r_1$ only has outgoing edges. We note that
the reduction is also correct if $t_1 = s_2$. 
\begin{corollary} \label{akb*-is-hard}
\kcolordisjointpaths is W[1]-hard even if $t_1 = s_2$.
  \end{corollary}

The two colors in the proof of Theorem~\ref{theo:kcolor} play a central role:
since the $a$-path cannot use any $b$-edges and vice
versa, we have much control over where the two paths can be. The following Theorem
shows that the construction in Theorem~\ref{theo:kcolor} can be
strengthened so that we do not need the two colors.
\begin{restatable}{theorem}{TheoDisjointPaths}\label{theorem:two-disjoint-is-hard}
  \knodedisjointpaths is W[1]-hard.
\end{restatable}
\begin{proof}
  We adapt the reduction from Theorem~\ref{akb*-is-hard}.  The only
  change we make is that we replace each $b$-edge by a
  directed path of $k'$ edges (introducing $k'-1$ new nodes for each
  such edge). We prove in Appendix~\ref{app:twodisjointpaths} that the reduction is correct.
\end{proof}

For completeness, we mention the complexity of other variants of
\knodedisjointpaths, some of which can be shown by extending the technique
from Theorem~\ref{theorem:two-disjoint-is-hard}. We define $\twonodedisjointpaths_{\leq k}$ and
$\twonodedisjointpaths_{\geq k}$ analogously to \knodedisjointpaths by
requiring that $p_1$ has length $\leq k$ and $\geq k$, respectively.
\begin{restatable}{theorem}{TheoPathsAllTheRest}\label{theo:3statements}
  \begin{itemize}
  \item   $\twonodedisjointpaths_{\leq k}$ is W[1]-hard.
  \item   $\twonodedisjointpaths_{\geq k}$ is \np-complete for every
    constant $k \in \nat$ (\cite{FortuneHW-TCS80}).
  \item \knodecolordisjointpaths, \knodedisjointpaths, and $\twonodedisjointpaths_{\leq k}$ are in W[P].
      \end{itemize}
\end{restatable}

\subsection{Two Edge-Disjoint Trails}\label{sec:twoedgedisjoint}
Here, we study the trail versions of the disjoint paths problem where
we require the trails to be \emph{edge-disjoint}.
\begin{itemize}
\item \kedgedisjointpaths: Given a graph $G$, nodes $s_1,t_1, s_2, t_2$, and parameter $k \in \nat$, are there trails $p_1$ from $s_1$ to $t_1$ and $p_2$ from $s_2$ to $t_2$ such that $p_1$ and $p_2$ are edge-disjoint and $p_1$ has length $k$?
\end{itemize}

The following theorem shows W[1] hardness for \kedgedisjointpaths. 
\begin{restatable}{theorem}{twoEdgeDisjointPaths}\label{theo:edgeTwodisjoint}
\kedgedisjointpaths is W[1]-hard.
\end{restatable}
The W[1] hardness follows directly from
Theorem~\ref{theorem:two-disjoint-is-hard} and the reduction in
Lemma~\ref{lem:splitgraph}. Nonetheless, we give a proof analogous to
Theorem~\ref{theorem:two-disjoint-is-hard} in
Appendix~\ref{app:paramcompl}, since we feel that this might help to better
understand our W[1] hardness proof of
Theorem~\ref{theo:edgedichotomy}. 

For completeness, we mention the complexity of other variants of
\kedgedisjointpaths, some of which can be shown by extending the technique
from Theorem~\ref{theorem:two-disjoint-is-hard}. We define $\twoedgedisjointpaths_{\leq k}$ and
$\twoedgedisjointpaths_{\geq k}$ analogously to \kedgedisjointpaths by
requiring that $p_1$ has length $\leq k$ and $\geq k$, respectively.
The next theorem is an immediate result of Lemma~\ref{lem:splitgraph} and Theorem~\ref{theo:3statements}.
\begin{restatable}{theorem}{TheoEdgePathsAllTheRest}
	\begin{itemize}
		\item   $\twoedgedisjointpaths_{\leq k}$ is W[1]-hard. 		\item   $\twoedgedisjointpaths_{\geq k}$ is \np-complete for every
		constant $k \in \nat$ (\cite{FortuneHW-TCS80}). 		\item \kedgedisjointpaths, and $\twoedgedisjointpaths_{\leq k}$ are in W[P]. 	\end{itemize}
\end{restatable}
 
\makeatletter{}\section{Parameterized Complexity of Simple Regular Paths}\label{sec:tractlanguages}\label{sec:param-srp}
We now return to regular path query evaluation and consider
parameterized versions of \nodespath and \edgespath. In contrast to
Section~\ref{sec:ParamCompl}, the parameter $k$ will not be a
constraint on the length of the paths. Instead, we will search for
paths of arbitrary length (as in Sections~\ref{sec:enumerate} and \ref{sec:enum-yen}) and the
parameter $k$ will be determined by the regular expression. That is,
the parameterized version of \nodespath and \edgespath still has graph $(G,s,t)$ and
expression $r$ as input, but the parameter $k$ is implicitly determined by
$r$.  

\medskip \noindent \textbf{Some Concrete Languages.}
We
first consider a few simple examples of such problems and generalize
the approach later in this section. For $k\in\nat$, we define the regular expressions

\centerline{$a^k$, \qquad $(a?)^k$,\qquad $a^k a^*$,\qquad $a^k b^*$, \qquad and $a^{k-1} b a^*$}

\noindent to have \emph{parameter $k$}. By abusing notation, we denote by $a^k
b^*$ the class of regular expressions $\{a^k b^* \mid k \in \nat\}$
(similar for the other abovementioned expressions).  As such, the
\emph{parameterized problem} $\nodespath(a^k b^*)$ asks, given $(G,s,t)$
and a regular expression $r$ of the form $a^k b^*$, if there exists a
simple path from $s$ to $t$ that matches $r$. It is in \fpt if it can
be decided by an algorithm that runs in time $f(k) \cdot (|G|+|r|)^c$
for a computable function $f$ and a constant $c$.  The following can
now be easily deduced from Section~\ref{sec:ParamCompl}.
\begin{theorem} \label{theo:languageoverview}
\begin{enumerate}[(a)]
\item $\nodespath(a^k)$, $\nodespath((a?)^k)$, and $\nodespath(a^ka^*)$ are in \fpt.  \label{theo:languageoverview:a}
\item $\nodespath(a^kb^*)$ and $\nodespath(a^{k-1}b a^*)$ 
  are   W[1]-hard. \label{theo:languageoverview:b}
\end{enumerate}
\end{theorem}
\begin{proof}
  Part (a) is immediate from Theorems \ref{FPT:ak}, \ref{FPT:leq-ak},
  and \ref{theorem:longdirpathInFPT}, respectively.  For part (b), the
  hardness of $\nodespath(a^kb^*)$ is immediate from Corollary
  \ref{akb*-is-hard}. The hardness of $\nodespath(a^{k-1}b a^*)$ is
  obtained from Theorem \ref{theorem:two-disjoint-is-hard} by applying
  a simple proof of Mendelzon and Wood \cite[Theorem 1
  (2)]{mendelzon}, reducing \textsf{TwoDisjointPaths} to
  $\nodesimpath(a^*ba^*)$. 
  The idea is to add $t_1 \stackrel{b}{\to} s_2$ (and label all other edges $a$). 
  Then, every path from $s_1$ to $t_2$ matching $a^*ba^*$ must contain the $b$-edge. This implies that such 
  a path exists if and only if there exist two node-disjoint paths, one from
    $s_1$ to $t_1$ and the other from $s_2$ to $t_2$.
\end{proof}

We can even slightly generalize the proof of
Theorem~\ref{theo:languageoverview}(\ref{theo:languageoverview:a}) to deal with more complex
languages. Notice that the following result implies that $\nodespath(a^k b? a^*)$ is in
FPT, whereas $\nodespath(a^k b a^*)$ is W[1]-hard by Theorem~\ref{theo:languageoverview}(\ref{theo:languageoverview:b}).
\begin{restatable}{theorem}{WConstOrNOTIsFPT} \label{theo:WConstOrNOTIsFPT}
  For every constant $c$ and word $w$ with $|w|=c$, the problem
  \nodespath$(a^kw?a^*)$ with parameter $k$ is in FPT. 
\end{restatable}
\begin{proof}[Proof sketch]
  First we use the algorithm from Theorem
  \ref{theorem:longdirpathInFPT} to decide \nodesimpath$(a^ka^*)$. If the
  answer is no, we enumerate all possible paths $p$ that match $w$
  and change the algorithm from Theorem \ref{theorem:longdirpathInFPT}
  to find two disjoint simple paths, not intersecting $p$: one from $s$ to $p[0]$ matching
  $a^k$ and one from $p[c]$ to $t$ matching $a^*$. (Recall that
  $p[i]$ denotes the $i$th node in $p$.)
\end{proof}
If FPT $\neq$ W[1], then Theorem~\ref{theo:WConstOrNOTIsFPT} cannot be generalized to
arbitrary words $w$, since \nodesimpath($b^ka^k?b^*$) with parameter $k$
is W[1]-hard. This can be shown by adapting the proof of Theorem
\ref{theo:kcolor}: We add a path matching $b^k$ from a new node
$s_1'$ to $s_1$ and define $t_1 = s_2$ as in Corollary~\ref{akb*-is-hard}.
Then, the problem whether
there is a simple path matching $b^ka^k?b^*$ from $s'_1$ to $t_2$ is the
same as asking whether there is a simple path matching $a^k b^*$
from $s_1$ to $t_2$.

\subsection{Simple Transitive Expressions}\label{sec:blocklanguages} We now aim at generalizing
the previous results to more general (but still very restricted) regular
expressions. However, we feel that these expressions are relevant and
important from a practical perspective since they constitute more than
$99\%$ of the property paths   found in SPARQL query logs in an extensive
recent study \cite{BonifatiMT-corr17}. Notice that SPARQL property
paths are RPQs with added syntactic sugar, so the syntax of the
expressions is not restricted as, e.g., in Cypher.

\begin{table}[tb]
  \centering
  \begin{tabular}{@{}crll|crll@{}}
    \toprule
    \emph{Expression Type} & \emph{Relative} & $\ell$& STE? &
    \emph{Expression Type} & \emph{Relative} & $\ell$ & STE?\\
    \midrule
    $(a_1 + \cdots +a_\ell)^*$ & 39.12\% & 2--4& yes &
    $a_1 a_2? \cdots a_\ell?$ & 0.02\% & 1--3 & yes \\
    $a^*$ & 26.42\%& & yes &$(a b^*)+c$ & 0.01\%&& no\\
    $a_1 \cdots a_\ell$ & 11.65\% &  2--6 & yes& $a^* b?$ & 0.01\% && yes\\
    $a^* b$ & 10.39\% & & yes& $a b c^*$ &   0.01\%&& yes\\
    $a_1+ \cdots +a_\ell$ & 8.72\% & 2--6 & yes& $!(a+b)$ & 0.01\%&& no\\
    $a^+$ & 2.07\%&& yes &$(a_1 +\cdots +a_\ell)^+$ & 0.01\% & 2& yes\\
    $a_1? \cdots a_\ell?$ & 1.55\% & 1--5&  yes & $A_1 A_2$ &  $<$ 0.01\% & & yes\\
    $a(b_1+\cdots+b_\ell)$ & 0.02\% & 2 & yes& other  & 0.01\%& & mixed\\    
  \end{tabular}
  \caption{Structure of the 250K property paths in the corpus of
    Bonifati et al.~\cite{BonifatiMT-corr17}\label{fig:propertypaths}
}
\end{table}

In the following definition, we use $A \subseteq \Sigma$ to abbreviate
$(a_1 + \cdots + a_n)$ so that $A = \{a_1,\ldots,a_n\}$. We allow $A = \emptyset$.
\begin{definition}\label{def:ste}
  An \emph{atomic expression} is of the form $A \subseteq \Sigma$.  A
  \emph{$k$-bounded expression} is a regular expression of the form
  $A_1 \cdots A_k$ or $A_1? \cdots A_k?$, where $k \geq 0$ and each
  $A_i$ is an atomic expression. 
  Finally, a \emph{simple transitive expression (STE)} is a
  regular expression $$B_\text{pre} A^* B_\text{suff},$$ where
  $B_\text{pre}$ and $B_\text{suff}$ are bounded expressions and $A$
  is an atomic expression. For an STE $r = B_\text{pre} A^*
  B_\text{suff}$, we define the \emph{parameter $k_r = k_1 + k_2$},
  where $B_\text{pre}$ is $k_1$-bounded and $B_\text{suff}$ is
  $k_2$-bounded.
\end{definition}
Notice that about 99.7\% of the property paths in
Table~\ref{fig:propertypaths} are STEs or trivially equivalent to an
STE (by taking $A = \emptyset$, for example).\footnote{In fact, all
  expressions in the full version of Table~\ref{fig:propertypaths} in \cite{BonifatiMT-corr17} except for one
  can be handled with the techniques we present here. They just don't
  fit the definition of STEs.}

\subsection{Two Dichotomies}\label{sec:twodichotomies}

\paragraph*{Dichotomy for Simple Paths}

\begin{definition} 
Let $r = B_\text{pre} A^* B_\text{suff}$
be an STE with $L(r)\neq \emptyset$. 
If $B_\text{pre} = A_1 \cdots A_{k_1} $, then its
\emph{left cut border} 
$c_1$
is the largest value such that $A \not \subseteq
A_{c_1}$
if it exists and zero otherwise. If $B_\text{pre} = A_1? \cdots
A_{k_1}?$, 
then its left cut border is zero. 
We define right cut borders symmetrically (e.g., for
$B_\text{suff} = A'_{k_2} \cdots A'_1$, it is the largest 
$c_2$ such that $A \not \subseteq A_{c_2}$).
\end{definition}
We explain the intuition behind cut borders in Figure~\ref{fig:cuttable}.

For $c \in \nat$, an expression is \emph{$c$-bordered} if the maximum
of its left and right cut borders is $c$. We call a class $\cR$
of STEs \emph{cuttable} if there exists a constant $c \in \nat$ such
that each expression in $\cR$ is $c'$-bordered for some $c' \leq c$.
We can now prove a dichotomy on the complexity of $\nodespath(\cR)$ for
classes of STEs $\cR$, if $\cR$ satisfies the following mild
condition. We say that $\cR$ \emph{can be sampled} if there exists
an 
algorithm that, given $k \in \nat$, returns an 
expression in $\cR$ that is $k'$-bordered with $k' \geq k$, and ``no'' otherwise.

\begin{restatable}{theorem}{dichotomyThm}\label{theo:dichotomy}
  Let $\cR$ be a class of STEs that can be sampled. Then,
  \begin{enumerate}[(a)]
  \item if $\cR$ is cuttable, then $\nodespath(\cR)$ is in \fpt with
    parameter $k_r$ and
  \item otherwise, $\nodespath(\cR)$ is W[1]-hard with
    parameter $k_r$. 
 \end{enumerate}
\end{restatable}
\begin{proof}[Proof idea]
  The main idea of the proof is to attack case (a) using the
  techniques for proving
  Theorem~\ref{theorem:longdirpathInFPT}. If $\cR$ is cuttable, we can
  use exhaustive search to enumerate all possible pre- and suffixes of
  length at most $c$. We then use a variation of the representative
  sets technique \cite{fomin} to obtain an \fpt algorithm.  In case
  (b), we show that it is possible to adapt the reduction in the proof
  of Theorem~\ref{theorem:two-disjoint-is-hard}.
\end{proof}

Notice that the difference between cuttable and non-cuttable classes
of STEs can be quite subtle.
For instance, $b^k a^*$ and $a^k (a+b)^*$ are non-cuttable, but $(a+b)^k a^*$ is cuttable.

\begin{figure}
\centering
\begin{tikzpicture}[every node/.style={fill, circle, inner sep = 1pt}]

\node (S) at (-2,0) {};
\node (K) at (2,0) {};
\node (L) at (5,0) {};
\node (T) at (8,0) {};

\draw[very thick] (S) -- (-1,0);
\draw[dashed] (-1,0)--(K);
\draw[dashed] (L) -- (6,0);
\draw[very thick] (6,0)--(T);
\draw[dashed] (K) to [out = 30, in = -50] (2,.8);
\draw[dashed] (2,.8) to [out = 130, in = 110] (-1,0);
\draw[very thick] (-1,0) to [out = -60, in = 210] (4,.5);\draw[very thick] (4,.5) to [out = 30, in = 110] (6,0);
\draw[dashed] (6,0) to [out = -70, in = -20] (5,-0.7);
\draw[dashed] (5,-0.7) to [out = 160, in = 200] (L);

\draw[<->] ($(S)+(0,.2)$) -- ($(K)+(0,.2)$);
\draw[<->] ($(L)+(0,.2)$) -- ($(T)+(0,.2)$);

\draw[<->] ($(S)-(0,.2)$) -- ($(-1,0)-(0,.2)$);
\draw (-1.4,0) coordinate[label={below:$\geq c_1$}];

\draw[<->] ($(6,0)-(0,.2)$) -- ($(T)-(0,.2)$);
\draw (7.0,0) coordinate[label={below:$\geq c_2$}];

\draw (S) coordinate[label={left:$s$}];
\draw (0,0.2) coordinate[label={above:$k_1$}];
\draw (6.5,0.2) coordinate[label={above:$k_2$}];
\draw (T) coordinate[label={right:$t$}];
\end{tikzpicture}
\vspace{-.8cm}
\caption{Assume $r = A_1\cdots
  A_{k_1} A^* A'_{k_2} \cdots A'_1$ has left and right cut borders
  $c_1$ and $c_2$, respectively. Assume that an arbitrary path
  from $s$ to $t$ matches $r$ such that its length $k_1$ prefix and
  length $k_2$ suffix are node disjoint. If, after removing all loops, (1)
  the length $c_1$ prefix and length $c_2$ suffix are still the same and (2) the
  path still has length at least $k_1+k_2$, then it matches $r$.}
\label{fig:cuttable}
\end{figure}
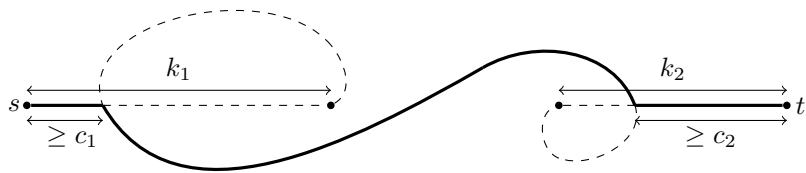

\paragraph*{Dichotomy for Trails}
We now present the dichotomy for trails. Perhaps
surprisingly, this dichotomy is slightly different. The underlying
reason is that $\edgespath(a^kb^*)$ is in FPT because the $a$-path and
the $b$-path can be evaluated independent of each other (no $a$-edge
will be equal to a $b$-edge). On the other
hand we have that $\edgespath(a^kba^*)$ is
W[1]-hard. 

Let $r = A_1 \cdots A_{k_1} A^* A'_{k_2} \cdots A'_1$ be an STE with left cut border $c_1$ and right cut border $c_2$. We say that $A_i$ with $i \leq c_1$ (resp., $A'_j$ with $j \leq c_2$) is a \emph{conflict position} if there exists a symbol $\sigma \in A_i \cap A$ (resp., $\sigma \in A_j \cap A$). We say that $\cR$ is \emph{almost conflict free} if there exists a constant $c$ such that each $r \in \cR$ has at most $c$ conflict positions.

We say that $\cR$ \emph{can be conflict-sampled} if there exists
an algorithm that, given $k \in \nat$, returns an 
expression in $\cR$ that has $k'$ conflict labels with $k' \geq k$, and ``no'' otherwise.  
\begin{restatable}{theorem}{edgeDichotomyThm}\label{theo:edgedichotomy}
  Let $\cR$ be a class of STEs that can be conflict-sampled. Then,
  \begin{enumerate}[(a)]
  \item if $\cR$ is almost conflict free, then $\edgespath(\cR)$ is in \fpt with
    parameter $k_r$ and 
  \item otherwise, $\edgespath(\cR)$ is W[1]-hard with
    parameter $k_r$.  \end{enumerate}
\end{restatable}

\makeatletter{}\section{Enumeration Problems for Simple Transitive Expressions}\label{sec:enumfpt}

We now observe that our tractability results can be carried over to the enumeration setting. To this end, a \emph{parameterized enumeration problem} is defined analogously as an enumeration problem, but its input is of the form $(x,k) \in \Sigma^* \times \nat$. It is in \emph{FPT delay} if there exists an algorithm that enumerates the output such that the time between two consecutive outputs is bounded by $f(k)\cdot |x|^c$ for a constant $c$. Notice that each problem in polynomial delay is also in FPT delay.

The problems in the following theorems are straightforward enumeration versions of problems we already considered. 
\begin{restatable}{theorem}{fptdelayi}\label{theo:fptdelay:1}
  $\enumnodespaths_{\geq k}$ is in FPT delay. 
\end{restatable}

\begin{restatable}{theorem}{fptdelayii}\label{theo:fptdelay:2}
  For each constant $c$ and each word $w$ with length $|w|=c$, the problem $\enumnodespaths(a^kw?a^*)$ is in FPT delay.
\end{restatable}

\begin{restatable}{theorem}{fptdelayiii}\label{theo:fptdelay:3}
  Let $\cR$ be a cuttable class of STEs. Then $\enumnodespaths(\cR)$ is in FPT delay.
\end{restatable}

The proofs of these theorems are all along the same lines. 
In the proof of Theorem~\ref{theorem:BaganPolyDelay} we adapted Yen's algorithm to work with simple instead of shortest paths. We already showed that the problems $\nodesimpath_{\geq k}$, $\nodesimpath(a^kw?a^*)$, and $\nodesimpath(\cR)$ are in FPT. Furthermore, these FPT algorithms can trivially be adjusted to also return a matching path if it exists. We also need to show that we can find simple paths matching \emph{suffixes}\footnote{More precisely, we need language derivatives, see Appendix~\ref{app:yen}.} in the language (for the adapted line~\ref{alg:yen:12} of Yen's algorithm in the proof of Theorem~\ref{theorem:BaganPolyDelay}). This can also be done for each of these theorems, essentially because the suffixes of the languages we need to consider again can be solved with our FPT algorithms.
In Appendix~\ref{app:enumfpt} we prove that this approach works.

Furthermore, we can also show that the FPT result from Theorem~\ref{theo:edgedichotomy} carries over to enumeration problems.
\begin{restatable}{theorem}{fptdelayedge}
	Let $\cR$ be a class of STE that is almost conflict-free. Then,
	$\enumedgespaths(\cR)$ is in FPT delay.
\end{restatable}

\makeatletter{}\section{Conclusions}

Our main results are two dichotomies on the parameterized complexity of
evaluating \emph{simple transitive expressions (STEs)}, which are a
class of regular expressions powerful enough to capture over 99\% of
the RPQs occurring in a recent practical study
\cite{BonifatiMT-corr17}. These dichotomies are for \emph{simple path
  semantics} and \emph{trail semantics}, respectively.

For \emph{simple path semantics}, the central property that we require
for a class of expressions so that evaluation is in FPT is
\emph{cuttability}, i.e., constant-size \emph{cut borders} (also see
Figure~\ref{fig:cuttable}).  

For \emph{trail semantics}, the dichotomy is such that the FPT
fragment is slightly larger. Even if the cut borders of a class of
expressions is not bounded
by a constant, it can be evaluated in FPT if the number of
\emph{conflict positions} are bounded by a constant. An example of a
non-cuttable class of expressions with a constant number of conflict
positions is $\{a^kb^* \mid k \in \nat\}$. For this class, evaluation
over trail semantics is in FPT (with parameter $k$) but W[1]-hard over
simple path semantics.

Looking at Table~\ref{fig:propertypaths}, we see that the cut borders
for expressions in practice are indeed very small: it is one for $a^*
b$, two for $a b c^*$, and zero in all other cases. All these
expressions have FPT evaluation for simple path and trail semantics.
Therefore, although the simple path and trail semantics
of RPQs are known to be hard in general, it seems that the RPQs that
users actually ask are much less harmful. In fact, since the vast
majority of expressions in Table~\ref{fig:propertypaths} has cut
borders of at most two, our FPT result in Theorem~\ref{theo:dichotomy}
implies that evaluation for this majority of expressions is in
polynomial time combined complexity. Furthermore, matching paths can
be enumerated in polynomial delay. (Recall that, if P $\neq$ NP, this
is impossible even for fixed expressions: evaluation for $a^*ba^*$ or
$(aa)^*$ under simple path semantics is \np-complete.)

For the expressions in Table~\ref{fig:propertypaths}, the parameter
$k_r$ is at most six. Since the function $f$ in our FPT algorithms is only
single exponential, we believe that these expressions can be dealt
with in practical scenarios, in principle. The data complexity of our
FPT algorithms is currently $O(m n \log n + n^2 + mn)$ with $m = |E|$ and $n = |V|$. This 
bound comes from Fomin et al.'s representative set technique
\cite{fomin} and we did not yet investigate yet if this can be
improved. We believe that this would be an interesting future direction.

\section*{Acknowledgments}

We are grateful to Phokion Kolaitis for suggesting us to study 
enumeration problems on simple paths matching RPQs. We
are also grateful to Holger Dell for pointing us to
Theorem~\ref{theorem:longdirpathInFPT} and providing us with a proof sketch.

\bibliographystyle{abbrv}
\bibliography{references}

\onecolumn
\newpage

\appendix
In the Appendix we provide proofs for which there was no space in the
body of the paper. In some proofs, we indicate by \fbox{$\cdots$}
where we continue a proof that was partly presented in the body.

\makeatletter{}\section{Proofs for Section~\ref{sec:preliminaries}}

\splitgraphlemma*
\begin{proof}
  This is an easy consequence of the construction.
  \end{proof}

\linegraphlemma*
\begin{proof}
  This is an easy consequence of the construction.
  \end{proof}

\makeatletter{}

\noindent \textbf{Notation.} In the appendix, we sometimes use
\emph{$u$-$v$-path} to refer to a \emph{path from $u$ to $v$}.

\section{Proofs for Section~\ref{sec:enum-yen}}\label{app:enum-yen}
\propprefixclosed*
\begin{proof}
  Assume that $(G,s,t)$ and $r$ is an input for \enumnodespaths such that
  $L(r)$ is downward closed. Let $N = (Q,\Sigma,\delta, Q_I, Q_F)$ be
  an NFA for $r$.  We change Algorithm~\ref{alg:yen} as follows:
  \begin{itemize}
  \item In line~\ref{alg:yen:3}, instead of finding a shortest path
    $p$ in $G$, we first find a shortest path $p$ in $(G,s,t) \times
    N$. We then replace every node of the form $(u,q) \in V\times Q$
    in $p$ by $u$.
  \item In line~\ref{alg:yen:12} we need to find a shortest path in a
    product between $(G',p[i,i],t)$ and $N$. More precisely, let $J =
    \delta^*(\lab(p[0,i]))$ and denote by $N_J$ the NFA with initial
    state set $J$, that is, $(Q,\Sigma,\delta,J,Q_F)$. Then, in line
    \ref{alg:yen:12} we first find a shortest path $p_2$ from any node
    in $\{(p[i,i],q_i) \mid q_i \in \delta^*(\lab(p[0,i]))\}$ to any
    node in $\{(t,q_F) \mid q_F \in Q_F\}$ in $(G',p[i,i],t) \times
    N_J$. We then replace every node of the form $(u,q) \in V\times Q$
    in $p_2$ by $u$.
  \end{itemize}

  \fbox{$\cdots$} We now prove that this leads to a polynomial delay algorithm for
  \enumnodespaths.  As the \productautomaton can be constructed in time
  $O(|G||N|)$, the algorithm still runs in polynomial delay. 

  To prove that this algorithm is correct, we first show that no path
  is written to the output more than once: Each such path is stored in
  $A$ and cannot be found again, because the prefix $p[0,i]$ differs
  or at least one edge will be deleted in line~\ref{alg:yen:9}.

  We now prove that the algorithm only writes simple paths that match
  $r$ to the output. Each shortest path $p_2$ considered in the
  \productautomaton $(G',p[i,i],t) \times N$ in line~\ref{alg:yen:12} is, after
  replacing nodes $(u,q)$ with $u$, a simple path in $G$, because $L$
  is downward closed.
        Therefore, since $p_2$ is disjoint from $V(p[0,i-1])$ due to
  line~\ref{alg:yen:8} of the algorithm, $p[0,i]\cdot p_2$ is also a simple
  path that matches $r$.

  Finally, we prove that the algorithm finds all such simple paths.
      If a simple path $p$ in $(G,s,t)$ matches $r$, then this path is also a
  simple path in $(G,s,t) \times N$. So, we can find this path using the
  changed algorithm if and only if we do not delete any edge from $p$
  in $G$, which is only done in Line~\ref{alg:yen:9}.  But we did
  not change this line, so it follows from the correctness of Yen's
  algorithm that $p$ can be found.
\end{proof}

\trailtopathfptreduction*
\begin{proof}
			Given a graph $(G,s,t)$, we will construct a graph $(H,s',t')$ such
	that there exists a simple path from $s'$ to $t'$ matching $a r$ in
	$H$ if and only if there exists a trail from $s$ to $t$ matching $r$
	in $G$, where $a$ is an arbitrary symbol. 
	Excluding $s'$ and $t'$, the graph $H$ is the line graph of $G$. 	We can then enumerate all possible $a$-edges that start in $s'$ to obtain up to $n$ new instances $(H_1,s'_1,t'), \ldots (H_n,s'_n,t')$, such that there exists a trail from $s$ to $t$ matching $r$ in $G$ if and only if there exists an $i$ such that there is a simple path from $s'_i$ to $t'$ in $H_i$ that matches $r$.
	
		So it remains to give the construction of $H$ and prove the correctness of the reduction. Let $a \in \Sigma$ be fixed. 
	Let $H=(V',E')$ with $V' = \{v_e \mid e \in E\}\cup \{s',t'\}$ and $E' = \{(v_{(u_1,\sigma_1,u_2)},\sigma_1,v_{(u_2,\sigma_2,u_3)})\mid u_1, u_2, u_3 \in V\} \cup \{(s',a ,v_{(s,\sigma,u)}), (v_{(u,\sigma,t)},\sigma,t') \}$. 
		An example of this reduction can be seen in Figure~\ref{fig:edgeToNodes}.
	We will now show the correctness of the reduction. Assume there exists a path $$p = (s,a_0,v_1)(v_1,a_1, v_2)\cdots (v_k,a_{k},t)$$ from $s$ to $t$ in $G$ that matches $r$ and has pairwise disjoint edges. Then the path $$p' = (s',a,v_{(s,a_0, v_1)})(v_{(s,a_0, v_1)},a_0,v_{(v_1,a_1, v_2)})(v_{(v_1,a_1, v_2)},a_1,v_{v_2,a_2,v_3})\cdots (v_{(v_k,a_k,t)},a_{k},t')$$ is a simple path from $s'$ to $t'$ in $H$ that matches $a r$. The other direction follows analogous since each path from $s'$ to $t'$ in $H$ that matches $a r$ has this form and we can therefore find the corresponding path from $s$ to $t$ in $G$.
\end{proof}

We note that, in the above proof there is a clear correspondence
between nodes in $H_i$ and edges in $G$.
\begin{corollary}
  Furthermore, each node in $H_i$, except for $s_i$ and $t_i$, corresponds to exactly one edge in $G$. 
\end{corollary}

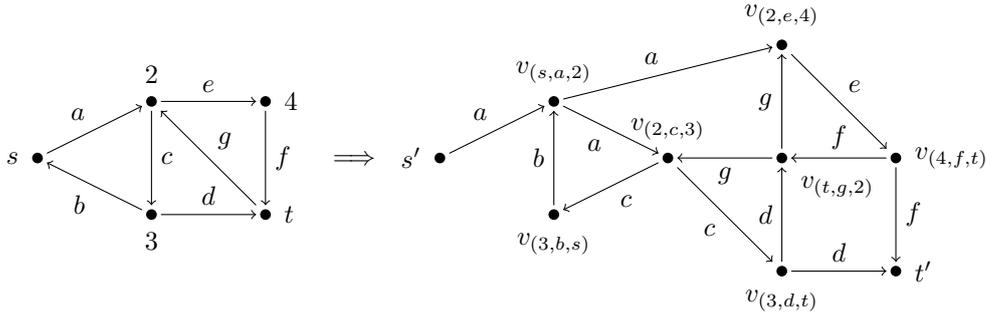
\begin{figure}[t]
	\centering
	\begin{minipage}{.3\linewidth}
		\begin{tikzpicture}[->, scale = 1, auto ,b/.style={fill, circle, inner sep = 1pt}]
		\def\stretch{1.5};
		\node [label=left:$s$] (v1) at (\stretch*0,\stretch*0) {};
		\node [label=above:$2$] (v2) at (\stretch*1,\stretch*0.5) {};
		\node [label=below:$3$] (v3) at (\stretch*1,\stretch*-0.5) {};
		\node [label=right:$4$] (v4) at (\stretch*2,\stretch*0.5) {};
		\node [label=right:$t$] (v5) at (\stretch*2,\stretch*-.5) {};
		\foreach \i in {1,2,3,4,5}{
			\fill (v\i) circle (2pt);
		}
		\path 
		(v1) edge	[->]  node  {$a$} (v2)
		(v3) edge	 [->]    node  {$b$} (v1) 
		(v2) edge	 [->]    node  {$c$} (v3)
		(v3) edge	 [->]    node  {$d$} (v5)
		(v2) edge	 [->]    node  {$e$} (v4)
		(v4) edge	 [->]    node  {$f$} (v5)
		(v5) edge	 [->,swap]    node  {$g$} (v2) 
		;
		\end{tikzpicture}
	\end{minipage} 
	$\implies$
	\begin{minipage}{.58\linewidth}
		\begin{tikzpicture}[->, scale = 1, auto ,b/.style={fill, circle, inner sep = 1pt}]
		\def\stretch{3};
		\node [label=left:$s'$] (v1) at (\stretch*0,\stretch*0) {};
		\node (v2) at (\stretch*1,\stretch*0.5) {};
		\node (v3) at (\stretch*1,\stretch*-0.5) {};
		\node (v4) at (\stretch*2,\stretch*0.5) {};
		\node [label=right:$t'$] (v5) at (\stretch*2,\stretch*-.5) {};
		\node [label=above:$v_{(s,a,2)}$] (e1) at ($(v1)!0.5!(v2)$) {};
		\node [label=below:$v_{(3,b,s)}$] (e2) at ($(v3)!0.5!(v1)$) {};
		\node [label=above:$v_{(2,c,3)}$] (e3) at ($(v2)!0.5!(v3)$) {};
		\node [label=below:$v_{(3,d,t)}$] (e4) at ($(v3)!0.5!(v5)$) {};
		\node [label=above:$v_{(2,e,4)}$] (e5) at ($(v2)!0.5!(v4)$) {};
		\node [label=right:$v_{(4,f,t)}$] (e6) at ($(v4)!0.5!(v5)$) {};
		\node [label=below right:$v_{(t,g,2)}$] (e7) at ($(v5)!0.5!(v2)$) {};
		\foreach \i in {1,2,3,4,5,6,7}{
			\fill (e\i) circle (2pt);
		}
		\fill (v1) circle (2pt);
		\fill (v5) circle (2pt);
		
		\path 
		(v1) edge	[->]  node  {$a$} (e1)
		(e2) edge	 [->]    node  {$b$} (e1) 
		(e1) edge	 [->,swap]    node  {$a$} (e3)
		(e3) edge	 [->]    node  {$c$} (e2)
		(e3) edge	 [->,swap]    node  {$c$} (e4)
		(e4) edge	 [->]    node  {$d$} (v5)
		(e5) edge	 [->]    node  {$e$} (e6) 
		(e6) edge	 [->]    node  {$f$} (v5)
		(e4) edge	 [->]    node  {$d$} (e7)
		(e6) edge	 [->,swap]    node  {$f$} (e7)
		(e7) edge	 [->]    node  {$g$} (e3) 
		(e1) edge	 [->]    node  {$a$} (e5)
		(e7) edge	 [->]    node  {$g$} (e5) 
		;
		\end{tikzpicture}
	\end{minipage}
	\caption{Example of a part of the reduction in Lemma~\ref{lemma:edgeToNodes}. There exists a trail from $s$ to $t$ matching $r$ in the left graph if and only if there exists a simple path from $s'$ to $t'$ matching $a \cdot r$ in the right graph.}
	\label{fig:edgeToNodes}
\end{figure}

\subsection{Proofs for Section \ref{sec:BeyondDWC}}  \label{app:yen}

In the following proof, we need \emph{language derivatives}.  For a
language $L$ and word $w$, the \emph{left derivative}\footnote{These
  are sometimes also called Brzozowski derivatives of $L$
  \cite{BrzozowskiJ-JACM64}.}  \mbox{w.r.t.} $w$, denoted $w^{-1} L$,
is defined as $\{v \mid wv \in L\}$. 

\BaganPolyDelay*
\begin{proof}
  Part (b) immediately follows from \cite[Theorem 1]{bagan}. It
  therefore only remains to prove (a). Our plan is to use Bagan et
  al.'s algorithm for simple paths as a subroutine in Yen's
  algorithm. We refer to Bagan et al.'s algorithm as the BBG
  algorithm.

  So, we adapt Yen's algorithm by calling BBG in lines
  \ref{alg:yen:3} and \ref{alg:yen:12}, so that the algorithm
  receives
  \begin{enumerate}[(i)]
  \item a simple path from $s$ to $t$ that matches $r$ in line
    \ref{alg:yen:3} and
  \item a simple path $p_2$ from $p[i,i]$ to $t$ such that
    $p[0,i]\cdot p_2$ matches $r$ in line \ref{alg:yen:12},
  \end{enumerate}
  respectively. (We do not need to change line \ref{alg:yen:12}.)
  We refer to the adapted algorithm as YenSimple.

  Change (i) to Yen's algorithm is trivial. Change (ii) can be done by
  calling BBG with $G'$ for the language of the automaton $N_J$ in the
  proof of Proposition~\ref{prop:prefixclosed}.

  \fbox{$\cdots$}
  We show that YenSimple is correct. First of all, notice
  that the algorithm still does not output duplicates, because the for-loop
  in line~\ref{alg:yen:9} deletes at least one edge of all paths that
  are already in $A$.
	
  Second, we show that the algorithm writes all simple paths matching
  $r$ into the output.  Therefore, 
      let $\pi$ be a simple path from $s$ to $t$ in $G$ that
  matches $r$.
    Due to the correctness of
  the BBG algorithm, YenSimple has found a simple path matching $r$ in
  line~\ref{alg:yen:3} and therefore $A\neq \emptyset$. 
  We now consider an arbitrary iteration of the while-loop and prove
  that either $\pi$ must have been found already or YenSimple will find
  a path $\tilde{\pi}$ that shares a longer prefix with $\pi$ than all
  paths in $A$. Clearly, this shows that $\pi$ will eventually be found.
  
  Assume that we are at the beginning of the while-loop and let $S$ be
  the set of paths in $A$ that share the longest prefix with $\pi$,
  that is, $S =\{\pi' \in A \mid \exists 0 \leq i \leq |V|:
  \pi'[0,i]=\pi[0,i]$ and there exists no path $\tilde{\pi}\neq \pi'$
  in $A$ with $\tilde{\pi}[0,i+1]=\pi[0,i+1]\}$. Since $A \neq
  \emptyset$ and $A$ only contains paths from $s$ to $t$ and
  $\pi[0,0]=s$, i.e., they share at least the first node, $S$ is not
  empty. Take $\pi' \in S$ such that $\pi'$ was the last element in
  $S$ that was added to $A$ (and therefore written to the output).  As
  $\pi'$ and $\pi$ are both simple paths from $s$ to $t$ and $\pi$ was
  not yet written to the output, $\pi'$ must have at least one edge
  that is not in $\pi$.

  After having added $\pi'$ to $A$, YenSimple searched, for each $i$
  from $1$ to $|\pi'|$, a simple path having $\pi'[0,i-1]$ as prefix, but
  not having the edge $(\pi'[i-1,i])$. Let $i'$ be maximal with
  $\pi'[0,i'-1]=\pi[0,i'-1]$.  Notice that the edge $(\pi'[i'-1,i'])$ was
  not deleted in line~\ref{alg:yen:9}, because otherwise there would
  have been a path $\tilde{\pi} \in A$ with $\tilde{\pi}[0,i'] = \pi[0,i']$,
  contradicting $\pi' \in S$.  So, if $\pi$ is a simple path from $s$ to $t$, $\pi$
  must have been found already or the algorithm will find a
  simple path $\tilde{\pi}$ with $\tilde{\pi}[0,i] = \pi[0,i]$ that is
  not in $A$. This concludes the proof that all simple paths from $s$ to $t$ that
  match $r$ will be found by YenSimple.

  Before we turn to complexity, we need a simple observation. 
  \begin{observation}\label{obs:ctract-derivative}
    \ctract is closed under taking left derivatives, that is, if $L
    \in \ctract$ and $w$ is a word, then $w^{-1} L \in \ctract$.
  \end{observation}
  \begin{proof}
    The observation immediately follows from the definition of
    \ctract. Indeed, if $w^{-1}L \notin \ctract$, then, for every $i \in
    \nat$ there exist words $w_\ell, w,w_r \in \Sigma^*, w_1, w_2 \in
    \Sigma^+$, such that $w_\ell w_1^i w w_2^i w_r \in w^{-1}L$ but
    $w_\ell w_1^i w_2^i w_r \notin w^{-1}L$. However, then we also
    have for every $i \in \nat$ that
    $w w_\ell w_1^i w w_2^i w_r \in L$ and $ww_\ell w_1^i w_2^i
    w_r \notin L$, which contradicts that $L \in \ctract$.
  \end{proof}

  We now turn to complexity. In terms of data complexity, the time
  bound of Yen's algorithm (i.e., polynomial delay) is not affected by
  searching simple instead of shortest paths, since BBG operates in
  polynomial time for $L(r)$, which is in \ctract \cite[Lemma
  16]{bagan}. The same holds for $L(N_J)$, which is in \ctract due to
  Observation~\ref{obs:ctract-derivative}. This concludes the proof.
\end{proof}

The algorithm for Theorem~\ref{theorem:BaganPolyDelay}(a) does not yield any order on
the paths.  But an order from shortest to longest paths can easily be
obtained by changing BBG to output a shortest path. As Bagan et al.\
already note \cite[Section 3.2]{bagan}, this is indeed a simple change
in their algorithm.\footnote{See also Lemma~\ref{BaganPolyDelayOrder} in Appendix~\ref{app:yen}.} 
Moreover, it is also possible to use Ackerman and Shallit's algorithm \cite{Ackerman-TCS09} 
for finding shortest and lexicographically smallest paths in the BBG
algorithm. It is therefore also possible to enumerate the paths in
radix order in polynomial delay.

\begin{lemma}\label{BaganPolyDelayOrder}
   If $L(r) \in \ctract$, then 
   in terms of data complexity, $\enumspaths(r)$ can be solved in
   polynomial delay and with all paths enumerated in radix order.
\end{lemma}
\begin{proof}
Instead of showing that Yen's algorithm also works with simple paths, as we did in the proof of Theorem \ref{theorem:BaganPolyDelay} part a), we slightly
change the algorithm in Bagan et al.'s paper
  \cite{bagan} to find a shortest and lexicographically smallest simple path. Then we can use this algorithm as subroutine in Yen's algorithm.

To this end, we change the algorithm introduced in \cite[Lemma 15]{bagan}. There, in the second step, we replace $(\text{left}_i,\text{cut}_{C_i}, \text{right}_i)$ with a smallest simple $\text{Set}_i$-restricted $\Sigma^*_{C_i}$ path in radix order from $\text{left}_i$ to $\text{right}_i$ (i.e., we require additionally that it is a lexicographically smallest path in radix order). We can find such a path using the algorithm of Ackerman and Shallit~\cite{Ackerman-TCS09}, see also Theorem~\ref{theo:ackermanAndShallit}. (We view the subgraph of $G$ that contains only $\text{Set}_i$-restricted nodes (and $\text{right}_i$) and $\Sigma_{C_i}$-labeled edges between these nodes as NFA with start state $\text{left}_i$ and final state $\text{right}_i$.) Since shortest simple paths are always admissible \cite[Lemma 13]{bagan}, we are indeed able to find a smallest simple path in radix order in the following way: 
Just like Bagan et al., we enumerate all possible candidate summaries $S$ w.r.t.\ $(L(r),G,s,t)$ and then apply to each the adapted algorithm from \cite[Lemma 15]{bagan}. If we find a solution, we do not return `yes' immediately, but continue the enumeration while holding the smallest solution w.r.t. radix order in our storage and update it when necessary.
 
As the algorithm of Ackerman and Shallit~\cite{Ackerman-TCS09} runs in polynomial time, the adapted algorithm \cite[Lemma 15]{bagan} still runs in polynomial time. Since $w^{-1}L(r) \in \ctract$ (see Observation \ref{obs:ctract-derivative}), we can use this algorithm as subroutine in Yen's algorithm in the lines \ref{alg:yen:3} and \ref{alg:yen:12} to obtain a polynomial delay algorithm that enumerates all paths in radix order.
\end{proof}

\bagantrails*
\begin{proof}
  Part (a) follows directly from Lemma~\ref{lemma:edgeToNodes} and the
  upper bound of Bagan et al.~\cite[Theorem 2]{bagan}.

  It remains to show (b). The upper bound again follows from
  Lemma~\ref{lemma:edgeToNodes}. The hardness is similar to
  \cite[Lemma 2]{bagan}, as we show next.

  \fbox{$\cdots$}  Let $L(r) \notin \ctract$. We exhibit a reduction from the
  \twodisjointpaths problem, that is: Given a graph $G = (V,E)$ and
  nodes $s_1,t_1,s_2,t_2$. Is there a pair of node-disjoint simple
  paths in $G$, one from $s_1$ to $s_2$ and the other from $s_2$ to
  $t_2$?  This problem is NP-complete \cite{FortuneHW-TCS80} and can
  be transformed into the corresponding trail problem using
  \cite[Lemma 1]{LapaughR-jcss80}.                         	
	Let $N= (Q,\Sigma,\Delta,Q_I, Q_F)$ be an NFA with $Q_I=s_N$ and $L(N) = L(r)$. 
	For this proof we need a different definition of
        \ctract. A language $L(r) \notin \ctract$ if and only if there
        exists a \emph{witness for hardness} $(q, w_\text{m},
        w\text{r}, w_1, w_2)$ where $q \in Q, w_\text{r} \in \Sigma^*,
        w_1 \in Loop(w_1)$, and $w_2, w_\text{m} \in \Sigma^+$
        satisfying $w_\text{m} w_2^* w_r \subseteq L_q$ and
        $(w_1+w_2)^*w_\text{r} \cap L_q = \emptyset$. 
              This definition is equivalent to Definition~\ref{def:loopabbrev} , see \cite[Definition 1, Theorem 4]{bagan}
	
	Since $L(r) \notin \ctract$, there exists a witness for
        hardness $(q, w_\text{r}, w_\text{m},w_1,w_2)$ \cite[Lemma
        1]{bagan}. Let $w_\ell$ be a word such that
        $\Delta^*(w_\ell) = q$.  By definition we have
	$w_\ell(w_1 + w_2)^*w_\text{r} \cap L = \emptyset$  and $w_\ell w_1^*w_\text{m}w_2^*w_\text{r} \subseteq L(r)$. 	
	We build from $G$ a generalized graph $G'$ whose edges are
        labeled by non empty words. The generalized graph $G'$ can
        easily be turned into a graph by adding intermediate nodes,
        replacing an edge labeled by word $w$ by a path whose edges form the word $w$.
	The graph $G'$ is constructed as follows. The nodes of $G'$ are the same as the nodes of $G$. For each edge $(v_1,v_2$) in G, we add two edges $(v1, w_1, v2)$ and $(v1, w_2, v2)$. Moreover, we add two new nodes $x, y$ and three edges $(x, w_\ell, x_1), (y_1, w_\text{m}, x_2)$, and $(y_2,wr,y)$.
	By construction, for every trail $p$ from $x$ to $y$ in $G'$
        that contains the edge $(y_1 , w_\text{m} , x_2 )$, we can
        obtain a similar path that matches a word in $w_\ell
        w_1^*w_\text{m}w_2^*w_\text{r}$ by switching $w_1$ and $w_2$
        edges, keeping the same nodes. Every trail $p$ from $x$ to $y$
        in $G$ that does not contain the edge $(y_1,w_\text{m},x_2)$
        matches a word in $w_\ell(w_1 + w_2)^*w_\text{r}$. By
        definition of $q \in Q, w_\text{r} \in \Sigma^*$,
        $w_\text{m},w_1,w_2 \in \Sigma^+$, no path of the form
        $w_\ell(w_1 + w_2)^*w_\text{r}$ matches $r$, whereas every
        path matching $w_\ell w_1^*w_\text{m}w_2^*w_\text{r}$
        automatically matches $r$. Thus, $\edgespath(r)$ returns
        ``yes'' for $(G',x,y)$ iff there is a trail from $x$ to $y$ in
        $G'$ that contains the edge $(y_1 , w_\text{m} , x_2 )$ that
        is, iff Two-Edge-Disjoint-Path returns ``yes'' for $(G,x_1,
        y_1,x_2,y_2)$. 
\end{proof}

\makeatletter{}
\section{Proofs for Section \ref{sec:ParamCompl}}\label{app:paramcompl}

\section*{Proofs for Section \ref{sec:onepath}}

We present how Theorem \ref{theorem:longdirpathInFPT} can be proved,
following the explanation we received from Holger Dell
\cite{holgercomm}. (To the best of our knowledge, the result and proof should be
attributed to the authors of \cite{fomin}.)  We first need some
terminology. The following is Definition 3.1 from \cite{fomin}, which
we rephrased from matroids to graphs to simplify presentation. \begin{definition}[$k$-representative family \cite{fomin}] 
  Given a graph $G=(V,E)$ and a family $\cS$ of subsets of $V$, and $k
  \in \nat$, we
    say that a
  subfamily $\hat{\cS} \subseteq \cS$ is \emph{$k$-representative for
    $\cS$} if the following holds: for every set $Y \subseteq V$ of
  size at most $k$, if there is a set $X \in \cS$ disjoint from $Y$
  with $|X \cup Y| \leq 2k$, then there is a set $\hat{X} \in \hat{\cS}$
  disjoint from $Y$ with $|\hat{X} \cup Y|\leq 2k$.  We abbreviate this
  by $\hat{\cS} \subseteq^k_\text{rep} \cS$.
\end{definition}
In the following we define 
\begin{multline*}
  P^k_{sv} = \{X \mid X\subseteq V \text{ such that } s,v \in X,|X| = k, 
   \text{ and there is a path from $s$ to $v$ in } G \\
   \text{ of length } k-1 \text{ containing exactly the nodes in } X\}
\end{multline*}
The following Lemma shows that, in order to find a solution for
\longdirpath, it suffices to consider paths where the first $k+1$
nodes belong to a set in $\hat{P}^{k+1}_{sv} \subseteq^{k+1}_\text{rep}
P^{k+1}_{sv}$. This lemma and its proof are analogous to Lemma 5.1 in
\cite{fomin}. We state it here because it is bordered for the
correctness of Algorithm~\ref{fptalg} and we need to adapt this proof
to show the correctness of Algorithms~\ref{b?fptalg} and
\ref{alg:block}, which build upon
Algorithm~\ref{fptalg}. 
\begin{lemma} \label{proofforfptalg} Assume that $\hat{P}^{k+1}_{sv}
  \subseteq^{k+1}_\text{rep} P^{k+1}_{sv}$. Then a graph $G=(V,E)$ has a
  simple $s$-$t$-path of length at least $k$ if and only if there
  exists a node $v \in V$ and $X \in \hat{P}^{k+1}_{sv}
  \subseteq^{k+1}_\text{rep} P^{k+1}_{sv}$, such that $G$ has a simple
  $s$-$t$-path of length at least $k$ with the first $k+1$ nodes
  belonging to $X$.
\end{lemma}
\begin{proof}
  The only-if direction is straightforward: if $G$ has a simple path
  $p$ whose first $k+1$ nodes belong to $X$, then its length is at
  least $k$.  For the other direction, let $p=(v_0,v_1) \cdots
  (v_{r-1},v_r)$ be a shortest $s$-$t$-path of length at least $k$,
  i.e., such that $r \geq k$.
 
  If $|p|=r \leq 2k+1$,   we define $P= (v_0,v_1) \cdots (v_{k-1},v_k)$ and $Q=
  (v_{k+1},v_{k+2}) \cdots \allowbreak (v_{r-1},v_r)$. Because $|V(Q)|
  \leq k+1$, $V(P) \in P^{k+1}_{sv_k}$ and $V(P)\cap V(Q)= \emptyset$,
  the definition of $\hat{P}^{k+1}_{sv_k} \subseteq^{k+1}_\text{rep}
  P^{k+1}_{sv_k}$ guarantees the existence of an $s$-$v_k$-path $P'$
  with $V(P') \in \hat{P}^{k+1}_{sv_k}$ and $V(P') \cap V(Q) =
  \emptyset$. By replacing $P$ with $P'$ in $p$, we obtain a simple
  $s$-$t$-path of length $|p|$.

  Otherwise, we have that $|p| = r > 2k+1$.  Then we define $P=
  (v_0,v_1) \cdots (v_{k-1},v_k)$, $R = (v_{k+1},v_{k+2}) \cdots
  \allowbreak (v_{r-k-2},v_{r-k-1})$, and $Q= (v_{r-k},v_{r-k+1})
  \cdots \allowbreak (v_{r-1},v_r)$, so we have $p = P\cdot (v_k,v_{k+1})\cdot R\cdot (v_{r-k-1},v_{r-k})\cdot Q.$  (We wrote some of the
  concatenation operators $\cdot$ explicitly to improve readability.)
    Since $|V(Q)| = k+1$, $V(P) \in
  P^{k+1}_{sv_k}$, and $V(P)\cap V(Q)= \emptyset$, the definition of
  $\hat{P}^{k+1}_{sv_k} \subseteq^{k+1}_\text{rep} P^{k+1}_{sv_k}$
  guarantees the existence of an $s$-$v_k$-path $P'= (v_0,v'_1) \cdots
  (v'_{k-1},v_k)$ with $V(P') \in \hat{P}^{k+1}_{sv_k}$ and $V(P')
  \cap V(Q) = \emptyset$.  If $P'$ is disjoint from $R$, the path
  $$p' = (v_0,v'_1) \cdots (v'_{k-1},v_k)(v_k,v_{k+1}) \cdots
  (v_{r-1},v_r)$$ is a simple path of length $r$, so we are done.
      
  We show that $P'$ must be disjoint from $R$. Towards a
  contradiction, assume that there is an $i \in \{1, \ldots, k-1\}$
  such that $v'_i=v_j \in R$. We choose $i$ minimal and build a new
  simple path $p' = (v_0,v'_1) \cdots
  (v'_{i},v_{j+1})(v_{j+1},v_{j+2}) \cdots (v_{r-1},v_r)$ with $|p'|
  \geq k$, because it contains $Q$. But $V(p')$ does not contain
  $v_{k}$, so $p'$ is shorter than $p$, which contradicts that $p$ was
  a shortest $s$-$t$-path of length at least $k$.  So
  $P'$ must be disjoint from
  $R$. \end{proof}

We still need to show that $\hat{P}^{k+1}_{sv} \subseteq^{k+1}_\text{rep}
P^{k+1}_{sv}$ exists.  This is done by Lemma 5.2 in \cite{fomin},
which also restricts the size of $\hat{P}^k_{sv}$ and gives an upper
bound for its computation time. Since their Lemma 5.2 in \cite{fomin}
is 
more general than we need, we give more concrete bounds here that come
from the proof of Theorem 5.3 in \cite{fomin}.
\begin{lemma}
\label{lemmaSizeBound}
A collection of families $\hat{P}^k_{sv} \subseteq^k_\text{rep} P^k_{sv}, v \in V\setminus \{s\}$ of size at most $\binom{2k}{k} \cdot 2^{o(2k)}$ each can be found in time 
$$O\left( 8^{k+o(k)}m \log n \right),$$ 
where $n = |V|$ and $m = |E|$.
\end{lemma}

\begin{algorithm}[t]
\caption{FLPS Algorithm} \begin{algorithmic}[1]
\Require Graph $G=(V,E)$, nodes $s,t$ in $G$, parameter $k$ 
\Ensure Decide if there exists a simple path from $s$ to $t$ with length at least $k$
\For {every $v \in V$} 
\State Compute $\hat{P}^{k+1}_{sv} \subseteq^{k+1}_\text{rep} P^{k+1}_{sv}$\label{fptalg:2}
\For {every $X \in \hat{P}^{k+1}_{sv}$}\label{fptalg:3}
\State $V' \gets (V\setminus X)\cup \{v\}$
\State $E' \gets E \cap (V' \times V')$
\If {there exists a path from $v$ to $t$ in $(V',E')$}\label{fptalg:4} 
\State return YES
\EndIf
\EndFor
\EndFor
\State return NO
\end{algorithmic}
\label{fptalg}
\end{algorithm}

From Lemmas~\ref{proofforfptalg} and \ref{lemmaSizeBound}, we can
infer that Algorithm~\ref{fptalg} correctly solves
\longdirpath in FPT.

\LongDirPathFPT*
\begin{proof}
  The problem can be solved using Algorithm \ref{fptalg}. Its
  correctness follows directly from Lemma \ref{proofforfptalg}. Using
  Lemma \ref{lemmaSizeBound}, we 
  now show that the algorithm is indeed a \fpt algorithm.

  Let $n = |V|$ and $m = |E|$.  We obtain from
  Lemma~\ref{lemmaSizeBound} that line~\ref{fptalg:2} of the
  Algorithm~\ref{fptalg} takes $O\left( 8^{k+o(k)}m \log n \right)$
  time for each $v \in V$. Since we need to consider at most $n \cdot
  \binom{2(k+1)}{{k+1}} \cdot 2^{o(2(k+1))} $ sets $X$ in
  line~\ref{fptalg:3}, the number of such sets we need to consider
  throughout the entire algorithm is at most
  $O(n4^{k+o(k)})$. Finally, line~\ref{fptalg:4} can be checked by
  a reachability test (say, depth-first search) in time $O(m+n)$, so the overall running time is bounded
  by $$O\left( 8^{k+o(k)}mn \log n + 4^{k+o(k)}\cdot(n^2+mn)
  \right),$$ which is clearly in \fpt for the parameter $k$.
\end{proof}

\subsection*{Proofs for Section~\ref{sec:twodisjointpaths}}\label{app:twodisjointpaths}

\TheoTwoColor*
\begin{proof}

   We now prove that the reduction is correct, that is, $G$ has a
   $k$-clique iff there are simple paths $p_1$ from
   $c_1$ to $c_{k+1}$ and $p_2$ from $r_1$ to $r_{k+1}$ such that $p_1$ and $p_2$ are node-disjoint, $p_1$ is colored $a$ and has length $k'$ while $p_2$ is colored $b$. We need
   the following Lemma. Let $G'_a$ denote the subgraph obtained from $G'$ 
  by removing the $b$-edges and the nodes $r_1,\ldots,r_{k+1}$ (which
  have no adjacent $a$-edges).
  Then $G'_a$ has the following properties:\footnote{We only need part
    (\ref{lem:twopaths:b}) in this proof. Parts (\ref{lem:twopaths:c})
    and (\ref{lem:twopaths:d}) are used to prove Theorem \ref{theorem:two-disjoint-is-hard}.}
   \begin{restatable}{lemma}{twopathslemma} \label{lem:twopaths}
     $G'_a$ has the following properties:
     \begin{enumerate}[(a)]
                    \item Each path in $G'_a$ has length exactly $k'$ if and only if it
       is from $c_1$ to $c_{k+1}$. \label{lem:twopaths:b}
       
     \item Each path in $G'_a$ of length $k'$ visits all control
       nodes, i.e., it contains all $c_{i}$ and $c_{i_1i_2}$, with $i \in
       \{1,\ldots, k+1\}$ and $1\leq i_1<i_2 \leq k$. \label{lem:twopaths:c}
     \item Each path in $G'_a$ of length $k'$ has at least one edge
       $u_\ell \stackrel{a}{\to} v_\ell$ in  every row of $G'_a$. \label{lem:twopaths:d}
     \end{enumerate}
   \end{restatable}
   We prove the lemma after the present proof.

   Let us first assume that the undirected graph $G$ has a $k$-clique
   with nodes $\{n_1, \ldots, n_k\}$. Then an $a$-path can go
   from $c_1$ to $c_{k+1}$ using only the gadgets $G_{i,n_i}$ with $i = 1,
   \ldots, k$. The reason is that, since $(n_{i_1},n_{i_2}) \in E$,
   the edges $G_{i_2,n_{i_2}}[v_{i_1}] \stackrel{a}{\to}
   G_{i_1,n_{i_1}}[u_{i_2+1}]$ exist for all $i_1 \leq i_2$. Due to
   Lemma~\ref{lem:twopaths}(\ref{lem:twopaths:b}), this path has exactly $k'$ edges.
   The $b$-path, on the other hand, can go from $r_1$ to $r_{k+1}$
   and skip exactly $G_{i,n_i}$ for all $i = 1, \ldots, k$ (using the
   diagonal edges in Figure~\ref{fig:gadgets:row}). Since it
   skips these $G_{i,n_i}$, it is node-disjoint from the $a$-path and therefore we have a solution for
   \kcolordisjointpaths.

   For the other direction let us assume that there exist a simple
   $a$-path $p_a$ from $c_1$ to $c_{k+1}$ and a simple $b$-path $p_b$ from $r_1$ to $r_{k+1}$ in $G'$ such that $p_a$
   and $p_b$ are node-disjoint and $p_a$ has length $k'$.  We show
   that $G$ has a $k$-clique.
                     Since every $b$-path from $r_1$ to $r_{k+1}$ goes through
   each row, that is, from $r_i$ to $r_{i+1}$ for all $i=1,\ldots,k$,
   this is also the case for $p_b$. By construction $p_b$ must also
   skip exactly one gadget in each row, using the diagonal edges in
   Figure~\ref{fig:gadgets:row}. Furthermore, for each gadget
   $G_{i,j}$ that $p_b$ visits, it must be the case that it either
   visits all nodes $u_1, \ldots, u_{k+1}$ or all nodes $v_1, \ldots,
   v_{k+1}$. (This is immediate from Figure~\ref{fig:OneGadget},
   showing all internal edges of a gadget.)  Therefore, since $p_a$
   and $p_b$ are node-disjoint, the $p_a$ cannot visit any gadget
   $G_{i,j}$ already visited by $p_b$.  Therefore, $p_a$, which goes
   from $c_1$ to $c_{k+1}$ by
   Lemma~\ref{lem:twopaths}(\ref{lem:twopaths:b}), can only do so
   through the $k$ skipped gadgets, call them $G_{i,n_i}$ for
   $i=1,\ldots,k$. Recall that the edges between the gadgets
   $G_{i_2,n_{i_2}}$ and $G_{i_1,n_{i_1}}$ only exist if
   $(n_{i_1},n_{i_2}) \in E$.  As these edges are necessary for the
   existence of the $a$-path from $c_1$ to $c_{k+1}$, all $n_i$
   must be pairwise adjacent in $G$. That is, they form a clique of
   size $k$ in $G$.
 \end{proof}

\begin{proof}[Proof of Lemma~\ref{lem:twopaths}]
  First observe that $G'_a$ contains a fixed part that only depends on
  $n$ and $k$, plus a set of edges that represent edges in $G$, i.e.,
  edges that are present in $G'$ if and only if there exists a
  corresponding edge in $G$. Therefore, every possible graph $G'$
  that the reduction produces is a subgraph of the case where $G$ is
  a complete graph (i.e., if $G$ has $n$ nodes, it is the $n$-clique). Let $C'$ denote the graph $G'$ in the case where
  $G$ is the $n$-clique. We prove the following points, which imply
  the Lemma:
  \begin{enumerate}[(1)]
  \item The subgraph $C'_a$ of $C'$ consisting of the $a$-colored
    edges is a DAG.
  \item Each path $C'_a$ from $c_1$ to $c_{k+1}$ has length exactly $k'$.
  \item Each path in $C'_a$ has length exactly $k'$ if and only if it
    is from $c_1$ to $c_{k+1}$.
  \item Each path in $C'_a$ of length $k'$ visits all control
    nodes, i.e., it contains all $c_{i}$ and $c_{i_1i_2}$, with $i \in
    \{1,\ldots, k+1\}$ and $1\leq i_1<i_2 \leq k$. 
    \item Each path in $C'_a$ of length $k'$ has at least one edge $u_\ell\stackrel{a}{\to} v_\ell$ in every row of $C'_a$.
  \end{enumerate}

  We first prove part (1). We first show that, if $C'_a$ has a cycle,
  then this cycle must contain a control node. Indeed, within the same
  row, the graph $C'_a$ only has the edges from $u_i$ to $v_i$ in all
  the gadgets. So, there cannot be a cycle that only contains nodes
  from a single row. Therefore, the cycle must contain a path from
  some node in a row $i_1$ to a node in row $i_2$, for $i_1 <
  i_2$. Since every path in $C'_a$ from row $i_1$ to $i_2$
  with $i_1 < i_2$ contains, by construction, at least one control
  node, we have that every cycle in $C'_a$ must contain a control node.

  It therefore remains to show that $C'_a$ contains no cycle that uses
  a control node. To this end, observe that the relation $\prec$ where
  $n_1 \prec n_2$ iff $n_1 \neq n_2$ and $n_2$ is reachable from $n_1$ 
  is a strict total order
  \begin{equation}\label{eq:control-nodes-ordering}
    c_1, c_{12},c_{13},\ldots, c_{1k},c_2,c_{23},\ldots, c_{k-2k},
    c_{k-1k}, c_k, c_{k+1}
     \tag{$\dagger$} 
  \end{equation}
  on the control nodes. That is, the order is such that control nodes are reachable in $C'_a$
  from all control nodes to their left and none to their right.

  We now prove part (2). 
                First we prove that, between two consecutive control nodes in
  $C'_a$, each path has a fixed length that depends only on the kind
  of control nodes.  Then, since $C'_a$ is a DAG by part (1), we can
  simply concatenate paths to obtain the length of paths from $c_1$ to
  $c_{k+1}$, showing (2). In this proof, when we consider a path that visits nodes in
  row $i$ in $C'_a$, then by construction of $C'$, the length of this
  path is independent of the gadget $G_{i,j}$ that the path
  visits. That is, the path's length is the same for every $j =
  1,\ldots,n$. To simplify notation, we therefore omit the $j$ in
  $G_{i,j}[u]$ and write $G_{i}[u]$ instead.
  
  We first consider the length of paths between consecutive control
  nodes in the ordering \eqref{eq:control-nodes-ordering}. Therefore,
  fix two such consecutive control nodes $n_1$ and $n_2$. We make a
  case distinction:
  \begin{itemize}
  \item $n_1 = c_i$ and $n_2 = c_{i(i+1)}$: Each path from $c_i$ to
    $c_{i(i+1)}$ is of the form $c_{i},G_{i}[u_{i+1}], \allowbreak
    G_{i}[v_{i+1}], \allowbreak c_{i (i+1)}$ and therefore has length 3.
  \item $n_1 = c_{ij}$ and $n_2 = c_{i(j+1)}$: Each path from $c_{ij}$
    to $c_{i(j+1)}$ with $1 \leq i < j \leq k-1$ has the form
    $c_{ij}, \allowbreak G_{j}[u_{i}], G_{j}[v_{i}], G_{i}[u_{j+1}],
    \allowbreak G_{i}[v_{j+1}], \allowbreak c_{i(j+1)}$ and therefore 
    length 5.
  \item $n_1 = c_{ik}$ and $n_2 = c_{i+1}$: Each path from $c_{ik}$ to
    $c_{i+1}$ has the form $c_{ik}, G_{k}[u_{i}], G_{k}[v_{i}], \allowbreak
    G_{i}[u_{k+1}], \allowbreak G_{i}[v_{k+1}], c_{i+1}$ and therefore
    length 5.
  \item $n_1 = c_{k}$ and $n_2 = c_{k+1}$: Each path
  from $c_{k}$ to $c_{k+1}$ is of the form $c_{k}, G_{k}[u_{k+1}],
  G_{k}[v_{k+1}],\allowbreak c_{k+1}$ and therefore of length 3.
  \end{itemize}

  Since $\prec$ is a strict total order, this means that each path
  from $c_1$ to $c_{k+1}$ in $C'_a$ has the same length.
   We show that this length is
  exactly $k(k-1)/2\cdot 5+3k = k'$. The paths $c_i$ to $c_{ii+1}$ $(i
  = 1, \ldots, k-1)$ and $c_k$ to $c_{k+1}$ sum up to length $3k$. For a
  fixed $i$ we have $5\cdot(k-i-1)$ paths from $c_{ii+1}$ to $c_{ik}$,
  which sum up to length $5(k(k-1)/2)-5k+5$ for $i=1,\ldots, k-2$.
  Finally, we need to consider the paths from $c_{ik}$ to $c_{i+1}$,
  which, for $i=1,\ldots, k-1$, sum up to length $5k-5$.  This shows
  (2).

  Since $C'_a$ is a DAG, every node in $C'_a$ is reachable from $c_1$,
      since $c_{k+1}$ does not have outgoing edges in $C'_a$, and since each
  path of length $k'$ starting from $c_1$ ends in $c_{k+1}$, we also have
  (3).
  Since $\prec$ is a strict total order on the control nodes, we also have (4).

  Due to (3) and (4) each path of length $k'$ in $G'_a$ contains $c_i$ for $i=1,\ldots, k+1$. Since each path from $c_i$ to the next control node contains $(G_{i,j}[u_{i+1}],G_{i,j}[u_{i+1}])$, for a $j \in \{1,\ldots,n\}$ we also have (5).
            \end{proof}

\TheoDisjointPaths*
\begin{proof}
  We adapt the reduction from Theorem~\ref{akb*-is-hard}.  The only
  change we make is that we replace each $b$-edge by a
  directed path of $k'$ edges (introducing $k'-1$ new nodes for each
  such edge). 

  \fbox{$\cdots$}
 Call the
  resulting graph $G''$. 
  We make the following observation:
 \begin{observation} \label{obs:replace-b}
                In $G''$, we have that
    \begin{enumerate}[(a)]
    \item every path from $c_1$ to $c_{k+1}$ has length $\geq k'$ and
    \item every path from $c_1$ to $c_{k+1}$ has length exactly $k'$ if
      and only if it only uses $a$-edges.
    \end{enumerate}
  \end{observation}
          We prove the observation using
  Lemma~\ref{lem:twopaths}(\ref{lem:twopaths:b}). For part (a) we
  have two cases. If a path from $c_1$ to $c_{k+1}$ uses $a$-edges only,
  the result is immediate from
  Lemma~\ref{lem:twopaths}(\ref{lem:twopaths:b}). If it uses at
  least one $b$-edge, then it uses at least $k'$ $b$-edges by
  construction. 

  For part (b), if a path from $c_1$ to $c_{k+1}$ has length
  exactly $k'$, it uses at least one $a$-edge since $c_{k+1}$ only has
  incoming $a$-edges. If it would use at least one
  $b$-edge, it uses at least $k'$ $b$-edges by construction, which
  contradicts that the length is $k'$. The converse direction is
  immediate from Lemma~\ref{lem:twopaths}(\ref{lem:twopaths:b}). This
  concludes the proof of Observation~\ref{obs:replace-b}.

  We show that $G'$ and $k'$ are in \knodecolordisjointpaths if and only
  if $G''$ and $k'$ are in \knodedisjointpaths. That is, $G'$ has a
  simple $a$-path $p_a$ from $s_1$ to $c_{k+1}$ (of
  length $k'$) and simple $b$-path $p_b$ from $r_1$ to
  $r_{k+1}$ such that $p_a$ has length $k'$ and is node-disjoint from $p_b$ if and only if $G''$ has simple paths $p_1$
  from $c_1$ to $c_{k+1}$ and $p_2$ from $r_1$ to $r_{k+1}$, where $p_1$
  has length $k'$ and $p_1$ and $p_2$ are node-disjoint.

  If $G'$ and $k'$ are in \knodecolordisjointpaths, then we can use the
  corresponding paths in $G''$ (where we follow $b$-paths in $G''$
  instead of $b$-edges in $G'$). Conversely, if $G''$ and $k'$ are in
  \knodedisjointpaths, it follows from Observation \ref{obs:replace-b}
  that $p_1$ can only use $a$-edges.       We now show that the path $p_2$ from $r_1$ to $r_{k+1}$ can only use
  $b$-edges, that is, we show that it cannot use $a$-edges. There are
  3 types of $a$-edges: (i) the ones from and to control-nodes,
  (ii) ``upward'' edges that connect row $j$ to row $i$ with
  $j>i$, and (iii) edges from $u_\ell$ to $v_\ell$ in one
  gadget.

  Notice that, by construction, $p_2$ must visit nodes in row 1 and
  later also nodes in row $k$. To do so, $p_2$ cannot use edges from
  or to control nodes (type (i)), since, due to
  Lemma~\ref{lem:twopaths}(\ref{lem:twopaths:c}), $p_1$ already visits
  all of them.  So $p_2$ cannot go from row $i$ to a row $j$ with
  $i<j$ via $a$-edges. This means that, if $i < j$, then $p_2$ can
  only go from row $i$ to row $j$ by going through $r_{i+1}$ (and
  through nodes in row $i+1$), since every remaining path from row $i$
  to a larger row goes through $r_{i+1}$. So, in order to go from row
  1 to row $k$, path $p_2$ needs to visit all nodes $r_1,\ldots,r_k$,
  in that order.  This means that it is also impossible for $p_2$ to
  use edges of type \emph{(ii)}. Indeed, if $p_2$ were to use an edge
  from row $j$ to row $i$ with $j > i$, then it would need to visit
  $r_{i+1}$ a second time to arrive back in row $j$.
                  Finally, if $p_2$ used an $a$-edge of type \emph{(iii)} in row
  $i$, then, by construction, it would have to visit every gadget in this
  row.  But
  since $p_1$ already uses at least one edge in each row, see Lemma~\ref{lem:twopaths}(\ref{lem:twopaths:d}), this means
  that $p_2$ cannot be node-disjoint with $p_1$.  This shows that $G'$
  and $k'$ are in \knodecolordisjointpaths.
  \end{proof}

\TheoPathsAllTheRest*
\begin{proof}
  The theorem follows immediately from
  Lemmas~\ref{lemma:two-disjoint-at-most-k-is-hard}, \ref{lem:fortune}, and \ref{lem:UpperBound}.
\end{proof}

\begin{lemma} \label{lemma:two-disjoint-at-most-k-is-hard}
  $\twonodedisjointpaths_{\leq k}$ is W[1]-hard.
\end{lemma}
\begin{proof}
    We start from the same graph as in the proof of Theorem \ref{theorem:two-disjoint-is-hard}. Then,
  from Observation~\ref{obs:replace-b} we know that there exist no
  path from $c_1$ to $c_{k+1}$ that has length smaller than $k$. So the answer to our problem on
  this instance is the same as for \knodedisjointpaths, which completes
  the reduction.
\end{proof}

For completeness, we observe that $\twonodedisjointpaths_{\geq k}$ with
$k=0$ is simply the \twonodedisjointpaths problem.
\begin{lemma}[\cite{FortuneHW-TCS80}]\label{lem:fortune}
  $\twonodedisjointpaths_{\geq k}$ is \np-complete for every constant $k \in \nat$.
\end{lemma}

\begin{lemma}\label{lem:UpperBound}
  \knodecolordisjointpaths, \knodedisjointpaths, and $\twonodedisjointpaths_{\leq k}$ are in W[P].
\end{lemma}
\begin{proof}
  We show membership in W[P] by using Definition 3.1 in Flum and Grohe~\cite{FlumG-springer06}.    They say that W[P] is the class of parameterized problems that can
  be decided by a nondeterministic Turing machine (NTM) in time
  $f(k)\cdot |x|^{O(1)}$ and such that it makes at most $O(h(k)\cdot
  \log n)$ nondeterministic choices in the computation of any input $(x,k)$.
 
    The problem \knodedisjointpaths can be decided by such an nondeterministic Turing machine as follows. Given a graph $G=(V,E)$ and a parameter $k$.
    The NTM
    first uses $(k-1) \cdot \log n$ steps to guess $k-1$ nodes $v_1,\ldots, v_{k-1}$
    in the right order. Then we can verify in $O(k \log n)$ steps that these
    nodes form a simple path $p_1=(s_1, v_1)(v_1, v_2) \cdots
    (v_{k-1}, t_1)$ from $s_1$ to $t_1$.  
    After this, the NTM     tests deterministically that $t_2$ is reachable from $s_2$ while
    avoiding the nodes $s_1,v_2,\ldots,v_{k-1},t_1$. This can be done
    in time polynomial in $|G|$.

    For \knodecolordisjointpaths and $\twonodedisjointpaths_{\leq k}$  the proof is analogous.
\end{proof}

\subsection{Proofs for Section~\ref{sec:twoedgedisjoint}}

\twoEdgeDisjointPaths*
\begin{proof}
  Just like the reduction in Theorem~\ref{theo:kcolor}, this proof is
  inspired by the adaption of Grohe and Gr\"uber \cite[Lemma
  16]{GroheICALP-07} of the main reduction of Slivkins~\cite{slivkins}. In
  fact, since the reduction in Slivkins also considers trails, we use
  exactly the same gadgets here.  

 We give a reduction
  from \kclique. Let $(G,k)$ be an instance of \kclique.  We construct
  the graph $G'$ from Theorem~\ref{theo:kcolor} and make the following
  changes to obtain our final graph $H$:
\begin{itemize}
\item In each gadget $G_{i,j}$, we split each $u_\ell$ in two nodes, that is $u_\ell^\text{in}$ and $u_\ell^\text{out}$. We call the two nodes that resulted from the same node a \emph{node pair}. We redirect all incoming edges from $u_\ell$ to $u_\ell^\text{in}$ and let all outgoing edges from $u_\ell$ begin in $u_\ell^\text{out}$.  Finally, we add an edge $u_\ell^\text{in} \to u_\ell^\text{out}$.  We depict this in Figure~\ref{fig:edgeOneGadget}.
We make exactly the same change to all $v_\ell$, $c_i$, $c_{i_1i_2}$, and $r_i$. 
\item We replace each $b$-edge by a $b$-path of length $k_\text{new} = k(k-1)/2\cdot 5 + 3k+ k+1+k(k-1)/2 +k \cdot 2(k+1)= 5k^2+3k+1$. Notice that this $k_\text{new}$ is longer than  $k'=k(k-1)/2\cdot 5 + 3k$ in Theorem~\ref{theo:kcolor} because we split some nodes and added new edges between them. To be precise, the length of the $a$-path from $c_1^\text{in}$ to $c_{k+1}^\text{out}$ became longer because we split all $k+1+k(k-1)/2$ control nodes and it has to pass in total $k \cdot 2(k+1)$ new edges between $u_\ell^\text{in}$ and $u_\ell^\text{out}$ and between $v_\ell^\text{in} $ and $v_\ell^\text{out}$. 
\end{itemize}  
The correctness proof now follows the lines of the proof of Theorem~\ref{theorem:two-disjoint-is-hard}. We show that there exist paths $p_a$ and $p_b$ in $G'$ that such that $p_a$ is an $a$-path of length $k'$ from $c_1$ to $c_{k+1}$ and $p_b$ is a $b$-path from $r_1$ to $r_{k+1}$ if and only if there exist two edge-disjoint $p_1$ and $p_2$, where $p_1$ has length exactly $k_\text{new}$ and is from $c_{1}^\text{in}$ to $c_{k+1}^\text{out}$ and $p_2$ is from $r_1^\text{in}$ to $r_{k+1}^\text{out}$. 

We will now show that $p_1$ corresponds to $p_a$ and $p_2$ corresponds to $p_b$. This then proves the lemma. If $p_a$ and $p_b$ exist, then we can use their nodes (or node pairs) to build edge disjoint paths $p_1$ and $p_2$. (We use the same nodes or node pairs thereof and do not change the order.)

For the other direction, let us assume that $p_1$ and $p_2$ exist in $H$.
We first show that $p_1$ corresponds to an $a$-path in $G'$. 
We have constructed our graph $H$ such that each path from $c_{1}^\text{in}$ to $c_{k+1}^\text{out}$ has length at least $k_\text{new}$ and length exactly $k_\text{new}$ if and only if it chooses a path corresponding to an $a$-path in $G'$, see also Observation~\ref{obs:replace-b}. 
We now show that $p_2$ cannot correspond to any path in $G'$ that uses $a$-edges.
Recall that there are 3 types of $a$-edges in $G'$: (i) the ones from and to control nodes, (ii) ``upward'' edges that connect row $j$ to row $i$ with $j > i$, and (iii) edges from $u_\ell$ to $v_\ell$ in one gadget.
Since $p_1$ is a path from $c_{1}^\text{in}$ to $c_{k+1}^\text{out}$ of length exactly $k_\text{new}$, the corresponding $a$-path from $c_1$ to $c_{k+1}$ in $G'$ uses all control nodes, see Lemma~\ref{lem:twopaths}(\ref{lem:twopaths:c}). Therefore, $p_1$ must do the same. Since we did split all control nodes, $p_1$ especially contains the edge between each node pair of control nodes. This implies that $p_2$ cannot use the edge between any node pair of control nodes and its corresponding path in $G'$ cannot contain any $a$-edge from or to an control node, that is type (i). 
So $p_2$ cannot go from row $i$ to a row $j$ with $i<j$ via control nodes. This means that, if $i < j$, then $p_2$ can only go from row $i$ to row $j$ by going through $r_{i+1}^\text{in}$ and $r_{i+1}^\text{out}$ since every remaining path from row $i$ to a larger row goes through $r_{i+1}^\text{in}$ and $r_{i+1}^\text{out}$. So, in order to go from row 1 to row $k$, path $p_2$ needs to visit all nodes $r_{1}^\text{in},r_{1}^\text{out}, \ldots, r_{1}^\text{in},r_{1}^\text{out}$, in that order. This means that it is also impossible for $p_2$ to use ``upward'' edges. (Otherwise there would be an $i$, such that $p_2$ would use the edge between $r_{i+1}^\text{in}$ and $r_{i+1}^\text{out}$ twice.) So the corresponding path in $G'$ must not use $a$-edges of type (ii). 
Finally, if $p_2$ used an edge between $u_\ell^{out}$ and $v_\ell^{in}$ in any gadget in row $i$, then it would have to visit every gadget in this row by construction, i.e., for all $j \in \{1,\ldots,n\}$ and all $\ell \in \{1,\ldots,k+1\}$: $(G_{i,j}[u_\ell^{in}],a,G_{i,j}[u_\ell^{in}]) \in p_2 \lor (G_{i,j}[v_\ell^{in}],a,G_{i,j}[v_\ell^{in}]) \in p_2$. 
But we know from Lemma~\ref{lem:twopaths}(\ref{lem:twopaths:d}) that the path corresponding to $p_1$ in $G'$ uses at least one edge in each row. This means that in each row there exists a gadget $G_{i,j}$ and an $\ell$ such that $p_1$ uses the edges $(G_{i,j}[u_\ell^{out}],a,G_{i,j}[u_\ell^{out}]), (G_{i,j}[u_\ell^{out}],a,G_{i,j}[v_\ell^{in}])$, and $(G_{i,j}[v_\ell^{in}],a,G_{i,j}[v_\ell^{out}])$. So $p_2$ cannot be edge-disjoint with $p_1$ if is uses such an edge. This implies that the path corresponding to $p_2$ in $G'$ cannot use edges of type (iii), so we finally know that the path corresponding to $p_2$ in $G'$ only contains $b$-edges.
So $G'$ and $k'$ are indeed in \knodecolordisjointpaths and we have an FPT-reduction.
\end{proof}

 \begin{figure}[t]
   \centering
 \begin{tikzpicture}[->, scale = 1, auto ]
\def\xdistance{2}  \def\distance{2.5}
\def\xout{+.0}
\def\yout{-.7}
\def\labell{${}$}
	\foreach \i in {1,2,4,5}{
		\node (Ai\i) at (\xdistance*\i,0) {};
		\node  (Bo\i) at (\xdistance*\i+\xout,\distance*-.7+\yout) {}; 
		\node  (Ao\i) at (\xdistance*\i+\xout,0+\yout) {};
		\node  (Bi\i) at (\xdistance*\i,\distance*-.7) {}; 
		\path	(Ai\i) edge  node  {\labell} (Ao\i);
		\path	(Bi\i) edge  node  {\labell} (Bo\i);
		\fill (Ai\i) circle (2pt);
		\fill (Bi\i) circle (2pt);
		\fill (Ao\i) circle (2pt);
		\fill (Bo\i) circle (2pt);
		\path	(Ao\i) edge  node  {\labell} (Bi\i);
	} 

    \foreach \i in {1,2}{
	  \node[label=right:$u_\i^\text{in}$,xshift = -1ex] at (Ai\i) {};
	  \node[label=left:$v_\i^\text{out}$,xshift = +1ex] at (Bo\i) {}; 
	  \node[label=left:$u_\i^\text{out}$,xshift = +1ex] at (Ao\i) {};
	  \node[label=right:$v_\i^\text{in}$,xshift = -1ex] at(Bi\i)  {}; 
    }
    \foreach \i in {4}{
	  \node[label=right:$u_k^\text{in}$,xshift = -1ex] at  (Ai\i) {};
	  \node[label=left:$v_k^\text{out}$,xshift = +1ex] at (Bo\i){};
	  \node[label=left:$u_k^\text{out}$,xshift = +1ex] at (Ao\i){};
	  \node[label=right:$v_k^\text{in}$,xshift = -1ex] at (Bi\i){};
    }    
    
    \foreach \i in {5}{
	  \node[label=right:$u_{k+1}^\text{in}$,xshift = -1ex] at (Ai\i){};
	  \node[label=left:$v_{k+1}^\text{out}$,xshift = +1ex] at (Bo\i){};
	  \node[label=left:$u_{k+1}^\text{out}$,xshift = +1ex] at (Ao\i){};
	  \node[label=right:$v_{k+1}^\text{in}$,xshift = -1ex] at  (Bi\i){};
    }
    
	\path 
    (Ao1) edge	  node  {\labell} (Ai2)
    (Bo1) edge	 [->]    node  {\labell} (Bi2) 
    (Ao4) edge	 [->]    node  {\labell} (Ai5)
    (Bo4) edge	 [->]    node  {\labell} (Bi5);   
    \node[label=center:$\cdots$] (A3) at ($(Ao2)!0.5!(Ai4)$) {}; 
    \node[label=center:$\cdots$] (B3) at ($(Bo2)!0.5!(Bi4)$) {}; 
    \draw[->] (Ao2) -- ($(Ao2)!0.75!(A3)$);
    \draw[->] (Bo2) -- ($(Bo2)!0.75!(B3)$);
    \draw[->] ($(A3)!0.25!(Ai4)$) -- (Ai4);
    \draw[->] ($(B3)!0.25!(Bi4)$) -- (Bi4);
\def\offset{1}     \draw[dashed, rounded corners=10pt] (\xdistance*1-\offset,\distance*-.9-\offset+.4) rectangle (\xdistance*5+\offset,\offset-0.6) {};     \end{tikzpicture}
    \caption{Internal structure of each gadget $G_{i,j}$ in the proof of Theorem~\ref{theo:edgeTwodisjoint}. All edges are $a$-edges.}
    \label{fig:edgeOneGadget}
  \end{figure}
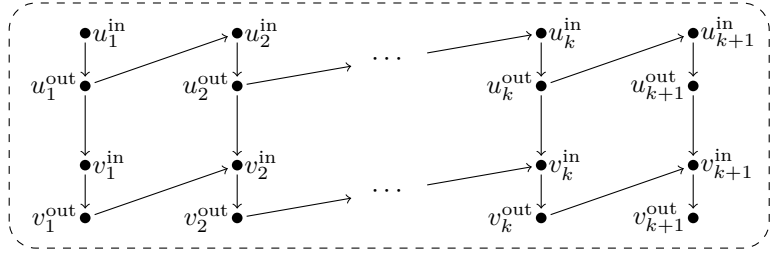

\makeatletter{}\section{Proofs for Section~\ref{sec:tractlanguages}}

\begin{algorithm}[t]
\caption{}\begin{algorithmic}[1]
\Require Graph $G=(V,E)$, nodes $s,t$ in $G$, parameter $k$, \re $a^kw?a^*$
\Ensure Decide if there exists a simple path from $s$ to $t$ matching the given \re
\If {FLPS$((V,E\cap(V \times \{a\} \times V)),s,t,k)$} return
YES  \label{alg:b?fptalg:1} \Comment{Call Algorithm~\ref{fptalg}}
\EndIf
\State $S = \{p_c \mid p_c=(u_0,u_1)\cdots (u_{c-1},u_c)$ is a simple
path that matches $w\}$
\For {each $p_c \in S$} \label{alg:b?fptalg:3} 
\State $V'_a \gets V \setminus V(p_c[1,c-1])$ \Comment{Delete all but the first and last node of $p_c$}
\State $E'_a = E \cap (V'_a \times \{a\} \times V'_a)$ \Comment{Consider only $a$-edges}
\State Compute $\hat{P}^{k+1}_{su_0} \subseteq^{k+1}_\text{rep} P^{k+1}_{su_0}$ in $(V'_a,E'_a)$ \label{alg:b?fptalg:7} 
\For {all sets $X \in \hat{P}^{k+1}_{su_0}$} \label{alg:b?fptalg:8} 
\State $V' \gets (V'_a \setminus X)$
\State $E' \gets E'_a \cap (V' \times V')$
\If {there exists a path from $u_c$ to $t$ in $(V',E')$}
\State return YES
\EndIf
\EndFor
\EndFor
\State return NO
\end{algorithmic}
\label{b?fptalg}
\end{algorithm}

\WConstOrNOTIsFPT*
\begin{proof}
          We give an FPT algorithm that solves this problem, see
  Algorithm~\ref{b?fptalg}.
  We first prove that Algorithm~\ref{b?fptalg} is correct.  If
  there exists a simple path matching $a^ka^*$, we find it using
  Algorithm~\ref{fptalg} in line~\ref{alg:b?fptalg:1} and return `yes'.

  So let us assume that there exists no such path.  We then use brute
  force to find all simple paths $p_c=(u_0,u_1)\cdots (u_{c-1},u_c)$
  matching $w$ and store them in a set $S$.   For each path $p_c \in S$, we compute the graph
  $(V'_a,E'_a)$ that does not contain the inner nodes of $p_c$ and
  only contains $a$-edges. Then we compute a set
  $\hat{P}^{k+1}_{su_0} \subseteq^{k+1}_\text{rep} P^{k+1}_{su_0}$ in
  $(V'_a,E'_a)$. We will
  argue using Lemma~\ref{proofforfptalg} that it suffices to consider
  paths in which the first $k+1$ nodes belong to a set $X \in
  \hat{P}^{k+1}_{su_0}$. To this end, assume that there is a
  simple path $$p = (v_0,v_1) \cdots (v_{k-1},u_0) \cdot p_c \cdot (u_c,v_{k+1})
  \cdots (v_{r-1},v_r)$$ matching $a^kwa^*$.  (We wrote some of the
  concatenation operators $\cdot$ explicitly to improve readability.)

  We consider two cases. If $|p| \leq
  2k+c$, we define $P= (v_0,v_1) \cdots (v_{k-1},u_0)$ and
  $Q=(u_c,v_{k+1}) \cdots \allowbreak (v_{r-1},v_r)$. Notice that
  $|V(Q)| \leq k+1$, $V(P)\cap V(Q) = \emptyset$, and $P \in  P^{k+1}_{su_0}$.
  Therefore, since $\hat{P}^{k+1}_{su_0}
  \subseteq^{k+1}_\text{rep} P^{k+1}_{su_0}$, there exists at least one
  path $P'$ with $V(P') \in \hat{P}^{k+1}_{su_0}$ and $V(P')\cap
  V(Q)=\emptyset$.  
  Since  $(V'_a,E'_a)$ does not contain any nodes
  of $p_c[1,c-1]$ by definition, we also know that $V(P') \cap V(p_c) =
  \{u_0\}$. (Notice that $u_c$ cannot be in the intersection, because
  it is in $V(Q)$.) This means that $P' \cdot p_c\cdot Q$ is indeed a simple path that matches $a^kwa^*$. 
	
  Otherwise we have that $|p| > 2k+c$, in which case we define $P= (v_0,v_1) \cdots (v_{k-1},u_0)$
  and $Q=(v_{r-k},v_{r-k+1}) \cdots \allowbreak (v_{r-1},v_r)$ and $R
  = (u_c,v_{k+1}) \cdots (v_{r-k-2},v_{r-k-1})$. So we have that 
  $$p = P \cdot p_c \cdot R\cdot (v_{r-k-1},v_{r-k}) \cdot Q.$$
  We also know that $P \in  P^{k+1}_{su_0}$ and $|V(Q)| = k+1$.
  Therefore, by
  definition of $\hat{P}^{k+1}_{su_0} \subseteq^{k+1}_\text{rep}
  P^{k+1}_{su_0}$, there must be a path $P'=(v_0,v'_1)(v'_1,v'_2)\cdots(v'_{k-1},u_0)$ with $V(P') \in
  \hat{P}^{k+1}_{su_0}$ such that $V(P')\cap V(Q) = \emptyset$. The
  path $P'$ is also does not contain any of the inner nodes of $p_c$, because $G'_a$ does not
  contain nodes of
  $V(p_c[1,c-1])$. 
	
  There are again two possibilities: $P'$ intersects with $R$ or
  not. In the first case, there exists a node $v'_i \in V(P')$ with
  minimal $i$ such that $v'_i = v_j \in V(R)$.  Then we replace the path $p$ by a new simple path $p' = (v_0,v'_1) \cdots
  (v'_{i},v_{j+1})(v_{j+1},v_{j+2}) \cdots (v_{r-1},v_r)$. But then
  $p'$ matches $a^ka^*$, because it does not contain $p_c$, whereas it
  still contains $Q$. This contradicts that no such path $p'$ exists
  (we would have found this path with Algorithm~\ref{fptalg} in
  line~\ref{alg:b?fptalg:1}).  Therefore, $P'$ does not intersect with
  $R$ and we have found our path.

  Finally we note that this algorithm is indeed an FPT algorithm. Line
  \ref{alg:b?fptalg:1} works in \fpt due to
  Theorem~\ref{theorem:longdirpathInFPT}. The set $S$ in
  line~\ref{alg:b?fptalg:3} can contain at most $O(n^c)$ different
  paths, so enumerating all of them is in \ptime. The rest is
  analogous to Algorithm~\ref{fptalg} and therefore in \fpt.
\end{proof}

\section{Proofs for Section~\ref{sec:blocklanguages}}

\subsection{Dichotomy for Node-Disjoint Paths}

\begin{algorithm}[t]{}
\caption{}
\begin{algorithmic}[1]
  \Require 
  \begin{tabular}[t]{l}
    Graph $G=(V,E)$, nodes $s,t$ in $G$, and 0-bordered $r= A_1 \cdots
    A_{k_1} A^* A'_{k_2} \cdots A'_1$\\
    \multicolumn{1}{r}{(assuming $A, A_{1},\ldots, A_{k_1}, A'_1,\ldots A'_{k_2} \neq \emptyset$)}
  \end{tabular}
  \Ensure Does there exist a simple path from $s$ to $t$ matching $r$?
\State $k \gets k_1 + k_2$ 
\If {there exists a simple path from $s$ to $t$ matching  $A_1 \cdots A_{k_1} A^{\leq k} A'_{k_2} \cdots A'_1$\label{alg:block1}}
\State return YES
\EndIf
\For {all $v \in V$\label{alg:block4}}
\State Compute $\hat{P}^{k+1}_{sv,r_1} \subseteq^{k+1}_\text{rep} P^{k+1}_{sv,r_1}$ in $G$ with  $r_1= A_1 \cdots A_{k_1} A^{k_2} $.  \label{alg:block5} \For {all sets $X \in \hat{P}^{k+1}_{sv,r_1}$\label{alg:block6}}
\State $V' \gets (V \setminus X) \cup \{v\}$
\State $E' \gets E \cap (V' \times \Sigma \times V')$\label{alg:blockRevEdges}\For {all $u \in V$} \label{alg:block12}
\State Compute $\hat{P}^{k_2+1}_{ut,r_2} \subseteq^{k_2+1}_\text{rep} P^{k_2+1}_{ut,r_2}$ in $(V',E')$ with  $r_2= A'_{k_2} \cdots A'_{1}$. \label{alg:block13}
\For {all sets $X' \in \hat{P}^{k_2+1}_{ut,r_2}$}
\State $V'' \gets (V' \setminus X') \cup \{u\}$
\State $E'' \gets E' \cap (V'' \times A \times V'')$ \Comment{$(V'',E'')$ has only $A$-edges}
\If {there exists a path from $v$ to $u$ in $(V'',E'')$\label{alg:blockShortest2}}
\State return YES
\EndIf 
\EndFor
\EndFor \label{alg:block22}
\EndFor
\EndFor \label{alg:block:lastfor}
\State return NO
\end{algorithmic}
\label{alg:block}
\end{algorithm}

\begin{lemma} \label{lemma:0crit}
  Let $\cR$ be the class of 0-bordered STEs. Then
  $\nodesimpath(\cR)$ is in \fpt with parameter $k_r$.
\end{lemma}
\begin{proof}
  We prove the lemma by case distinction on the form of $$r =
  B_\text{pre} A^* B_\text{suff} \in \cR.$$ There are two cases for $A$:
  either $A = \emptyset$ or $A \neq \emptyset$. In the former case,
  $L(r)$ is finite. Then we can use an algorithm obtained by 
  Bagan et al.~\cite[Theorem 6, Corollary 2]{bagan}, to solve it in
  FPT. (Bagan et al.\ use the size of the NFA as parameter for their
  algorithm.)
  
  Otherwise, we know that $A \neq \emptyset$. Furthermore, we can
  assume \mbox{w.l.o.g.} that all $A_1,\ldots, \allowbreak A_{k_1}, \allowbreak A'_{1},\ldots,
  A'_{k_2}$ are non-empty. Indeed, if this would not be the case,
  then the expression can be simplified ($L(r) = \emptyset$ is easy to
  test and $\emptyset ?$ can be simplified to $\varepsilon$).  We now
  differentiate between the forms of $B_\text{pre}$ and
  $B_\text{suff}$. There are two possible forms, that is (1) $B_1?
  \cdots B_{\ell}?$ with $\ell \geq 0$ or (2) $B_1 \cdots B_{\ell}$
  with $\ell \geq 1$. If
  $B_\text{pre}$ and $B_\text{suff}$ are of form (1), the language is
  downward closed. Therefore we can evaluate the answer in \ptime, see
  Proposition~\ref{prop:prefixclosed}.  If $B_\text{pre}$ and
  $B_\text{suff}$ are both of form (2), we will show that we can use
  Algorithm~\ref{alg:block}.  We show the correctness of this
  algorithm next and we explain later how to change it if
  $B_\text{pre}$ has form (2) (resp., (1)) and $B_\text{suff}$ has
  form (1) (resp., (2)).

  We first give the idea of the algorithm. Let $k = k_1 + k_2$ (that
  is, $k$ is the parameter $k_r$ from Definition~\ref{def:ste}).
   First the algorithm tests
  if there is a simple path that matches $A_1 \cdots A_{k_1} A^{\leq
    k} A'_{k_2} \cdots A'_1$. Dealing with this case separately
  simplifies the cases we need to treat in
  lines~\ref{alg:block4}--\ref{alg:block:lastfor}. In
  lines~\ref{alg:block4}--\ref{alg:block:lastfor} the algorithm
  essentially performs two nestings of Algorithm~\ref{fptalg}: Since
  neither the language of $B_\text{pre}$ nor the language of
  $B_\text{suff}$ is downward closed, we need to execute
  Algorithm~\ref{fptalg} once to find the prefix and once to find the
  suffix. Furthermore, we use a variant of $\hat{P}^k_{sv}$, namely
  $\hat{P}^k_{sv,r_1}$, that allows us to make sure that the prefix
  (suffix, resp.) of the path matches $B_\text{pre}$ ($B_\text{suff}$,
  resp.). More precisely, 
  \begin{multline*}
    P^{k}_{sv,r_1} := \{X \mid X\subseteq V \text{ such that } s,v \in X,
|X| = k, \text{ and there is a path from $s$ to $v$ in } G \\
 \text{ of length } k-1, \text{ matching } r_1,  \text{ and containing exactly the nodes in } X \}.
\end{multline*}

 We now show the correctness.
  If we have a simple $s$-$t$-path matching $r$ of length up to $2k$, it will be found in
  line~\ref{alg:block1}.   So it remains to test whether there exists a simple $s$-$t$-path matching $r$ of length at least $2k+1$.
  In each such path the first $k+1$ nodes must match $r_1 = A_1 \cdots A_{k_1} A^{k_2}$, while the rest of the path matches $A^* A'_{k_2} \cdots A'_{1}$. 
  We now prove analogously to Lemma~\ref{proofforfptalg} that it suffices to consider paths in which the first $k+1$ nodes belong to $X \in \hat{P}^{k+1}_{sv,r_1} \subseteq^{k+1}_\text{rep} P^{k+1}_{sv,r_1}$. 
In fact, the proof of Lemma~\ref{proofforfptalg} needs no adaption. It
still works since $r$ is 0-bordered. To be more precise: if we start
with a shortest simple path $p=P \cdot (v_k,v_{k+1}) \cdot  R \cdot  (v_{|p|-k-1},v_{|p|-k}) \cdot  Q$ of length at least $2k+2$ (or $p=P \cdot (v_k,v_{k+1}) \cdot  Q$ with $R = \varepsilon$ if $|p| = 2k+1$) that matches $r$, we also find a path $P'$ with $V(P') \in \hat{P}^{k+1}_{sv_k,r_1}$ that is disjoint from $V(Q)$. If $P'$ and $R$ intersect, we obtain a shorter simple path that still matches $r$ because $r$ is 0-bordered and the resulting path is still longer than $k$ (it contains $Q$). Notice that this is the reason why we need to consider paths of length
$k$ in line~\ref{alg:block5} of the algorithm, instead of length $k_1$.

We now obtained that if there is a simple $s$-$t$-path matching $r$ of length at least $2k+1$, then there exists a $v \in V$ and a set $X$ in $\hat{P}^{k+1}_{sv,r_1}$, such that its first $k+1$ nodes belong to $X$. Then we need to find the rest of the path, that is, a simple $v$-$t$-path matching $A^* A'_{k_2} \cdots A'_1$ in the graph without $X\setminus \{v\}$. 

Due to line \ref{alg:block1} we know that the $v$-$t$-path must have length at least $k_2+1$. Symmetrically to before we can show that, if such a path exists, then there exists a $u$ such that its last $k_2+1$ nodes belong to a set $X' \in \hat{P}^{k_2+1}_{ut,r_2} \subseteq^{k_2+1}_\text{rep} P^{k_2+1}_{ut,r_2}$. It then remains to test if there is a path from $v$ to $u$ that matches $A^*$ which is done in line \ref{alg:blockShortest2}.
This concludes the correctness proof.

We next show that the algorithm is indeed in \fpt.
Bagan et al.~\cite[Theorem 6]{bagan} showed that the test in line \ref{alg:block1} is in FPT, to be precise $O(2^{O(2k)}\cdot|N|\cdot |G| \cdot \log |G|)$ where $|N|$ is the size of the NFA used. Here we use an NFA of size $2k$.
It remains to show that the $\hat{P}^{k+1}_{sv,r_1}$ in line
\ref{alg:block5} (and $\hat{P}^{k_2+1}_{ut,r_2}$ in line
\ref{alg:block13}, resp.) can be computed in \fpt time and its size
depends only on $k$. This is guaranteed by
Lemma~\ref{lemma:SizeBoundWithRE}.                                 (Notice that we can efficiently compute $\hat{P}^{k_2+1}_{ut,r_2}$ for all $u$ by using the construction in Lemma~\ref{lemma:SizeBoundWithRE} to compute $\hat{P}^{k_2+1}_{tu,r_2}$ on the graph with reversed edges.)

So Algorithm~\ref{alg:block} is indeed an \fpt algorithm if $B_\text{pre}$ and $B_\text{suff}$ are of form (2). We now explain how it can be changed to work if $B_\text{pre}$ is of form (2) and $B_\text{suff}$ of form (1), that is: Assume we have $r = A_1 \cdots A_{k_1} A^* A'_{k_2}?\cdots A'_1?$. Then we make the following changes:
\begin{itemize}
\item in line \ref{alg:block1} the path should match $A_1 \cdots A_{k_1} A^{\leq k+k_2} A'_{k_2}?\cdots A'_1?$.
\item we replace line \ref{alg:block12} to \ref{alg:block22} with a test for an $v$-$t$-path matching $r'=A^* A'_{k_2}?\cdots A'_1?$. Notice that this implies that there exists a simple $v$-$t$-path matching $r'$ since $L(r')$ is downward closed.
\end{itemize}
The case $r= A_1? \cdots A_{k_1}? A^* A'_{k_2}\cdots A'_1$ is symmetric. It has indeed the same answer as $A'_1 \cdots A'_{k_1} A^* A_{k_2}?\cdots A_1?$ on the graph with reversed edges. 
\end{proof}

So, to complete the proof of Lemma~\ref{lemma:0crit}, it remains to
prove Lemma~\ref{lemma:SizeBoundWithRE}. We need to introduce some
terminology and notation. 
For $p \in \nat$, a \emph{$p$-family} $\mathcal{A}$ is a set containing
sets of size $p$. By $|\mathcal{A}|$ we denote the number of sets in $\mathcal{A}$.

We also restate Lemma 3.3
and Corollary 4.16 from \cite{fomin} since we need them in the
proof. Lemma~\ref{Fomin3.3} states that the relation ``is a $q$-representative
set for'' is transitive.
\begin{lemma}[Lemma 3.3 in \cite{fomin} for directed graphs]\label{Fomin3.3}
Given a graph $G=(V,E)$ and $\cS$ a family of subsets of $V$. If $\cS'\subseteq^q_\text{rep} \cS$ and $\hat{\cS}\subseteq^q_\text{rep} \cS'$, then $\hat{\cS}\subseteq^q_\text{rep} \cS$.
\end{lemma}

\begin{corollary}[Corollary 4.16 in \cite{fomin}, without weight function]
\label{Fomin4.16}
There is an algorithm that, given a $p$-family $\mathcal{A}$ of sets
over a universe $U$ of size $n$ and an integer $q$, computes in time $O\left(|\mathcal{A}| \cdot \left(\dfrac{p+q}{q} \right)^q \cdot 2^{o(p+q)} \cdot \log n \right)$ a subfamily $\hat{\mathcal{A}} \subseteq^q_\text{rep} \mathcal{A}$ such that $|\hat{\mathcal{A}}|\leq  \binom{p+q}{p} \cdot 2^{o(p+q)}$. 
\end{corollary}

We now adapt Lemma 5.2 in Fomin et al.~\cite{fomin} to show a time and
space bound for the sets $\hat{P}^{k+1}_{sv,r_1}$ and $\hat{P}^{k_2+1}_{ut,r_2}$ in the proof of
Lemma~\ref{lemma:0crit} (where $r_1=A_1\cdots A_{k_1}A^{k_2}$ and $r_2 = A'_{k_2}\cdots A'_1$). 
  We think
that this lemma might be of interest for others that try to find more
languages that can be evaluated in \fpt.

Recall that 
  \begin{multline*}
    P^{k}_{sv,r} := \{X \mid X\subseteq V \text{ such that } s,v \in X,
|X| = k, \text{ and there is a path from $s$ to $v$ in } G \\
 \text{ of length } k-1, \text{ matching } r,  \text{ and containing exactly the nodes in } X \}.
\end{multline*}
\begin{lemma} \label{lemma:SizeBoundWithRE} For each regular expression
  $r=A_1\cdots A_{k-1}$, a collection of families $\hat{P}^k_{sv,r}
  \subseteq^k_\text{rep} P^k_{sv,r}$, $v \in V\setminus \{s\}$ of size at
  most $\binom{2k}{k} \cdot 2^{o(2k)}$ each can be found in time
  $O\left( 8^{k+o(k)}m \log n \right),$ where $n = |V|$ and $m = |E|$.
\end{lemma}
\begin{proof}

  We describe a dynamic programming-based algorithm. Let $V =
  \{s,v_1,\ldots,v_{n-1}\}$ and $D$ be a $(k-1)\times (n-1)$ matrix
  where the rows are indexed with integers in ${2,\ldots, k}$ and the
  columns are indexed with nodes in $\{v_1, \ldots ,
  v_{n-1}\}$. The entry $D[i, v]$ will store the family
  $\hat{P}^i_{sv,r^{i-1}} \subseteq^{2k-i}_\text{rep} P^i_{sv,r^{i-1}}$,
  where $r^i$ denotes the prefix $A_1 \cdots A_i$ of $r$. We fill the
  entries in the matrix $D$ in the increasing order of rows.  For $i =
  2$, we set $D[2,v] = \{\{s,v\}\}$ if $G$ has an edge of the form $s
  \stackrel{a}{\to} v$ with $a \in A_1$ (otherwise $D[2,v] = \{\}$). 
  Assume that we have filled
  all the entries until row
  $i$. 
 For two families of
sets $\mathcal{A}$ and $\mathcal{B}$, we define $$\mathcal{A} \bullet
\mathcal{B} = \{ X \cup Y \mid X \in \mathcal{A},\text{ }Y \in \mathcal{B},
\text{ and } X \cap Y = \emptyset\}.$$  
  We denote by $u \stackrel{A_i}{\to} v$ that
  there exists an edge $u \stackrel{a}{\to} v$ with $a \in A_i$. 
   Let $$\mathcal{N}^{i+1}_{sv,r^i}=\bigcup_{u\stackrel{A_i}{\to}
    v} \hat{P}^i_{su,r^{i-1}} \bullet \{v\}.$$
We next adapt Claim 5.1 in \cite{fomin} such that it takes $r$ into account, that is:
\begin{claim}
\label{Fomin5.1}
$\mathcal{N}^{i+1}_{sv,r^i}\subseteq^{2k-(i+1)}_\text{rep} P^{i+1}_{sv,r^i}$
\end{claim}
\begin{proof}
Let $S \in P^{i+1}_{sv,r^i}$ and $Y$ be a set of size $2k-(i+1)$ such that $S \cap Y = \emptyset$. We will show that there exists a set $S' \in \mathcal{N}^{i+1}_{sv,r^i}$ such that $S' \cap Y = \emptyset$. This will imply the desired result. Since $S \in P^{i+1}_{sv,r^i}$,
there exists a path $P = (s,u_1)\cdots (u_{i-1},v)$ in $G$ such that
$S = V(P)$ and $u_{i-1} \stackrel{A_{i}}{\to} v$. 
The existence of path $P[0, i-1]$, the subpath of $P$ between $s$ and $u_{i-1}$, implies that $X^* = S\setminus\{v\} \in P^{i}_{su_{i-1},r^{i-1}}$. Take $Y^* = Y \cup \{v\}$. Observe that $X^*\cap Y^* = \emptyset$ and $|Y^*|=2k-i$. Since $\hat{P}^i_{su_{i-1},r^{i-1}} \subseteq^{2k-i}_\text{rep} P^i_{su_{i-1}r^{i-1}}$, there exists a set $\hat{X}^* \in \hat{P}^i_{su_{i-1},r^{i-1}}$ such that $\hat{X}^*\cap Y^* = \emptyset$.
However, since $u_{i-1}\stackrel{A_i}{\to} v$ and $\hat{X}^*\cap \{v\} = \emptyset$ (as $\hat{X}^*\cap Y^* = \emptyset$), we have $\hat{X}^*\bullet \{v\} = \hat{X}^*\cup \{v\}$ and $\hat{X}^*\cup \{v\} \in \mathcal{N}^{i+1}_{sv,r^i}$. Taking $S' = \hat{X}^*\cup\{v\}$ suffices for our purpose. This completes the
proof of the claim.
\end{proof}

We fill the entry for $D[i+1, v]$ as follows. Observe that 
$$\mathcal{N}^{i+1}_{sv,r^i}=\bigcup_{u\stackrel{A_i}{\to} v} D[i,u] \bullet \{v\}.$$
We already have computed the family corresponding to $D[i, u]$ for all $u$. By Corollary~\ref{Fomin4.16}, we have
$|\hat{P}^i_{su,r^{i-1}}| \leq \binom{2k}{i} 2^{o(2k)}$ and thus also
$|\mathcal{N}^{i+1}_{sv,r^i}| \leq d^-(v)\binom{2k}{i}2^{o(2k)}$.
Furthermore, we can compute $\mathcal{N}^{i+1}_{sv,r^i}$ in time
$O\left(d^-(v)\binom{2k}{i}2^{o(2k)} \right)$, where $d^-(v)$ denotes
the indegree of $v$, i.e., the number of edges that end in $v$. (We assume that testing whether there
exists an edge with a label in $A_i$ has the same complexity as
testing whether there exists an edge at all, up to a constant factor.) Now, using Corollary~\ref{Fomin4.16}, we compute $\hat{\mathcal{N}}^{i+1}_{sv,r^i} \subseteq^{2k-i-1}_\text{rep} \mathcal{N}^{i+1}_{sv,r^i}$ in time $$O\left( d^-(v)\binom{2k}{i} 2^{o(2k)}  \cdot \left(\dfrac{2k}{2k-i-1} \right)^{2k-i-1} \cdot 2^{o(2k)} \cdot \log n \right).$$ By Claim 5.1, we know that
$\mathcal{N}^{i+1}_{sv,r^i} \subseteq^{2k-i-1}_\text{rep} P^{i+1}_{sv,r^i}$. Thus, Lemma~\ref{Fomin3.3} implies that $\hat{\mathcal{N}}^{i+1}_{sv,r^i} = \hat{P}^{i+1}_{sv,r^i} \subseteq^{2k-i-1}_\text{rep} P^{i+1}_{sv,r^i}$. We assign this family to $D[i+1,v]$. This completes the description and the correctness of the algorithm. 

Notice that, if we keep the elements in the sets in the
order in which they were built using the $\bullet$ operation, then
they directly correspond to paths. As such, every ordered set in our family represents a path in the graph.

Since our only change was that we test $u\stackrel{A_i}{\to} v$ instead of $u\to v$, the
time bound of $O\left( 8^{k+o(k)}m \log n \right)$ \cite[Lemma 5.2]{fomin} carries over. The size bound is still guaranteed by Corollary~\ref{Fomin4.16}.\end{proof}

\dichotomyThm*

\begin{proof}

  We first prove part (a). To this end, let $c$ be a constant such
  that every $r \in \cR$ is at most $c$-bordered. Let $r \in R$. 
  Then we know that the
  maximum of its left cut border $c_1$ and its right cut borders $c_2$ is at most $c$.
    So we can enumerate, for all $u,v \in V$, all simple $s$-$u$-paths $p_1$ matching
  $A_1 \cdots A_{c_1}$ and all simple $v$-$t$-paths $p_2$ matching
  $A'_{c_2} \cdots A'_1$ in time $O(n^c)$. For all such
  node-disjoint paths $p_1$ and $p_2$, we delete in $G$ all nodes in
  $(V(p_1)\setminus\{u\}) \cup (V(p_2) \setminus \{v\})$. 
  In the
  remaining graph, we search a $u$-$v$-path that matches
  \begin{itemize}
  \item $r'=A_{c_1+1} \cdots A_{k_1} A^* A'_{k_2} \cdots A'_{c_2+1}$,
  \item $r'=A_{1}? \cdots A_{k_1}? A^* A'_{k_2} \cdots A'_{c_2+1}$,
  \item $r'=A_{c_1+1} \cdots A_{k_1} A^* A'_{k_2}? \cdots A'_{1}?$, or
  \item $r'=A_{1}? \cdots A_{k_1}? A^* A'_{k_2}? \cdots A'_{1}?$.
  \end{itemize}
  This are the only possibilities. Remember that the left cut border of $B_\text{pre}=A_{1}? \cdots A_{k_1}?$ is 0.
  Since $r'$ is 0-bordered,
  Lemma~\ref{lemma:0crit} allows us to solve \nodesimpath$(r')$ in time $f(k_{r'})\cdot |G|^d$ for a constant
  $d$ and a computable function $f$. Since $k_{r'}+c_1+c_2 = k_r$, this
  shows that $\nodesimpath(\cR)$ is in FPT with parameter $k_r$. 

  We now prove part (b). Let $\cR$ be an arbitrary but fixed class of
  STEs that can be sampled.
    We show that $\nodesimpath(\cR)$ is W[1]-hard by giving an FPT reduction from
  \kclique, which is known to be W[1]-hard (with parameter $k$). 
  Let $(G,k)$ be an input to \kclique. We will construct a graph
  $(H,s,t)$ and an expression $r \in \cR$ such that $(G,k) \in
  \kclique$ if and only if $H$ has a simple $s$-$t$-path that matches
  $r$. Let $k' = k(k-1)/2\cdot 5 + 3k$.  Since $\cR$ can be sampled, a
  $k''$-bordered expression $r\in\cR$ with $k'' \geq k'+1$ can be
  computed within time $f(k)$, for some computable function
  $f$. Therefore, we also know $k_r \leq f(k)$, else $r$ could not be computed in this time.     
  Since $r$ is $k''$-bordered, we have that its left cut border is $k''$ or its right cut border is $k''$ (or both). 
  
    Here we only consider the case that the left cut border is $k''$, i.e., $A \not \subseteq A_{k''}$, the other is symmetric. 
  For $r$ to have a left cut border of $k''$, it must be of the form
 $$r = A_1 \cdots A_{k'} \cdots A_{k''} \cdots A_{k_1} A^* A'_{k_2} \cdots A'_1 
 \text{\hspace{0.25cm} or \hspace{0.25cm}}
 r = A_1 \cdots A_{k'} \cdots A_{k''} \cdots A_{k_1} A^* A'_{k_2}? \cdots A'_1?$$ 
  with $A,A_1,\ldots,A_{k_1},A'_1,\ldots,A'_{k_2}\neq \emptyset$. (Remember that we assume $L(r)\neq \emptyset$ and the left cut border of $B_\text{pre} = A_1?\cdots A_{k_1}?$ is $0$. Furthermore, $A \not \subseteq A_{k''}$ implies $A \neq \emptyset$.)  
The following construction holds for both forms that $r$ can have.

  We now construct $(H,s,t)$.  The main idea is to have at most one
  edge with a label in $A_{k''}$ that is reachable from $s$ by a path of
  length $k''-1$. Then each path matching $r$ must route through it
  and we can do a similar proof as for $\nodespath(a^{k-1}b a^*)$ in
  Corollary~\ref{theo:languageoverview}(\ref{theo:languageoverview:b}).    

  More formally, fix an $x \in (A \setminus A_{k''})$. Fix three words
  $w_1$, $w_2$, and $w_3$ such that
  \begin{itemize}
  \item $w_1 \in L(A_1 \cdots A_{k'})$,
  \item $w_2 \in L(A_{k'+1} \cdots  A_{k''} \cdots
    A_{k_1})$, and
  \item $w_3 \in L(A'_{k_2} \cdots A'_1)$.\footnote{We use $w_3 \in
      L(A'_{k_2} \cdots A'_1)$ in case that $r$ ends with $A'_{k_2}
      \cdots A'_1$ but also if it ends with $A'_{k_2}? \cdots A'_1?$.}
  \end{itemize}
  Notice that such words indeed exist.  For the construction of $H$,
  we start with the graph $G'$ used in the proof of Theorem
  \ref{theo:kcolor} and make the following changes:
  \begin{itemize}
  \item We replace each $b$-edge in $G'$ with an $x$-path of length $k''$
    (using $k''-1$ new nodes for each replacement).
  \item We change the labels of the $a$-edges in $G'$ such that each
    $c_1$-$c_{k+1}$-path is labeled $w_1$. Notice that the label for each
    such edge is well-defined. Indeed, by
    Lemma~\ref{lem:twopaths}(\ref{lem:twopaths:b}) we have that 
    each $a$-path from $c_1$ to $c_{k+1}$ has length exactly $k'$.  
                If there would be an edge
    $e$ on an $a$-labeled $c_1$-$c_{k+1}$-path that is reachable from $c_1$ through
    $n_1$ edges and also through $n_2$ edges, with $n_1 \neq n_2$,
    then, since $c_{k+1}$ is reachable from $e$, it means that there would
    be paths of different lengths from $c_1$ to $r_1$.
    We relabel all other edges with $x$.
          \item We add a path labeled $w_2$ from $c_{k+1}$ to
    $r_1$. We refer to this path as the \emph{$w_2$-labeled path} in
    the remainder of the proof.
          \item We add a path labeled $w_3$     from $r_{k+1}$ to a new node $t_\text{new}$, 
            to which we will refer as the
    \emph{$w_3$-labeled path} in the remainder of the proof.
      \end{itemize} 
  The resulting graph $(H,c_1,t_\text{new})$ together with the expression $r \in
  \cR$ serves as input for $\nodesimpath(\cR)$. This concludes the reduction.

  We show that the reduction is correct. This can be proved analogously 
  to the proof of Theorem \ref{theorem:two-disjoint-is-hard}, that is,
  we show that $G'$ and $k'$ are in \knodecolordisjointpaths if and only if
  $(H,c_1,t_\text{new})$ and $r$ are in \nodesimpath$(\cR)$.

  If $G'$ and $k'$ are in \knodecolordisjointpaths with solution $p_a$ and
  $p_b$, then there exists a (unique) simple path from $c_1$ to $t_\text{new}$ in $H$ that contains the nodes
  $V(p_a) \cup V(p_b)$ and matches $r$. 

  Conversely, if $(H,c_1,t_\text{new})$ and $r$ are in $\nodesimpath(\cR)$, there
  exists a simple path $p$ from $c_1$ to $t_\text{new}$ in $H$ that matches $r$. We will
  prove the following:
  \begin{enumerate}[(i)]
  \item The prefix of $p$ of length $k'$ corresponds to a simple path
    from $c_1$  to $c_{k+1}$ in $G'_a$ from the proof of
    Theorem~\ref{theo:kcolor} in
    Appendix~\ref{app:twodisjointpaths}.\footnote{$G'_a$ is the graph
      obtained from $G'$ by deleting all $b$-edges and nodes that have no adjacent $a$-edges.} (That is, $p[0,k']$ is a path
    from $c_1$ to $c_{k+1}$-path in $G'_a$.) 
  \item The prefix of $p$ of length $k_1$ ends in $r_1$. Its prefix
    is labeled $w_1$ and its suffix is the $w_2$-labeled path.
  \item We show that $\lab(p)=w_1w_2 w' w_3$ with $w' \in L(A^*)$. 
     \end{enumerate}
   We prove (i). By definition of $r$, the edge $p[k''-1,k'']$ is
  labeled by some symbol in $A_{k''}$. Therefore, this symbol cannot
  be $x$. By construction of $H$, this edge is either an edge that was
  labeled $a$ in $G'$, an edge on the $w_2$-labeled path, or an edge
  on the $w_3$-labeled path (since all other edges are labeled $x$).
  
  Since the $w_3$-labeled path is not reachable with a path of length
  smaller than $k''$ and the $w_2$-labeled path starts in $c_{k+1}$
  and is therefore only reachable with a path of length at least $k'$,
  see Observation \ref{obs:replace-b}, the first $k'+1$ nodes must
  form an $a$-path. This implies that $p[0,k']$ is entirely in
  $G'_a$. 
 From
  Lemma~\ref{lem:twopaths}(\ref{lem:twopaths:b}), we know that each
  path in $G'_a$ of length $k'$ goes from $c_1$ to $c_{k+1}$ which implies (i).
Since all nodes (except $r_1$) that belong to the $w_2$-labeled path of length $k_1-k'$ have only one outgoing edge, we have that $p[0,k_1]$ ends in $r_1$ and must match $w_1w_2$. This shows (ii).

Since $p$ matches $r=A_1 \cdots A_{k_1} A^* A'_{k_2} \cdots A'_1$ or
$r=A_1 \cdots A_{k_1} A^* A'_{k_2}? \cdots A'_1?$, and since each word
in $A_1 \cdots A_{k_1}$ has length $k_1$, it follows that $\lab(p)=w_1
w_2 w'$ with $w' \in L(A^* A'_{k_2} \cdots \allowbreak A'_1) \cup L(A^* A'_{k_2}? \cdots A'_1?)$.

By construction of $H$, the $w_3$-labeled path is the unique path of
length $|w_3|$ leading to $t_\text{new}$. Therefore, each
$c_1$-$t_\text{new}$-path in $H$ must end with the $w_3$-labeled
path. Since $w_3 \in L(A'_{k_2} \cdots A'_1)$ and $|w_3|$ is the
length of every word in $L(A'_{k_2} \cdots A'_1)$, we have that
$\lab(p) = w_1 w_2 w' w_3$ where $w\in L(A^*)$. So we have (iii).  Let
$p'$ be the part of $p$ labeled $w'$.  We now show that $p'$ can only
consist of edges labeled $x$.  Since $p$ is a simple path, it must be
node-disjoint with its prefix $p[0,k']$. We showed in (i) that
$p[0,k']$ corresponds to a $c_1$-$c_{k+1}$-path in $G'_a$, so we know
from Lemma~\ref{lem:twopaths}(\ref{lem:twopaths:c}) and
(\ref{lem:twopaths:d}) that it uses all control-nodes and at least one
edge in each row. Therefore, it follows as in the proof of
Theorem~\ref{theorem:two-disjoint-is-hard} that $p'$ cannot use edges
that correspond to ones in $G'_a$.  Therefore, $p'$ only consists of
edges labeled $x$.  This shows that $G'$ and $k'$ are in
\knodecolordisjointpaths, because $p[0,k']$ corresponds to a path
$p_a$ and $p'$ to $p_b$, which are solutions to
\knodecolordisjointpaths.

Finally, we note that the construction can indeed be done in FPT since
the expression $r \in \cR$ can be determined in time $f(k)$ for a
computable function $f$, the graph from the proof of Theorem
\ref{theo:kcolor} was constructed in FPT, and all changes
we made to the graph are in time $ h(k) \cdot |G'|^c$, for a constant $c$ and a computable function $h$, which is FPT. Indeed, we only relabeled
all edges, replaced each edge at most once with $k''$ new edges and
added other paths of length at most $|r|\leq f(k)$. Since $|r|\leq f(k)$, we also have $k_r\leq f(k)$, so we have indeed an FPT reduction.
\end{proof}

\subsection{Dichotomy for Edge-Disjoint Paths}

Next we will prove the dichotomy on STEs for trails.

\edgeDichotomyThm*
\begin{proof}
We first prove part (a). 
Since $\cR$ is almost conflict free, there exists a constant $c$ such that each $r \in \cR$ has at most $c$ conflict labels. Let $r \in \cR$ an STE with left cut border $c_1$ and right cut border $c_2$. We will show how to decide whether there exists a path from $s$ to $t$ in $G$ matching $r$ in FPT. 

We use the reduction from Lemma~\ref{lemma:edgeToNodes} to convert this problem into at most $n$ instances of the corresponding problem \nodesimpath($r$). We now show how to decide \nodesimpath($r$) on an instance $(H_i, s'_i, t')$ in time $f(k_r)\cdot |H_i|^{O(1)}$.
We observe that each node in $H_i$ (except $s'$ and $t'$) corresponds to exactly one edge in $G$.

If $r$ is $c'$-bordered for a $c' \leq c$, then we can immediately use the methods of Theorem~\ref{theo:dichotomy}(a). If this is not the case, we know that $A, A_1, \ldots, A_{k_1}, A'_1, \ldots, A_{k_2}\neq \emptyset$ and it remains to consider $r = A_1 \cdots A_{k_1} A^* A'_{k_2} \cdots A'_1$, $r = A_1? \cdots A_{k_1}? A^* A'_{k_2} \cdots A'_1$, and $r = A_1 \cdots A_{k_1} A^* A'_{k_2}? \cdots A'_1?$.
We will show how to solve this problem for $$r = A_1 \cdots A_{k_1} A^* A'_{k_2} \cdots A'_1,$$ by adapting Algorithm~\ref{alg:block}. The other cases then follow as in Theorem~\ref{theo:dichotomy}(a).

We will first explain what we change in Algorithm~\ref{alg:block} and show its correctness afterwards.
\begin{itemize}
\item We also change $r_1$ (line \ref{alg:block5}) and $r_2$ (line \ref{alg:block13}) by relabeling the conflict labels. Assume $A_\ell$ is an conflict label. Then we define $A_\ell$ in $r_1$ as $A_\ell \setminus A \cup \{a' \mid a \in A_\ell \cap A\}$. We proceed analogously for conflict labels $A'_j$ in $r_2$. \tina{easier understandable if we do this generally for $A_j$ with $j \leq c_1$ or $A'_j$ with $j \leq c_2$?}
\item We enumerate all subsets of up to $c$ nodes $v_{u_1,a,u_2}$ with $a \in A$ in $H_i$. For each possible subset $S$, we generate a graph $H_S$ by changing the nodes $v_{u_1,\sigma,u_2} \in S$ to $v_{u_1,\sigma',u_2} \in S$ and relabel the outgoing edges with $\sigma'$ (where $\sigma'$ is a new symbol, i.e., we add at most $|A|$ different symbols in total). 
\end{itemize}
This completes the changes we do. It is obvious that these changes are possible in FPT, since enumerating all subsets of size at most $c$ is in $O(|H_i|^c)$. Since the original algorithm was in FPT, the adapted one is as well.
We now show the correctness.
We first show that it indeed suffices to consider subsets of up to $c$ nodes.
Let $p_\text{pref}$ be an arbitrary path matching $A_1 \cdots A_{c_1}$ and $p_\text{suff}$ be an arbitrary path matching $A'_{c_2}\cdots A'_{1}$ that is edge-disjoint from $p_\text{pref}$.
Since $r$ only has $c$ conflict labels, we know that there are at most $c$ edges that can be shared between $p_\text{pref}$ and $p_\text{suff}$ with an arbitrary path matching $A^*$. Since we constructed $H_i$ such that every edge in $G$ corresponds to exactly one node in $H_i$, see Lemma~\ref{lemma:edgeToNodes}, this means that simple paths matching $A_1 \cdots A_{c_1}$ and $A'_{c_2}\cdots A'_{1}$ can share at most $c$ nodes with an arbitrary path matching $A^*$.
So, in order to assure node disjointness between those paths, it indeed suffices to consider subsets of up to $c$ nodes and force the paths matching $A_1 \cdots A_{c_1}$ and $A'_{c_2}\cdots A'_{1}$ to only use those while the $A$-path may only choose other nodes. We enforce this by changing the labels.

If there exists a simple path from $s'_i$ to $t'$ in $H_i$ of length at most $2k$, where $k= k_1+k_2$, it will be found in line \ref{alg:block1}.
We now show that, if there exists a path from $s'_i$ to $t'$ in $H_i$ of length at least $2k+1$, then it will be found in the adapted algorithm between line \ref{alg:block4} and \ref{alg:block:lastfor}. Like in Lemma~\ref{lemma:0crit} we have that it suffices to consider paths in which the $k+1$ first nodes belong to $X \in \hat{P}^{k+1}_{sv,r_1}$. The proof is again analogous to Lemma~\ref{proofforfptalg}: The paths $P'$ and $R$ it cannot intersect in the first $c_1+1$ nodes of $P'$ since those nodes only have outgoing edges that have labels not in $A$. Since $R$ matches $A^*$, it cannot use them. And if $P'$ and $R$ intersect after the first $c_1+1$, the obtained simple path still matches (if we replace the new symbols with their usual ones) $r$, since we have that $A \subseteq A_{c_1+1}, \ldots, A_{k_1}$, due to definition of $c_1$, and the path is long enough because it still contains $Q$.
From line \ref{alg:block1} we know that the remaining path from $v$ to $t$ must have length at least $k_2+1$.
So we can prove again analogous to Lemma~\ref{proofforfptalg} that it suffices to consider paths in which the last $k_2+1$ nodes belong to $X' \in \hat{P}^{k_2+1}_{ut,r_2}$, for some $u$. So the adapted algorithm is indeed correct.

It remains to consider case (b). Notice that $\cR$ is not cuttable, as this would imply that it is almost conflict free. The proof follows the lines of Theorem~\ref{theo:dichotomy} part (b), i.e., we give an reduction from \kclique. Given an instance $(G,k)$ from \kclique, we find an $r \in \cR$ that has at least $2k_{new}$ conflict labels where $k_{new} = 5k^2+3k+1$ (this comes from Theorem~\ref{theo:edgeTwodisjoint}. Notice that we need so many conflict labels to ensure that they are on the right position). Let us assume that we have at least $k_{new}$ conflict labels in $A_1 \cdots A_{k''}$, where $k''=c_1$ is the left cut border of $r$. The case that we have at least $k_{new}$ conflict labels in $A'_{c_2} \cdots A'_{1}$ is symmetric.
Notice that we use $k''$ instead of $c_1$ to avoid confusion with the node $c_1$. Furthermore, $A_{c_1}$ has the same property as $A_{k''}$ in Theorem~\ref{theo:dichotomy} part (b).
Due to definition of cut border, we have that $A \not \subseteq A_{k''}$.

We will now basically use the graph from Theorem~\ref{theo:dichotomy} part (b) except that we label the $a$-path from $c_1$ to $r_1$ differently and split the nodes like in Theorem~\ref{theo:edgeTwodisjoint}. 
We will now explain the changes in detail. 

First we fix an $x \in (A \setminus A_{k''})$. Fix three words
  $w_1$, $w_2$, and $w_3$ such that
  \begin{itemize}
  \item $w_1 \in L(A_1 \cdots A_{k''})$, such that $w_1$ contains as many symbols in $A$ as possible
  \item $w_2 \in L(A_{k''+1} \cdots
    A_{k_1})$, and
  \item $w_3 \in L(A'_{k_2} \cdots A'_1)$.\footnote{We use $w_3 \in
      L(A'_{k_2} \cdots A'_1)$ in case that $r$ ends with $A'_{k_2}
      \cdots A'_1$ but also if it ends with $A'_{k_2}? \cdots A'_1?$.}
  \end{itemize}
We start with the graph $G'$ from Theorem~\ref{theo:kcolor}.
\begin{itemize}
\item In each gadget $G_{i,j}$, we split each $u_\ell$ in two nodes, that is $u_\ell^\text{in}$ and $u_\ell^\text{out}$. We call the two nodes that resulted from the same node a \emph{node-pair}. We redirect all incoming edges from $u_\ell$ to $u_\ell^\text{in}$, while all outgoing edges begin in $u_\ell^\text{out}$. We depict this in Figure~\ref{fig:edgeOneGadget}.
 Finally, we add an edge $u_\ell^\text{in} \to u_\ell^\text{out}$. We make exactly the same change to all $v_\ell$, $c_i$, $c_{i_1i_2}$, and $r_i$. 
\item We replace each $b$-edge by a $x$-path of length $k''$ and label the edge between $r_i^\text{in}$ and $r_i^\text{out}$ with $x$ for all $i$.
\item We will now adapt the graph so that each path from $c_1^\text{in}$ to $c_{k+1}^\text{out}$ has length at least $k''$ and exactly $k''$ if and only if it uses edges corresponding to an $a$-path from $c_1$ to $c_{k+1}$ in $G'$. Since we have at least $k_{new}$ symbols from $A$ in $w_1$ and $k_{new}\leq k'' = |w_1|$, we can indeed label each edge with an symbol from $A$.\footnote{It would suffice to label each edge between each $u_\ell^\text{in}\to u_\ell^\text{out}$ and each $v_\ell^\text{in}\to v_\ell^\text{out}$ with a label in $A$.} 
We choose the first $k_{new}-1$ such symbols from $w_1$ and place them
on edges that correspond to $a$-paths from $c_1^{in}$ to
$c_{k+1}^{in}$ in the right order. We label $c_{k+1}^{in}
\stackrel{A_{k''}}{\to} c_{k+1}^{out}$. To obtain a path matching
$w_1$, we then insert paths matching subwords of $w_1$ that contain no
symbol in $A$ between two such edges, or containing arbitrary symbols,
right before $c_{k+1}^{in}$. (Here we just need to make sure that all
paths have the same length.) \tina{ sounds awful, I know, but we somehow need to make sure that all paths have the same length.}
\item We add a path matching $w_2$ from $c_{k+1}^{out}$ to $r_1^\text{in}$, which we will call \emph{$w_2$-labeled path}, and a path matching $w_3$ from $r_{k+1}^\text{out}$ to a new node $t_\text{new}$, which we will call \emph{$w_3$-labeled path}. 
\end{itemize}
This completes the construction of our graph $H$. We can now prove the correctness analogously to Theorem~\ref{theo:dichotomy} part (b) and Theorem~\ref{theo:edgeTwodisjoint}. 
For the first direction, let $p_a$ be a simple $a$-path from $c_1$ to $c_{k+1}$ and $p_b$ a simple $b$-path from $r_1$ to $r_{k+1}$ in $G'$, such that $p_a$ has length $k'$ and is node-disjoint from $p_b$. By construction, we can use the same nodes (or node-pairs) as $p_a$ to obtain a trail $p_1$ from $c_1^\text{in}$ to $c_{k+1}^\text{out}$ in $H$ that matches $w_1$. And we can use the same nodes (or node-pairs) as $p_b$ to obtain a trail $p_2$ from $r_1^\text{in}$ to $r_{k+1}^\text{out}$ in $H$ that matches $x^*$, i.e., $A^*$. We can complete it to a trail from $c_1^\text{in}$ to $t_\text{new}$ that matches $r$ by adding the $w_2$-labeled and $w_3$-labeled path.
For the other direction, let $p$ be a trail from $c_1^\text{in}$ to $t_\text{new}$ in $H$ that matches $r$. We will prove the following:
  \begin{enumerate}[(i)]
  \item The prefix of $p$ of length $k'$ corresponds to a simple path
    from $c_1$  to $c_{k+1}$ in $G'_a$ from the proof of
    Theorem~\ref{theo:kcolor} in
    Appendix~\ref{app:twodisjointpaths}.\footnote{$G'_a$ is the graph
      obtained from $G'$ by deleting all $b$-edges and nodes that have no adjacent $a$-edges.} (That is, $p[0,k']$ is a path
    from $c_1$ to $c_{k+1}$-path in $G'_a$.) 
  \item The prefix of $p$ of length $k_1$ ends in $r_1$. Its prefix
    is labeled $w_1$ and its suffix is the $w_2$-labeled path.
 \item We show that $\lab(p)=w_1w_2 w' w_3$ with $w' \in L(A^*)$. 
  \end{enumerate}
  We prove (i). By definition of $r$, the edge $p[k''-1,k'']$ is
  labeled by some symbol in $A_{k''}$. Therefore, this symbol cannot
  be $x$. By construction of $H$, this edge is either an edge that was
  labeled $a$ in $G'$, an edge on the $w_2$-labeled path, or an edge
  on the $w_3$-labeled path (since all other edges are labeled $x$).
  
  Since the $w_3$-labeled path is not reachable with a path of length
  at most $k''$ and the $w_2$-labeled path starts in $c_{k+1}^\text{out}$
  and is therefore only reachable with a path of length at least $k''$
  (due to construction), the first $k''+1$ nodes must
  form an $a$-path. This implies that $p[0,k'']$ is entirely in
  $G'_a$. 
 From
  Lemma~\ref{lem:twopaths}(\ref{lem:twopaths:b}), we know that each
  path in $G'_a$ of length $k'$ goes from $c_1$ to $c_{k+1}$ which implies (i).
Since all nodes that belong to the $w_2$-labeled path of length $k_1-k''$ have only one outgoing edge, we have that $p[0,k_1]$ ends in $r_1^\text{in}$ and must match $w_1w_2$. This shows (ii).

Since $p$ matches $r=A_1 \cdots A_{k_1} A^* A'_{k_2} \cdots A'_1$ (the case $r=A_1 \cdots A_{k_1} A^* A'_{k_2}? \cdots A'_1?$ is analogous) and each word in $A_1 \cdots A_{k_1}$ has length $k_1$, it follows that $lab(p)=w_1 w_2 w'$ with $w' \in L(A^* A'_{k_2} \cdots A'_1)$. 

  By construction of $H$, the $w_3$-labeled path is the unique path of
  length $|w_3|$ leading to $t_\text{new}$. Therefore, each $c_1^\text{in}$-$t_\text{new}$-path in
  $H$ must end with the $w_3$-labeled path. Since $w_3 \in L(A'_{k_2}
  \cdots A'_1)$ and $|w_3|$ is the length of every word in $L(A'_{k_2}
  \cdots A'_1)$, we have that $\lab(p) = w_1 w_2 w' w_3$ where $w\in
  L(A^*)$. So we have (iii).
  Let $p'$ be the part of $p$ labeled $w'$.
  We now show that $p'$ can only consist of edges labeled $x$.

    Since $p$ is a simple path, it must be node-disjoint with its prefix 
    $p[0,k'']$. We showed in (i) that $p[0,k'']$ corresponds to a $c_1$-$c_{k+1}$-path 
    in $G'_a$, so we know from Lemma~\ref{lem:twopaths}(\ref{lem:twopaths:c}) and (\ref{lem:twopaths:d}) that it uses all control nodes and at least one edge in each row. Therefore, it follows
    as in the proof of
  Theorem~\ref{theo:edgeTwodisjoint} that $p'$ cannot use edges that correspond to ones in $G'_a$.  Therefore, $p'$ only consists of edges labeled $x$. 
  This shows that $G'$ and $k'$ are in \knodecolordisjointpaths, 
  because $p[0,k'']$ corresponds to a path $p_a$ and $p'$ to $p_b$, which are solutions to \knodecolordisjointpaths.

Finally, we note that the construction can indeed be done in FPT since
the expression $r \in \cR$ can be determined in time $f(k)$ for a
computable function $f$, the graph from the proof of Theorem
\ref{theo:kcolor} was constructed in FPT, and all changes
we made to the graph are in time $ h(k) \cdot |G'|^c$, for a constant $c$ and a computable function $h$, which is FPT. Indeed, we only relabeled
all edges, replaced each edge at most once with $c_2$ new edges, split each node at most once in two new ones, and added other paths of length at most $|r|\leq f(k)$. Since $|r|\leq f(k)$, we also have $k_r\leq f(k)$, so we have indeed an FPT reduction.
\end{proof}

\section{Proofs for Section~\ref{sec:enumfpt}}\label{app:enumfpt}

\fptdelayi*
\begin{proof}
  In the proof of Theorem~\ref{theorem:BaganPolyDelay} we adapted
  Yen's algorithm to work with simple instead of shortest paths. We
  already showed that the problem \nodesimpath$_{\geq k}$ is in
  FPT. Furthermore, the FPT algorithm can be adjusted to also return a
  matching path.  This is because, due to definition of
  $P^{k+1}_{sv}$, the nodes in $\hat{P}^{k+1}_{sv}$ form a path from
  $s$ to $v$ of length $k$. Given those nodes, we can easily built
  such a path.  (In fact, the construction of $\hat{P}^{k+1}_{sv}$
  allows to order the elements in the sets in $\hat{P}^{k+1}_{sv}$ so,
  that they directly correspond to such a path, see \cite[Lemma
  5.2]{fomin}.)  If Algorithm~\ref{fptalg} returns `yes', then there
  exists a $v$ and a $X \in \hat{P}^{k+1}_{sv}$ such that there exists
  a path from $v$ to $t$ in the graph without $X \setminus \{v\}$. As
  explained before we can construct a path from $s$ to $v$ that uses
  only nodes in $X$. We concatenate this path with a simple path from
  $v$ to $t$ that does not use nodes in $X$ except $v$ to obtain a
  simple path from $s$ to $t$ that has length at least $k$.
  Therefore, we can use this algorithm as a subroutine of YenSimple to
  obtain an FPT delay algorithm. For proving the correctness of this
  approach, we need to note that we can also deal with derivatives of
  the language, i.e., \nodesimpath$_{\geq i}$ with $i \leq k$, which is
  needed in line~\ref{alg:yen:12} of the YenSimple algorithm. However,
  in this case, we can simply solve     \nodesimpath$_{\geq j}$ with $j = \max\{k-i,0\}$
      with the
  same technique as \nodesimpath$_{\geq k}$.
\end{proof}

\fptdelayii*
\begin{proof}
We adapt the proof of Theorem~\ref{theo:fptdelay:1} for this case.
Therefore, we first show that the Algorithm~\ref{b?fptalg} can indeed output simple paths.
If there exists a path matching \nodesimpath$(a^ka^*)$, Algorithm~\ref{fptalg} finds it and can output it (see Theorem~\ref{theo:fptdelay:1}). Otherwise the algorithm only returns `yes' if there exists a path $p_c \in S$, a set $X \in \hat{P}^{k+1}_{su_0}$ and a path from $u_c$ to $t$ that are all node-disjoint except for $u_0$ and $u_c$. As explained before we can built a simple path $p_1$ from $s$ to $u_0$ from the nodes in $X$ and find a simple path $p_2$ from $u_c$ to $t$ that does not use nodes in $X$.
We then return $p_1 p_c p_2$ which is indeed a simple path matching $a^k w? a^*$.

For the derivatives, we need to deal with \nodesimpath$(a^iw?a^*)$, \nodesimpath$(w'a^*)$ for suffixes of $w$, and \nodesimpath$(a^*)$.
The first, \nodesimpath$(a^iw?a^*)$ with $i \leq k$ can be solved with the
same technique as \nodesimpath$(a^kw?a^*)$. For \nodesimpath$(w'a^*)$, where
$w'$ is a suffix of $w$, we enumerate all possible paths $p'$ that
match $w'$, which are again at most $O(n^c)$ many since $|w'| \leq |w|
=c$. Then, we search for a simple path $p_2$ from the last node of $p'$ to $t$
that does not use other nodes from $p'$. If we found one (which must be the case if we return `yes'), we return $p'p_2$.
We can use the same technique to deal with \nodesimpath$(a^*)$, that is, we choose $w' = \varepsilon$.
\end{proof}

\begin{lemma}\label{lem:ste-derivatives}
  Let $w \in \Sigma^*$ and $r$ be a $c$-bordered STE of size $n$. Then $w^{-1}L(r)$ is a union
  of STEs $r_1, \ldots, r_m$ such that 
  \begin{itemize}
  \item $m \leq n$ and
  \item each $r_i$ is $c'$-bordered for some $c' \leq c$.
  \end{itemize}
\end{lemma}
\begin{proof}
  Let $r = B_1 \cdots B_n$ be a $c$-bordered STE such that each $B_i$ is either of
  the form $A$, $A?$ or $A^*$ as in Definition~\ref{def:ste}.
  Let $w\in \Sigma^*$ and $J = \{j \mid w \in L(B_1 \cdots
  B_j)\}$. Then $w^{-1}L(r) = L(\Sigma_{j \in J} B_{j+1} \cdots
  B_n)$. Since $|J| \leq n$ and each expression $B_{j+1} \cdots
  B_n$ is $c'$-bordered for some $c' \leq c$ by definition, the result follows.
\end{proof}

\fptdelayiii*
\begin{proof}
We adapt the proof of Theorem~\ref{theo:fptdelay:1} for this case. So let $\cR$ be a cuttable class of STEs and let $r \in \cR$ with left cut border $c_1$ and right cut border $c_2$.
First we enumerate all possible paths $p_{c_1}$ that can match $A_1 \cdots A_{c_1}$ and $p_{c_2}$ that match $A'_{c_2} \cdots A'_{1}$. (These paths can also be empty if $c_1 = 0$ or $c_2 = 0$.)
We know now that the remaining regular expression, that is $r' = B'_\text{pre} A^* B'_\text{suff}$, is 0-bordered. So we now search for a path matching $r'$ from the last node of $p_1$ to the first of $p_2$ in the graph without the other nodes of $p_1$ and $p_2$.

Now we do case distinctions depending on its actual form.
If $A = \emptyset$, we can use the algorithm of Bagan et
al.~\cite[Theorem 6]{bagan}. In the appendix of the arXiv-version
\cite{baganOld} they give the corresponding algorithm which is based
on color coding \cite{AlonYZ-jacm95} and dynamic programming. This algorithm
can easily be adapted to generate a witness $p$ in case the decision
algorithm returns `yes'. Indeed, while running the dynamic algorithm,
we can always store the witnessing information for each newly computed
entry, from which the witnessing path can be computed at the end of
the algorithm.
We then return $p_{c_1} p\, p_{c_2}$.
Otherwise we know that $A \neq \emptyset$. If $r' = A_{1}? \cdots A_{k_1}? A^* A'_{k_2}? \cdots A'_{1}?$, its language is downward closed, so we can find a simple path $p$ matching $r'$, see Proposition~\ref{prop:prefixclosed}. We then return $p_{c_1}p\, p_{c_2}$.

For $r' = A_{c_1+1} \cdots A_{k_1} A^* A'_{k_2} \cdots A'_{c_2+1}$, we show that Algorithm~\ref{alg:block} can indeed output simple paths. As explained before, we can obtain a simple path matching $r'$ of length at most $2k$ in line \ref{alg:block1} by using the algorithm of Bagan et
al.~\cite[Theorem 6]{bagan}. 
So, if Algorithm~\ref{alg:block} returned `yes', but we did not find a path in line \ref{alg:block1}, there exists nodes $u, v \in V$, sets $X\in \hat{P}^{k+1}_{sv,r_1}$ and $X'\in \hat{P}^{k_2+1}_{ut,r_2}$, and a simple path $p$ from $v$ to $u$ that matches $A^*$ and is node disjoint from $X$ and $X'$ except for $v$ and $u$. 
Due to definition of $P^{k+1}_{sv,r_1}$,
the nodes in $X\in \hat{P}^{k+1}_{sv,r_1}$ form a path from $s$ to $v$ that matches $r_1$ and has length $k$. Since we built the sets in $\hat{P}^{k+1}_{sv,r_1}$ analogous to \cite[Lemma 5.2]{fomin}, see Lemma~\ref{lemma:SizeBoundWithRE}, we can easily built such a path. (In fact, the construction of $\hat{P}^{k+1}_{sv,r_1}$ allows to order the elements in the sets so, that they directly correspond to such a path).
So we can construct a path $p_1$ from $s$ to $v$ that uses only nodes in $X$ and matches $r_1$ and a path $p_2$ from $u$ to $t$ that uses only nodes in $X'$ and matches $r_2$. 
So Algorithm~\ref{alg:block} can indeed output the path $p_1 p\, p_2$. 
We then obtain our solution for $r$ by adding $p_{c_1}$ and $p_{c_2}$, if necessary.  

Since, we can use an easier variant of Algorithm~\ref{alg:block} if $r' = A_{c_1+1} \cdots A_{k_1} A^* A'_{k_2}? \cdots A'_{1}?$ or $r' = A_{1}? \cdots A_{k_1}? A^* A'_{k_2} \cdots A'_{c_2+1}$, we can also output paths in this cases.

It remains to show that we can handle all possible derivatives of
STEs. According to Lemma~\ref{lem:ste-derivatives}, we only need to
consider $k_1+k_2+1$ STEs that are $c'$-bordered for some $c' \leq
c$. 
Indeed, according to Lemma~\ref{lem:ste-derivatives}, each possible derivative is a union of at most $k_1+k_2+1$ STEs. Since each such STE is $c'$-bordered for some $c' \leq c$, we can solve \nodesimpath\ for each of them in FPT. And, since it are at most $k_1+k_2+1$ many, solving it for each of them is still in FPT. We can obviously use the same case distinctions as above and return a path if \nodesimpath\ answers `yes'.  \end{proof}

 We will now show that we can even output paths in FPT delay with radix order. Therefore, we will use Yen's algorithm, so we need algorithms that output shortest and lexicographically smallest paths. We will show how to achieve this, also for the derivatives needed in line~\ref{alg:yen:12} of Yen's algorithm.
 \begin{lemma}\enumeratesimpaths$_{\geq k}$ is in FPT delay with radix order. \label{lemma:a-geqk-shortest}
 \end{lemma}
 \begin{proof}
 We have seen in Theorem \ref{theo:fptdelay:1} that Algorithm~\ref{fptalg} can indeed output paths. 
 We now explain how to change this algorithm to output a shortest simple path from $s$ to $t$ that has length at least $k$.
 Therefore, we first observe that the proof of Lemma~\ref{proofforfptalg} works with the shortest simple path longer than $k$. So this lemma also implies that there exists a shortest simple path longer than $k$ such that its first $k+1$ nodes belong to a $X \in \hat{P}^{k+1}_{sv}$. 
 So we can find a shortest simple path longer than $k$ by running the algorithm for each $v \in V$ and each $X \in \hat{P}^{k+1}_{sv}$ and searching a shortest $v$-$t$-path in line~\ref{fptalg:4}. We always store the actual shortest simple path that is still longer than $k$. 
  Therefore, we can use this
   algorithm as a subroutine of Yen's algorithm to obtain an FPT delay
   algorithm that enumerates the paths from shortest to longest. 
   Notice that this algorithm can also deal with derivatives of the language, i.e.,  \nodesimpath$_{\geq j}$ with $j = \max\{k-i,0\}$, which is needed in line \ref{alg:yen:12} of Yen's algorithm.
  \end{proof}

 \begin{lemma}
 For each constant $c$ and each word $w$ with length $|w|=c$, the problem $\enumnodespaths(a^kw?a^*)$ is in FPT delay with radix order. \label{lemma:w?-shortest}
 \end{lemma}
 \begin{proof}
 We have seen in Theorem \ref{theo:fptdelay:2} that Algorithm~\ref{b?fptalg} can indeed output paths. 
 We show how to change this algorithm 
 to output a shortest and lexicographically smallest simple path that matches $a^kw?a^*$. .    If there exists a shortest simple path matching $a^ka^*$, we will find it in line \ref{alg:b?fptalg:1} and store it in $p_1$. We explain in Lemma~\ref{lemma:a-geqk-shortest} how this can be done. In addition to this, we find the overall shortest and lexicographically smallest path matching $a^kw a^*$ (if it exists) by enumerating all paths $p_c \in S$ and again using the algorithm obtained in Lemma~\ref{lemma:a-geqk-shortest} to find the shortest paths that completes $p_c$ to a path matching $a^kw a^*$. 
  We store the actual shortest and lexicographically smallest path that
 matches $a^kw a^*$ in $p_2$.\footnote{We used in the correctness
   proof of this algorithm that there is no path matching $a^ka^*$ if
   we search for a simple path matching $a^kw a^*$. Indeed the much
   weaker version suffices: there exists no simple path matching
   $a^ka^*$ that is strictly shorter than the shortest path matching $a^kw a^*$. This suffices since, if the path matching $a^kw a^*$ is not simple and the repeated node is neither in the subpath matching $w$ nor in the subpath matching the last $k$ nodes, then the resulting simple path matches $a^ka^*$ and is really shorter.}
  If we only have one path in $p_1$ or $p_2$, we will output this one. 
 Otherwise, we compare the length of $p_1$ with the length of $p_2$. If $|p_1| = |p_2|$, we output the lexicographically smaller one, else the shorter one. This completes the changes. Notice that this algorithm also works if $a$ is not the lexicographically smallest symbol. Since we output the shorter path, the path must indeed be simple. Therefore, the algorithm still works correctly.
 
 We again need to show how to handle the derivatives. 
 This algorithm can also be used for \nodesimpath{($a^{k-i}w ?a^*$)} with $i \leq k$. For \nodesimpath{($w' a^*$)}, where $w'$ is a suffix of $w$, we can enumerate all possible simple paths $p'$ that match $w'$, which are again at most $O(n^c)$ many, since $|w'|\leq |w| = c$. Then, we search for a shortest path $p_2$ from the last node of $p'$ to $t$ matching $a^*$ that does not use other nodes from $p'$. We then choose the shortest and lexicographically smallest path $p'p_2$ and output it.
 \end{proof}

 \begin{lemma} 
 Let $\cR$ be a cuttable class of STEs. Then \enumnodespaths$(\cR)$ is in FPT delay with radix order if $|A_i|\leq c$ and $|A'_i|\leq c$ for a constant $c$.\label{lemma:0crit-shortest}
 \end{lemma}
 \begin{proof}
 We have seen in Theorem \ref{theo:fptdelay:3} that we can enumerate paths for all kinds of STEs. Our goal is now to show that it is also possible to enumerate them in radix order. 
 Let $r \in \cR$ with left cut border $c_1$ and right cut border $c_2$.
 We enumerate all possible paths $p_{c_1}$ that can match $A_1 \cdots A_{c_1}$ and $p_{c_2}$ that match $A'_{c_2} \cdots A'_{1}$. (These paths can also be empty if $c_1 = 0$ or $c_2 = 0$.)
 We know now that the remaining regular expression, that is $r' = B'_\text{pre} A^* B'_\text{suff}$, is 0-bordered. So we now search for a shortest and lexicographically smallest simple path matching $r'$ from the last node of $p_1$ to the first of $p_2$ in the graph without the other nodes of $p_1$ and $p_2$.
 
 We do a case distinctions depending on the form of $r'$.
 If $A = \emptyset$, we can use the algorithm of Bagan et
 al.~\cite[Theorem 6]{bagan}. 
 \tina{please write shorter:}
 This algorithm can also be used to obtain a smallest path in radix order. Since the witnessing path is computed from the back, we cannot directly determine the lexicographically smallest path, but we can compute the smallest path in reversed radix order, that is: $w_1 \leq^r w_2$ in reversed radix order if $|w_1| < |w_2|$ or $|w_1| = |w_2|$ and $w_1^R$ is lexicographically smaller than $w_2^R$, where $^R$ denotes the symbol-wise reverse order of a word. 
 So we use the algorithm on the graph with reversed edges and with the ``reversed'' regular expression, i.e., instead of $r' = A_{c_1+1} \cdots A_{k_1}  A'_{k_2} \cdots A_{c_2+1}$, we use $r'^R = A'_{c_2+1} \cdots A'_{k_2}  A'_{k_1} \cdots A_{c_1+1}$. In the end, we reverse the path to obtain the smallest path in radix order in the original graph that matches $r$ and is simple.
 We enumerate all possible color codings and compare for each color coding the smallest path in radix order to obtain the overall smallest path in radix order, which we call $p$.
 We then return $p_{c_1} p\, p_{c_2}$.
 
 Otherwise we know that $A \neq \emptyset$. If $r' = A_{1}? \cdots A_{k_1}? A^* A'_{k_2}? \cdots A'_{1}?$, its language $L(r')$ is downward closed, so we can find a simple path $p$ matching $r'$ that is smallest in radix order, see Proposition~\ref{prop:prefixclosed}. We then return $p_{c_1}p\, p_{c_2}$.
 
 We now explain how to change Algorithm~\ref{alg:block} for $r = A_1 \cdots A_{k_1} A^* A'_{k_2}\cdots A'_1$.
 As explained above, we can use the algorithm of Bagan et al.~\cite[Theorem 6]{bagan} to output smallest paths in radix order in line \ref{alg:block1} if $\nodesimpath(r')$ has a solution of length at most $2k$, where $k= k_1 -c_1 + k_2 -c_2$.
 
 The lines \ref{alg:block4} to \ref{alg:block12} highly resemble Algorithm~\ref{fptalg} and can therefore be changed to output shortest paths, see Lemma~\ref{lemma:a-geqk-shortest}. The same holds for lines \ref{alg:block12} to \ref{alg:block22}. So shortest paths are no problem.

 We can guarantee that the results are in radix order if there exists a constant $c$ with $|A_i|\leq c$ and $|A'_i|\leq c$ for all $i$. 
 This is because we can then enumerate in line \ref{alg:block5} all up to $c^{k_1}$ words $w_1 \in L(r_1)$ and compute $\hat{P}^{k+1}_{sv,w_1} \subseteq^{k+1}_\text{rep} P^{k+1}_{sv,w_1}$ for each such word. This way we can ensure that we really considered each lexicographically smallest word. We proceed analogous in line \ref{alg:block13} for all $w_2 \in L(r_2).$ 
 
 We will now show that we can indeed obtain the smallest path in radix order this way. We will therefore use a variant of Lemma~\ref{proofforfptalg}. 
 Let $p=(s,a_0,v_1)(v_1,a_1,v_2) \cdots (v_{\ell-1},a_{\ell-1},t)$ be a smallest path from $s$ to $t$ in radix order that is simple and matches $r'$. Then $w_1 = a_0 \cdots a_{k-1} \in L(r_1)$. 
 From line~\ref{alg:block1} we know that the solution now must have length longer than $2k$. We now define  
 \begin{multline*}
 P=(s,a_0,v_1)(v_1,a_1,v_2) \cdots  (v_{k-1},a_{k-1},v_k), \\
 R =  (v_{k+1},a_{k+1},v_{k+2}) \cdots  (v_{\ell-k-2},a_{\ell-k-2},v_{\ell-k-1}), \text{ and } \\ Q=(v_{\ell-k},a_{\ell-k},v_{\ell-k+1}) \cdots (v_{\ell-1},a_{\ell-1},t). 
 \end{multline*}
As usual we write $$p = P\cdot (v_k,a_k,v_{k+1}) \cdot R \cdot(v_{\ell-k-1},a_{\ell-k-1},v_{\ell-k})\cdot Q$$ 
If $k = \ell -k -1$, i.e., $\ell = 2k+1$, we write $p = P\cdot (v_k,a_k,v_{k+1}) \cdot Q$ instead with $R = \varepsilon$.

 Since $|V(Q)| =k+1$ and $V(P) \cap V(Q) = \emptyset$, we can find a simple path $P'$ from $s$ to $v_k$ that matches $w_1$ and consists of nodes in 
 $\hat{P}^{k+1}_{sv_k,w_1}$. Furthermore, $V(P') \cap V(Q) = \emptyset$. If $P'$ and $R$ intersect, the resulting path would contradict our choice of $p$ as smallest path from $s$ to $t$ that is simple and matches $r'$, since the resulting path is shorter. If $P'$ and $R$ do not intersect, the path $p' = P' (v_{k},a_{k},v_{k+1}) Q (v_{\ell-k-1},a_{\ell-k-1},v_{\ell-k}) R$ is still a smallest path from $s$ to $t$ in radix order that is simple and matches $r'$ and its first $k+1$ nodes indeed belong to a set of nodes in $\hat{P}^{k+1}_{sv_k,w_1}$.

 It remains to show that its last $k_2+1$ nodes belong to $\hat{P}^{k_2+1}_{v_{\ell-k_2}t,w_2}$, where $w_2 = a_{\ell-k_2} \cdots a_{\ell-1}$. We assume that our prefix $P'$ is fix, i.e., 
 let $p' = P'(v_{k},a_{k},v_{k+1}) Q (v_{\ell-k-1},a_{\ell-k-1},v_{\ell-k}) R$ be a smallest path in radix order from $s$ to $t$ that matches $r'$ and is simple.
 If the length of $Q (v_{\ell-k-1},a_{\ell-k-1},v_{\ell-k}) R$ is smaller than $2k_2$, i.e., $\ell \leq k + 2k_2 +1$,
 we have $$p' = P' \cdot Q_2 \cdot (v_{\ell-k_2-1},a_{\ell-k_2-1},v_{\ell-k_2}) \cdot  P_2$$ with  \begin{multline*}
 	Q_2=(v_k,a_{k},v_{k+1}) \cdots  (v_{\ell-k_2-2},a_{\ell-k_2-2},v_{\ell-k_2-1}) \text{ and }\\ P_2=(v_{\ell-k_2},a_{\ell-k_2},v_{\ell-k_2+1}) \cdots (v_{\ell-1},a_{\ell-1},t).	
 \end{multline*}
 Since $V(Q_2) \leq k_2+1$ and $V(Q_2) \cap V(P_2) = \emptyset$, we find a set $X' \in \hat{P}^{k_2+1}_{v_{\ell-k_2}t,w_2}$ that corresponds to a simple path from $v_{\ell-k_2}$ to $t$ that matches $w_2$ and does not intersect with $Q_2$. 
 
 Let us now assume that $\ell > k + 2k_2 +1$. In this case we have 
 $$p' = P' \cdot Q_2 \cdot (v_{k+k_2},a_{k+k_2},v_{k+k_2+1})
  \cdot R_2 \cdot (v_{\ell-k_2-1},a_{\ell-k_2-1},v_{\ell-k_2}) \cdot  P_2$$ with 
 \begin{multline*}
 	Q_2=(v_k,a_{k},v_{k+1}) \cdots  (v_{k+k_2-1},a_{k+k_2-1},v_{k+k_2}),\\
 	 R_2 =  (v_{k+k_2+1},a_{k+k_2+1},v_{k+k_2+2}) \cdots  (v_{\ell-k_2-2},a_{\ell-k_2-2},v_{\ell-k_2-1}) \text{ and }\\	 
 	  P_2=(v_{\ell-k_2},a_{\ell-k_2},v_{\ell-k_2+1}) \cdots (v_{\ell-1},a_{\ell-1},t).	
 \end{multline*}
 Since $V(Q_2) = k_2+1$ and $V(Q_2) \cap V(P_2) = \emptyset$, we find a set $X' \in \hat{P}^{k_2+1}_{v_{\ell-k_2}t,w_2}$ that corresponds to a simple path $P'_2$ from $v_{\ell-k_2}$ to $t$ that matches $w_2$ and does not intersect with $Q_2$. 
 If $P'_2$ and $R_2$ intersect, the resulting path would contradict our choice of $p'$ as smallest path from $s$ to $t$ that is simple and matches $r'$, since the resulting path is shorter. 
 If $P'_2$ and $R_2$ do not intersect, the path $p' = P' P' \cdot Q_2 \cdot (v_{k+k_2},a_{k+k_2},v_{k+k_2+1})
 \cdot R_2 \cdot (v_{\ell-k_2-1},a_{\ell-k_2-1},v_{\ell-k_2}) \cdot  P'_2$ is still a smallest path from $s$ to $t$ that is simple and matches $r'$ and its last $k_2+1$ nodes indeed belong to a set of nodes in $\hat{P}^{k_2+1}_{v_{\ell-k}t,w_2}$.
 
 In line \ref{alg:blockShortest2} we can use the algorithm of Ackermann and Shallit \cite{Ackerman-TCS09} to find a shortest and lexicographically smallest $v$-$u$-path matching $A^*$, see Theorem~\ref{theo:ackermanAndShallit}.
 
 Since, we can use an easier variant of Algorithm~\ref{alg:block} if $r' = A_{c_1+1} \cdots A_{k_1} A^* A'_{k_2}? \cdots A'_{1}?$ or $r' = A_{1}? \cdots A_{k_1}? A^* A'_{k_2} \cdots A'_{c_2+1}$, we can also output paths in radix order these cases. 
 Since for each expression $r \in \cR$, $w^{-1}L(r)$ is still a STE, we can also use this obtained algorithm as subroutine in Yen's algorithm in line~\ref{alg:yen:12}.

 For expressions of the form $A_1 \cdots A_{k_1} A^* A'_{k_2}? \cdots A'_1?$ and $A_1? \cdots A_{k_1}? A^* A'_{k_2} \cdots A'_1$ it follows analogous. 
 Since the derivative of an STE is again an STE with at most the same cut border (see Lemma~\ref{lem:ste-derivatives}), we can use the same methods for them.
 \end{proof}

\fptdelayedge*
\begin{proof}
	We will use the reduction from Lemma~\ref{lemma:edgeToNodes} to reduce $\edgespath(r)$ to $\nodesimpath(r)$. Since the paths have a one-to-one correspondence, we can use the path in the output of $\nodesimpath(r)$ to obtain a path for $\edgespath(r)$. So, if we can show that $\nodesimpath(r)$ in FPT delay for the graphs that can be constructed in the reduction, we can use this correspondence to also output the corresponding paths. Notice that we have already covered many cases of $r$ in Theorem~\ref{theo:fptdelay:3}, but since $\cR$ does not need to be cuttable, we still have to show that Algorithm~\ref{alg:block} can also output paths after the changes in Theorem~\ref{theo:edgedichotomy} part (a).
	Since we only relabeled some edges and labels in the regular expression, the adapted Algorithm~\ref{alg:block} can still output a witness. In this output, we reverse our label-changes again to obtain a path in the original graph without label-changes and therefore corresponding to a path in $\edgespath(r)$. 
	Since derivatives of $r$ have at most the same number of conflict labels and the graphs used in Yen's algorithm still have the property that each node corresponds to at most one edge, we can again use the same strategy in Yen's algorithm in line~\ref{alg:yen:12} and therefore solve $\enumspaths(r)$.
\end{proof}

 \begin{lemma} 
	Let $\cR$ be a class of STEs that is almost conflict-free. Then, \enumedgespaths$(\cR)$ is in FPT delay with radix order if $|A_i|\leq c$ and $|A'_i|\leq c$ for a constant $c$.
\end{lemma}
\begin{proof} 
	We will use the reduction from Lemma~\ref{lemma:edgeToNodes} to reduce $\edgespath(r)$ to $\nodesimpath(r)$. Since the paths have a one-to-one correspondence, we can use the path in the output of $\nodesimpath(r)$ to obtain a path for $\edgespath(r)$. In this correspondence, we even preserve the labels of the path.
	 So, if we can show that $\nodesimpath(r)$ is in FPT delay with radix order for the graphs that can be constructed in the reduction, we can use this correspondence to also output the corresponding paths in $\edgespath(r)$ with radix order. Notice that we have already covered many cases of the form of $r$ in Theorem~\ref{theo:fptdelay:3} and can in this cases find the smallest in radix order, see Lemma~\ref{lemma:0crit-shortest}.
	 But since $\cR$ does not need to be cuttable, we still have to show that Algorithm~\ref{alg:block} can also output smallest paths in radix order after the changes we did in Theorem~\ref{theo:edgedichotomy} part (a).
	 Recall that our changes were to enumerate sets of at most $c$ nodes, relabele the outgoing edges of those nodes and changed some labels in the regular expression.
	 Obviously, the adapted Algorithm~\ref{alg:block} can still
         output a witness and, since we obtain our original instance
         by removing all single quotes, i.e., we relabel $a'$ to our original symbol $a$, we can even compare each candidate and only output the smallest in radix order.
	 Since the reduction in Lemma~\ref{lemma:edgeToNodes} preserves the labels, we therefore have a smallest path in radix order in $\edgespath(r)$.
	  
	Since derivatives of $r$ have at most the same number of conflict labels and the graphs used in Yen's algorithm still have the property that each node corresponds to at most one edge, we can again use the same strategy in Yen's algorithm in line~\ref{alg:yen:12} and therefore solve $\enumspaths(r)$ by solving $\enumnodespaths(r)$ in FPT delay with radix order and output the corresponding paths.
\end{proof}

\end{document}